\pdfoutput=1
\RequirePackage{ifpdf}
\ifpdf 
\documentclass[pdftex]{sigma}
\else
\documentclass{sigma}
\fi

\numberwithin{equation}{section}

\newtheorem{Theorem}{Theorem}[section]

\newtheorem{Lemma}[Theorem]{Lemma}
\newtheorem{Proposition}[Theorem]{Proposition}
 { \theoremstyle{definition}
\newtheorem{Definition}[Theorem]{Definition}

\newtheorem{Remark}[Theorem]{Remark} }

\usepackage{mathrsfs}
\usepackage{subcaption}
\usepackage{bbm}
\usepackage{bm}
\usepackage{mathtools}
\usepackage{diagbox}
\usepackage{multirow}
\usepackage{rotating}

\newcommand{\Pc}{\mathcal{P}}

\newcommand{\C}{{\mathbb C}}

\newcommand{\diag}{\operatorname{diag}}
\newcommand{\re}{\operatorname{Re}}
\newcommand{\im}{\operatorname{Im}}

\newcommand{\Tr}{\operatorname{Tr}}

\newcommand{\Var}{\operatorname{Var}}
\newcommand{\Cov}{\operatorname{Cov}}

\newcommand{\lnorm}{\left\|}
\newcommand{\rnorm}{\right\|}

\newcommand{\Hp}{\mathcal{H}}
\newcommand{\Lw}{\mathscr{L}}
\newcommand{\Sw}{\mathcal{S}}
\newcommand{\rv}{\mathbf{r}}
\newcommand{\lv}{\mathbf{l}}
\newcommand{\Rv}{\mathbf{R}}

\newcommand{\blambda}{\boldsymbol{\lambda}}
\newcommand{\abs}[1]{\left| #1 \right|}
\newcommand{\av}{\mathbf{a}}
\newcommand{\bv}{\mathbf{b}}
\newcommand{\cv}{\mathbf{c}}
\newcommand{\GL}{\mathrm{GL}}

\def\norm#1{\lnorm {#1} \rnorm}

\newcommand{\RN}[1]{%
	\textup{\uppercase\expandafter{\romannumeral#1}}%
}

\usepackage[labelfont=bf,labelsep=period,font=small]{caption}

\begin{document}

\allowdisplaybreaks

\newcommand{\arXivNumber}{2504.12120}

\renewcommand{\PaperNumber}{058}

\FirstPageHeading

\ShortArticleName{Logarithmic Spectral Distribution of a Non-Hermitian $\beta$-Ensemble}

\ArticleName{Logarithmic Spectral Distribution \\ of a Non-Hermitian $\boldsymbol{\beta}$-Ensemble}

\Author{Gernot AKEMANN~$^{\rm a}$, Francesco MEZZADRI~$^{\rm b}$, Patricia P\"A{\SS}LER~$^{\rm a}$ and Henry TAYLOR~$^{\rm b}$}

\AuthorNameForHeading{G.~Akemann, F.~Mezzadri, P.~P\"a{\ss}ler and H.~Taylor}

\Address{$^{\rm a)}$~Faculty of Physics, Bielefeld University, P.O. Box 100131, 33501 Bielefeld, Germany}
\EmailD{\mail{akemann@physik.uni-bielefeld.de}, \mail{patricia@physik.uni-bielefeld.de}}

\Address{$^{\rm b)}$~School of Mathematics, University of Bristol,\\
\hphantom{$^{\rm b)}$}~Fry Building, Woodland Road, Bristol, BS8 1UG, UK}
\EmailD{\mail{f.mezzadri@bristol.ac.uk}, \mail{ht17630@bristol.ac.uk}}

\ArticleDates{Received May 07, 2025, in final form May 18, 2026; Published online June 18, 2026}

\Abstract{We introduce a non-Hermitian $\beta$-ensemble and determine its spectral density in the limit of large $\beta$ and large matrix size $n$. The ensemble is given by a general tridiagonal complex random matrix of normal and chi-distributed random variables, extending previous work of Mezzadri and Taylor (2025). The joint distribution of eigenvalues contains a~Vandermonde determinant to the power $\beta$ and a residual coupling to the eigenvectors. A~tool in the computation of the limiting spectral density is a single characteristic polynomial for centred tridiagonal Jacobi matrices, for which we explicitly determine the coefficients in terms of its matrix elements. In the low temperature limit $\beta\gg1$, our ensemble reduces to such a centred matrix with vanishing diagonal. A general theorem from free probability based on the variance of the coefficients of the characteristic polynomial allows us to obtain the spectral density when additionally taking the large-$n$ limit. It is rotationally invariant on a compact disc, given by the logarithm of the radius plus a constant. The same density is obtained when starting form a tridiagonal complex symmetric ensemble, which thus plays a~special role. Extensive numerical simulations confirm our analytical results and put this and the previously studied ensemble in the context of the pseudospectrum. The numerical study of the local nearest-neighbour spacing distribution shows agreement between the tridiagonal ensemble and two-dimensional Poisson statistics (independently of~$\beta$), whereas we observe a $\beta$-dependence for the previously introduced ensemble.}


\Classification{60B20; 33C45; 33E12}

\section{Introduction and main results}

One of the most surprising features of the spectra of random matrices is that statistically they behave like a two-dimensional gas of charged particles interacting with a logarithmic potential. Indeed, the joint density of the $n$
eigenvalues has the form of a Boltzmann factor ${\rm e}^{\beta H}$, where~$\beta$ is the inverse temperature and
the Hamiltonian $H$ depends on the ensemble of matrices. In the classical Hermitian and circular ensembles, the inverse temperature is restricted to specific integer values,
namely $\beta=1,2,4$, which correspond to the dimension of the associated division algebra with real coefficients~\cite{Dys62}. In a highly influential paper, Dumitriu and Edelman~\cite{DE02A} developed a model of tridiagonal real symmetric matrices where $\beta$ can be any positive real number.

Because the logarithmic potential is the Coulomb interaction in two dimensions, it is natural to ask if the same construction is possible for complex non-Hermitian random matrices.
Here, the classical ensembles introduced by Ginibre \cite{Ginibre} essentially only provide such a static Coulomb gas at a single value $\beta=2$, even if for real and quaternion valued matrix elements the joint density of eigenvalues is more complicated and given by a Pfaffian determinant \cite{PetersBook}. The local repulsion then corresponds to $\beta=2$.

In recent years, the two-dimensional Coulomb gas at a general value of $\beta$ has been intensely studied, both in the mathematics and physics literature. Based on numerical simulations it has been conjectured, that in the large-$n$ limit, the point process undergoes a transition from a~constant, uniform density to the triangular Abrikosov lattice, at around $\beta\approx 140$ \cite{AB81,CC83}, see~\cite{CSA} for recent work.
At $\beta=2$, a realisation in terms of complex non-Hermitian Gaussian random matrices exists, whose spectrum becomes a determinantal point process with a known kernel, as already pointed out by Ginibre.
Few rigorous mathematical results exist beyond the integrable point $\beta=2$. For general $\beta>0$, the normalising constant or partition function of the point process is unknown, that would be a complex analogue of the Selberg integral. The empirical spectral distribution is given by the circular law, but the crystallisation to a lattice has not been shown yet or is even debated. Apart from $\beta=2$ and the vicinity of $\beta\approx0$ of distance $1/N$, where the eigenvalue correlations become Poisson \cite{Lambert}, no rigorous results on the local limiting statistic exist. Around $\beta=2$ non-rigorous, perturbative results \cite{Can-et-al} show the onset of strong oscillations for $\beta>2$ close to the edge of the spectrum, which have been observed to become dominant close to the transition in the entire radial spectrum, cf.~\cite{CSA} for the most recent work. An analogous construction to Dumitriu and Edelman~\cite{DE02A} of a non-Hermitian $\beta$-ensemble would be highly desirable.

Recently, two of the authors~\cite{mezzadri-taylor} introduced
an ensemble of particular complex non-symmetric tridiagonal matrices where the inverse temperature $\beta$ is a continuous parameter.
They constructed the joint density of complex eigenvalues, by establishing a one-to-one correspondence between the random variables on the diagonal and lower diagonal of a particular tridiagonal matrix on the one hand, and the complex eigenvalues and the first row of the diagonalising similarity transformation on the other hand. The latter set of variables completely determines the diagonalisation.
The construction allowed them to construct the Jacobian and joint eigenvalue distribution explicitly. This ensemble has several attractive features. Its normalisation is known explicitly, corresponding to a Selberg type integral in the complex plane. The joint eigenvalue distribution contains the logarithmic Coulomb repulsion in two dimensions at arbitrary inverse temperature $\beta>0$. However, it also contains an additional integral where the eigenvalues and part of the eigenvectors are still coupled.
So far, this has made it difficult to compute the eigenvalue correlation functions analytically, beginning with the global spectral density, even in the limit of large matrix size $n\to\infty$.
In particular, the ensemble of \cite{mezzadri-taylor} does not reduce to the complex Ginibre ensemble when setting $\beta=2$. Furthermore, as we will see below, the tridiagonal ensemble introduced by \cite{mezzadri-taylor} does not share the feature of a condensation onto a triangular lattice either, at least up to the values of $\beta\approx 5000$ that we tested numerically. Therefore, the remaining coupling between eigenvalues and eigenvectors seems to have a significant impact.
We note in passing that for $n=2$ a normal matrix realisation exists \cite{AMP} for a non-Hermitian $\beta$-ensemble, which is trivially also tridiagonal. Because of normality the eigenvectors decouple, and it agrees with the Ginibre ensembles for $\beta=2$. This ensemble was used to formulate a~surmise for the local nearest-neighbour spacing distribution. Although it is known not to be a~good approximation to the limiting spacing distribution in the Ginibre ensemble \cite{GHS}, it turns out to be a~good approximation close to the Poisson distribution at $\beta=0$, for $\beta=0$--$0.5$.
This contrasts with the Wigner surmise for Hermitian random matrices, which is an excellent approximation for $\beta=1,2,4$, but not close to $\beta=0$.

In this work, we introduce an ensemble of general complex tridiagonal matrices. We will show that it shares many of the attractive features and several open questions of the first ensemble~\cite{mezzadri-taylor} that we just described.
It turns out that the spectral density of the general ensemble
is much more stable with respect to an expansion in large matrix size and large $\beta$.
In that limit the global spectral density of the general and symmetric ensemble agree.
The previously introduced ensemble \cite{mezzadri-taylor} is less stable in this limit, and we will show that to leading order the spectral density becomes a Dirac delta at the origin. In the more general ensemble, the spectral density
is absolutely continuous and very simple, depending on the logarithm of the modulus of the eigenvalue plus a constant. This limiting law seems to be a novel feature for non-Hermitian random matrix ensembles. In the numerical study of the local nearest-neighbour (NN) spacing distribution, we observe that the general tridiagonal model shares the statistics with a two-dimensional (2D) Poisson point process of independent particles. Moreover, it turns out that the local statistics are independent of $\beta$. For the complex non-symmetric tridiagonal ensemble, we find a $\beta$-dependence of the NN distribution which deviates from 2D Poissonian statistics. We currently do not have an understanding of this difference.

The article is organised as follows. In the remainder of this section, we define the general complex tridiagonal ensemble and present its features, the joint eigenvalue distribution, its limiting spectral density and the main tool we use, the characteristic polynomial for centred tridiagonal matrices, resulting from a large-$\beta$ (low temperature) expansion.
In particular, we will also explain the difference with the non-symmetric ensemble
\cite{mezzadri-taylor} and the agreement with the complex symmetric tridiagonal ensemble, which has a simpler joint density, although it is similar to~\cite{mezzadri-taylor}. In Section~\ref{Sec:General Beta Ensemble}, we will introduce a map between arbitrary complex tridiagonal matrices and their spectral parameter, compute its Jacobian and the joint probability density of the eigenvalues.
In Section~\ref{Sec:low temperature expansion}, we study the large-$\beta$ expansion, which leads to centred tridiagonal matrices, with vanishing diagonal. In Section~\ref{Sec:characteristic polynomial}, we explicitly determine the coefficients of the
characteristic polynomial for such ensembles.
In Section~\ref{Sec:equilibrium density}, we determine the equilibrium density for all ensembles using tools from free probability. In Section~\ref{Sec:Num}, we present numerical simulations to verify our results. All the ensembles are also compared using the concept of pseudospectrum and deviations from normality. We present a brief numerical study of the local NN statistics.
In Appendix~\ref{app-n=2}, we illustrate our results with matrix size $n=2$, and in Appendix~\ref{app-symm}, the spectral decomposition of the symmetric ensemble is derived.

\subsection{General complex tridiagonal ensemble}
\label{Subsec:complex jacobi}

We first introduce a general tridiagonal, complex non-Hermitian random matrix ensemble of dimension $n$. It takes the form
\begin{equation}
\label{eq-gen_trid}
T= \begin{pmatrix}
a_n & c_{n-1} & & & \\
{b}_{n-1} & a_{n-1} & c_{n-2} & & \\
 & \ddots & \ddots & \ddots&\\
 & & {b}_2 & a_2& c_1\\
 & & & {b}_1 & a_1
\end{pmatrix},
\end{equation}
with $a_j,{b}_j,c_j\in\mathbb{C}$, ${b}_j,c_j \neq 0$, and all other matrix elements vanishing.
In particular, we will choose the following distributions for the matrix elements, assumed to be all independent:
\begin{subequations}
 \label{mat_el_T}
\begin{gather}
 \re(a_j) \sim \mathcal{N}(0,1), \qquad \im(a_j) \sim \mathcal{N}(0,1), \qquad j=1,\dots,n, \\
 |b_j|, |c_j| \sim \chi_{\beta j/2}, \qquad
 \arg (b_j),\arg(c_j) \sim \mathcal{U}(0,2\pi), \qquad j=1,\dots,n-1.
\end{gather}
\end{subequations}
Here, $\mathcal{N}(0,1)$ is the standard centred normal distribution with unit variance, $\mathcal{U}(0,2\pi)$ denotes the uniform distribution, and $\chi_k$ is the distribution with density
\begin{equation}
\label{chik}
f_k(x)=\frac{1}{2^{\frac{k}{2}-1}\Gamma\bigl(\frac{k}{2}\bigr)} x^{k-1}\exp\bigl[-x^2/2\bigr], \qquad x>0.
\end{equation}
The matrix $T$ contains $2n+4(n-1)=6n-4$ real degrees of freedom. In the following, we will only consider matrices which are diagonalisable and have pairwise distinct non-vanishing eigenvalues.

We are interested in the distribution of complex eigenvalues $\lambda_j\in\mathbb{C}$, $j=1,\dots,n$, of matrix~$T$, and thus in the diagonalisation
\begin{equation}
\label{diagonal}
T=R\Lambda R^{-1}
\end{equation}
with $\Lambda= \diag(\lambda_1,\dots,\lambda_n)$ and $R \in \mbox{GL}(n,\C) $. The columns of $R$ contain the right eigenvectors, and the rows of $R^{-1}$ contain the left eigenvectors. These are in general
distinct.\footnote{Notice however, that for symmetric matrices any right eigenvector is also a left eigenvector~\cite{cullum-willoghby}.}
The diagonalising matrix $R$ is unique up to permutations of the eigenvalues and (right) multiplication by an invertible complex diagonal matrix, cf. Section~\ref{Sec:General Beta Ensemble}.
Below, we will also discuss two particular choices of ensembles, first, the one introduced in \cite{mezzadri-taylor} with $c_j=1$, $j=1,\dots,n-1$, which we denote by \smash{$\widetilde{T}$}, see \eqref{eq-T_tilde_def} below. The second choice is a symmetric tridiagonal complex matrix denoted by $S$, with $b_j=c_j$, $j=1,\dots,n-1$.

In \cite{mezzadri-taylor}, it was observed that the similarity transformation $\widetilde{T} = C TC^{-1}$ brings $T$ to the form%
\begin{equation}
\label{eq-T_tilde_def}
 \widetilde{T} = \begin{pmatrix}
a_n & 1 & & & \\
\widetilde{b}_{n-1} & a_{n-1} & 1 & & \\
 & \ddots & \ddots & \ddots & \\
 & & \widetilde{b}_2 & a_2 & 1 \\
 & & & \widetilde{b}_1 & a_1
\end{pmatrix}=C TC^{-1},
\end{equation}
where $\widetilde{b}_j = c_j {b}_j$, $j=1,\dots,n-1$, and
\begin{equation}
\label{eq-Cdef}
C = \begin{pmatrix} 1 & & & \\
 & c_{n-1} & & \\
 & & \ddots &\\
 & & & c_{n-1}\dotsm c_1
 \end{pmatrix}.
\end{equation}
The ensemble of matrices \smash{$\widetilde{T}$} was chosen as a starting point in \cite{mezzadri-taylor}. Here, we will start from the general ensemble~$T$ in \eqref{eq-gen_trid}, and later compare with an ensemble of symmetric matrices~$S$, which is defined by choosing $b_j=c_j$ for $j=1,\dots,n-1$ in \eqref{eq-gen_trid}.
Given the similarity transformation~\eqref{eq-T_tilde_def} the question immediately follows: Why are the ensembles $T$ and \smash{$\widetilde{T}$} not equivalent?
The reason is the particular choice of distribution of matrix elements. In both cases, the random variables are chosen as independent, with the $a_j$ complex normal, and
in~\cite{mezzadri-taylor}, \smash{$\widetilde{b}_j$} with uniform phase and chi-squared distribution $\chi_{j\beta}$, $j=1,\dots,n-1$ for its modulus \smash{$|\widetilde{b}_j|$}.
However,~choosing the independent distribution $\chi_{j\beta/2}$ for the moduli $|b_j|$ and $|c_j|$ in $T$, the~resulting distribution~$g_j(y)$ for the matrix elements \smash{$\widetilde{b}_j = c_j {b}_j$}, given by the product of matrix elements of matrix $T$ resulting from the similarity transformation~\eqref{eq-T_tilde_def},
\begin{equation*}
g_j(y)=\frac{1}{2^{\frac{j\beta}{4}-2}\Gamma\bigl(\frac{j\beta}{4}\bigr)}\ y^{j\beta-2}\exp\bigl[-y^4\bigr], \qquad y>0,
\end{equation*}
is {\it not} chi-square \eqref{chik}. Thus, the choice in~\cite{mezzadri-taylor} for \smash{$\widetilde{b}_j$} and the choice above in \eqref{mat_el_T} for $T$ made here are different, and constitute distinct ensembles.
Notice that the phase of \smash{$\widetilde{b}_j$} remains uniformly distributed.

\subsection{Joint distribution and limiting spectral density}
\label{Subsec:jpdf-density}

We present our main results: the joint probability distribution function (jpdf) and the limiting spectral density.
In~\cite{mezzadri-taylor}, it was shown, before choosing a specific distribution of matrix elements, how to recursively construct an invertible map between the matrix elements $a_j$ and \smash{$\widetilde{b}_j$} and the complex eigenvalues $\Lambda$ and the diagonalising matrix $R$ in~\eqref{diagonal}. This step allowed to construct the Jacobian explicitly, and then to derive the joint distribution of eigenvalues. In Section~\ref{Sec:General Beta Ensemble}, we will generalise this construction to the ensemble $T$, which is non-trivial because of the increased number of degrees of freedom.

\begin{Theorem}\label{jpdf_theorem}
The jpdf of the eigenvalues of the random matrix $T$ \eqref{eq-gen_trid} is given by
\begin{equation}
\label{eq:EigDens1}
 P_T(\Lambda) = \frac{1}{Z_\beta^T} \prod_{j=1}^n
 \exp\left(-\frac{|\lambda_j|^2}{2}\right)\prod_{1\le j < k \le n}|\lambda_k -
 \lambda_j|^\beta f_T(\Lambda),
\end{equation}
where $\beta\in \mathbb{R}_+$, $\Lambda=\diag(\lambda_1,\dots,\lambda_n)$, and $Z_\beta^T$ is the normalization constant
\begin{equation*}
Z_{\beta}^{T}=\pi^{3n-2}2^{\frac{n(\beta n-\beta+4)}{4}}\prod_{k=1}^{n-1}\Gamma\left(\frac{\beta k}{4}\right)^2.
\end{equation*}
The function $f_T(\Lambda)$ that contains a coupling between eigenvalues and eigenvectors reads
\begin{equation}
\label{factor}
f_T(\Lambda) =\int_{\mathbb{C}^{(n-1)}}\int_{\mathbb{C}^{n-1}\backslash \{0\}}
\exp\bigl(-g_{T}(\Lambda,\mathbf{r},\mathbf{R}_1)\bigr)\prod_{k=1}^{n}|r_k|^{\frac{\beta}{2}-2}\prod_{i=2}^n |R_i|^{-2}\mathrm{d}^2R_i\mathrm{d}^2r_{i-1},
\end{equation}
where $\mathbf{r} =(r_1,\dots,r_{n}) \in \mathbb{C}^{n}$, $1 = r_1 +\dots+ r_{n}$ denote the first element of each right eigenvector of $T$ $($the first row of the diagonalising matrix $R$ in \eqref{diagonal}$)$,
and $\mathbf{R}_1=(R_1,\dots,R_n)^{\rm t}$ denoting the first right eigenvector of $T$
$($the first column vector of~$R)$.
The exponent $g_T(\Lambda,\mathbf{r})$, which is positive, reads
\begin{equation}
\label{gTdef}
g_T(\Lambda,\mathbf{r},\mathbf{R}_1) = \frac{1}{2}\Tr\bigl(TT^\dagger\bigr)
 - \frac{1}{2}\sum_{j =1}^n |\lambda_j|^2.
\end{equation}
It explicitly depends on the respective random variables $a_j$, $j=1,\dots,n$, $b_k$ and $c_k$, $k=\dots,n-1$. Furthermore, it is rotationally invariant,
$g_{T}(\Lambda,\mathbf{r},\mathbf{R}_1)=g_{T}(\Lambda {\rm e}^{{\rm i}\phi},\mathbf{r},\mathbf{R}_1)$ for all $\phi\in[0,2\pi)$ and has the asymptotics
\begin{equation}
\label{g-asymptotic}
g_{T}(\Lambda,\mathbf{r},\mathbf{R}_1)=O\bigl(|\lambda_j|^2\bigr) \qquad \text{as}
\ \lambda_j\rightarrow\infty, \quad \text{for}\ j=1,2,\dots,n.
\end{equation}
\end{Theorem}
The proof for the ensemble \smash{$\widetilde{T}$} in \eqref{eq-T_tilde_def} was given in \cite{mezzadri-taylor} and takes the same form as Theorem~\ref{jpdf_theorem}. However, the integration in \eqref{factor} for $f_{\widetilde{T}}(\Lambda)$ becomes simpler, over only half as many eigenvector components $r_j$ (compare \eqref{factorS}). The asymptotics \eqref{g-asymptotic} for the corresponding function $g_T(\Lambda,\mathbf{r})$ becomes quartic instead.

It turns out that for the ensemble of complex symmetric tridiagonal matrices $S$, the same simplification appears as for $f_{\widetilde{T}}(\Lambda)$ in \cite{mezzadri-taylor}. Since this ensemble belongs to the same universality class as the general ensemble $T$, we briefly discuss the limiting spectral density here.
The corresponding function reads
\begin{equation}
\label{factorS}
f_S(\Lambda) =
\int_{\mathbb{C}^{n-1}}\exp\bigl[-g_S(\Lambda,\mathbf{r})\bigr]\prod_{j=1}^{n}|r_j|^{\frac{\beta}{2}-2} \prod_{k=1}^{n-1}\mathrm{d}^2r_k.
\end{equation}
The exponent $g_S(\Lambda,\mathbf{r})$, defined in the same way as in \eqref{gTdef} in terms of matrix $S$, shares the same asymptotic behaviour \eqref{g-asymptotic} and rotational invariance as the general case described above.

Let us come to the limiting spectral density for all ensembles considered here.
\begin{Theorem}\label{thm-density beta}
	The spectral density $\rho^{L}(z)$ in the limit $n\rightarrow \infty$ and low temperature regime $\beta \gg 1$ is given for both ensembles $L=T,S$ by
	\begin{equation}\label{eq-limiting density T,S}
		\rho^{L}(z)=\frac{{\rm e}^1}{2\pi}\bigl(\ln(2)-1-\ln\bigl(|z|^2\bigr)\bigr)
	\end{equation}
	for $|z|\in [0,r_0]$ and zero otherwise, with $r_0=\sqrt{2/{\rm e}}$. For the non-symmetric ensemble \smash{$\widetilde{T}$} from~\eqref{eq-T_tilde_def}, we obtain to leading order
	\begin{equation}\label{eq-limiting density T tilde}
		\rho^{\widetilde{T}}(z)=\delta^{(2)}(z),
	\end{equation}
which only has a singular component at the origin.
\end{Theorem}

An important step in the derivation is that in the low temperature limit all three ensembles~$T$,~$S$ and~\smash{$\widetilde{T}$} become centred to leading order, that is, the contribution from the diagonal in~\eqref{eq-gen_trid} and \eqref{eq-T_tilde_def} becomes subleading. Therefore, we explicitly compute the coefficients of characteristic polynomials of centred matrices $T$ with $a_j\equiv0$ for $j=1,\dots,n$, expressed in terms of the remaining random variables $b_j$, $c_j$ for $j=1,\dots,n-1$. Define the polynomial of degree $n$
\begin{align}\label{eq: Characteristic Polynomial}
 \mathcal{P}_n(z):=\det\bigl(z\mathbf{1}_n-T|_{a_j=0}\bigr)=z^n+\sum_{\ell=1}^{\lfloor \frac{n}{2}\rfloor}\kappa_\ell^{(n)}z^{n-2\ell}.
\end{align}
We also define $\mathcal{P}_0(z)=1$.
These polynomials have the following properties, where we define sums
$\sum_{l=j}^k$ to vanish for $k<j$.
\begin{Proposition}\label{prop-formula coefficients P}
The polynomials defined in \eqref{eq: Characteristic Polynomial} are either even or odd, depending on the parity of their degree $n$. They obey the following three-term recurrence relation:
\begin{equation}
\label{eq-3 term recurrence}
 \mathcal{P}_{n+1}(z)= z \mathcal{P}_{n}(z)-\widetilde{b}_{n}\mathcal{P}_{n-1}(z) \qquad \mbox{for}\ \ n=1,2,\dots,
\end{equation}
with $\widetilde{b}_j = c_j b_j$. The coefficients $\kappa_\ell^{(n)}$ defined in \eqref{eq: Characteristic Polynomial} are given by
	\begin{equation}\label{eq-formula coefficients P}
		\kappa_\ell^{(n)}=(-1)^\ell \sum_{\gamma_1=2\ell-1}^{n-1}\widetilde{b}_{\gamma_1} \sum_{\gamma_2=2\ell-3}^{\gamma_1-2}\widetilde{b}_{\gamma_2}
		\cdots \sum_{\gamma_\ell=1}^{\gamma_{\ell-1}-2}\widetilde{b}_{\gamma_\ell}
	\end{equation}
	for $\ell=1,2,\dots, \bigl\lfloor \frac{n}{2}\bigr\rfloor$. Note that $\kappa_\ell^{(n)}$ consists of the product of $\ell$ sums.	All sums can be extended to begin with $\gamma_j=1$ for all $j=1,\dots,\ell$ as the additional terms vanish.
\end{Proposition}
The exact form of the coefficients is known~\cite{PS06} when the polynomials \eqref{eq: Characteristic Polynomial} constitute orthogonal polynomials on $\mathbb{R}$, see also~\cite{LRR91} where a related form was presented without proof, cf.\ equation~\eqref{eq-coeff LRR91} below.
Let us emphasise, however, that in general the coefficients in \eqref{eq-3 term recurrence} are complex, $\widetilde{b}_j\in\mathbb{C}$. In contrast, for real orthogonal polynomials with respect to some density, the coefficients must be positive, as they represent the ratio of the squared norms of the monic orthogonal polynomials. This property is used in \cite{PS06}; thus, we present a detailed proof in Section~\ref{Sec:characteristic polynomial} without this assumption. A similar result on the moments of a specific tridiagonal matrix with vanishing diagonal and the off-diagonal elements being real and positive is given in \cite{AKvI00} equation~(5) with $p=1$, which can be related to the coefficients of the characteristic polynomial. However, we would like to emphasise that these proofs in terms of orthogonal polynomials rely massively on the fact that, due to Favard's theorem (cf.\ \cite[Theorem~4.4]{Chihara}), the existence of a unique probability measure for polynomials satisfying a three-term recurrence is guaranteed if and only if the elements of the \smash{$\widetilde{b}_j$} are real and positive for all $j=1,\dots ,n-1$.

The remaining steps of the proof of Theorem~\ref{thm-density beta} will use Proposition~\ref{prop-formula coefficients P} and results from free probability, to compute the limiting density from the variance of the explicit coefficients~\smash{$\kappa_\ell^{(n)}$}, in particular, \cite[Theorem~1.2]{Barbarino-Noferini}.

\section[Complex beta-ensemble]{Complex $\boldsymbol{\beta}$-ensemble}\label{Sec:General Beta Ensemble}
\subsection{Spectral decomposition of tridiagonal matrices}\label{sec:Decomp}
Let $\Lw(n)$ denote the set of matrices $T$, of the form given in equation~$\eqref{eq-gen_trid}$, with the eased constraints $a_j,b_j,c_j\in\mathbb{C}$ for all $j$, such that $\det(T)\neq0$ and the spectrum of $T$ is non-degenerate. Note that the probability distributions in equations~\eqref{mat_el_T}, define an absolutely continuous measure on the set of tridiagonal matrices~\eqref{eq-gen_trid}. Consequently, the exceptional set of matrices with degenerate spectra forms a set of measure zero and is therefore omitted.

We define
\begin{align}
 & \label{def:D(n)}D(n) = \{\diag(z_1,\dots,z_n) \mid z_j \in \C, \, z_j \neq 0, \,
 j=1,\dots,n \}, \\ \label{def:Lambda(n)}
 & \Lambda(n) = \{ \diag(\lambda_1,\dots,\lambda_n)\mid \lambda_j,\lambda_k \in \C,
 \, \lambda_j \neq \lambda_k \neq 0, \, 1 \le j < k \le n \},\\
 & R^T(n) = \bigl\{R \in \GL(n,\C) \mid T =
 R\Lambda R^{-1},\, T \in \Lw(n),\, \Lambda \in \Lambda(n) \bigr\}.\label{eq:R(n)def}
\end{align}
When $T \in \Lw(n)$, the spectral decomposition of $T$ is given by
\begin{equation}\label{eq:eigdecompgeneral}
T=R\Lambda R^{-1},
\end{equation}
where $\Lambda\in\Lambda(n)$ and $R\in R^{T}(n)$. This decomposition is unique up to a permutation of the eigenvalues and the right multiplication $R\mapsto RD$, where $D\in D(n)$.

As a consequence, the spectral decomposition in equation~\eqref{eq:eigdecompgeneral} defines a bijection
\begin{equation}\label{sing_valued_map_general}
 \mathcal{F}^{T}\colon \ \Lw(n)\rightarrow\mathcal{L}(n)\times\mathcal{R}^{T}(n),
\end{equation}
where
\begin{equation}\label{def:curlyL(n)}
\mathcal{R}^T(n)= R^T(n)/D(n), \qquad \mathcal{L}(n) = \Lambda(n)/\mathfrak{S}_n,
\end{equation}
and $\mathfrak{S}_n$ is the symmetric group of order $n$. We restrict the domain of $\mathcal{F}^{T}$ to $\Lw'(n)=\Lw(n)\setminus \widetilde{B}$, where
\begin{equation*}
 \widetilde{B}=\{T\in \Lw(n) \mid b_j=0\ \text{or}\ c_j=0\ \text{for some}\ j\in\{1,\dots,n\}\}.
\end{equation*}
Since the probability measure in equations~\eqref{mat_el_T}, is absolutely continuous, we note that the set of matrices not in $\Lw'(n)$ has measure zero and is therefore omitted from our analysis.
\begin{Remark}\label{rem:2}
If both $c_{n-j}=0$ and $b_{n-j}=0$ for any $j$ in the matrix~\eqref{eq-gen_trid}, we can perform a~direct sum decomposition. Specifically, we have
 \begin{equation}\label{eq:decomp2}
 T=T_j\oplus U_j,
 \end{equation}
 where $T_j$ is the principal sub-matrix obtained by keeping the first $j$ rows and columns of $T$. In~contrast, $U_j$ denotes the principal sub-matrix obtained by keeping the last $j$ rows and columns of $T$. Thus, we can carry out our analysis on these two smaller-dimensional matrices, namely~$T_j$ and~$U_j$, which are still tridiagonal.
\end{Remark}
Let
\begin{gather*}
 \widetilde{R}^{T}(n) = \{R \in \GL(n,\C) \mid T =
 R\Lambda R^{-1},\, T \in \Lw'(n),\, \Lambda \in \Lambda(n)\},
\end{gather*}
and
\begin{equation*}
 \widetilde{\mathcal{R}}^{T}(n)= \widetilde{R}^{T}(n)/D(n).
\end{equation*}
It was pointed out in Section~\ref{Subsec:complex jacobi} that the matrix ensembles consisting of $T$ \eqref{eq-gen_trid}, respectively~\smash{$\widetilde{T}$} with its definition in \cite{mezzadri-taylor}, are not equivalent in a probabilistic sense. Nevertheless, there exists the similarity transformation \eqref{eq-T_tilde_def}. We use this one-to-one mapping between the matrices (not in distributions) in this section. Therefore, we can inherit the proofs from \cite{mezzadri-taylor} with the replacement $b_j\rightarrow b_jc_j$. Another way to see, why this is the right mapping, is to look into the three-term recurrence relation of the characteristic polynomials. Again, one finds the same recurrence, which is used in the proofs \cite{mezzadri-taylor}, for the general tridiagonal matrix after the suggested replacement.
\begin{Lemma}\label{r_diff_from_zero}
 Let $\mathbf{r}^{\rm t}=(r_1,r_2,\dots,r_n)$ and $\mathbf{v}^{\rm t}=(v_1,\dots,v_n)$ be the first and last rows of matrix $R\in \widetilde{R}^{T}(n)$, respectively. We have
 \begin{equation*}
 r_j\neq0,\qquad v_j\neq0 \qquad \text{for all} \ j\in\{1,\dots,n\}.
 \end{equation*}
 Furthermore, the subset of $\widetilde{R}^{T}(n)$ such that
 \begin{equation*}
 r_1+r_2+\dots+r_n=1,
 \end{equation*}
 spans a set of representatives in $\widetilde{\mathcal{R}}^{T}(n)$.
\end{Lemma}
\begin{proof}
See proof of Lemma~2.1 in~\cite{mezzadri-taylor}.
\end{proof}

Next, we aim to show that the bijection in equation~\eqref{sing_valued_map_general} induces the injective map
\begin{equation*}
 \mathcal{G}^{T} \colon \ \Lw'(n)\rightarrow\mathcal{L}(n)\times\mathcal{H}_n\times \mathcal{V}_n,
\end{equation*}
where
\begin{equation}
 \label{hyperplane}
 \Hp_n = \bigl\{ \rv \in \C^n \mid r_1 + \dotsb + r_n =1, \, r_j \neq 0 \
 \text{for} \ j=1,\dots,n \bigr\},
\end{equation}
and
\[
\mathcal{V}_n=\bigl\{\mathbf{v}^{\rm t}=(v_1,\dots,v_n)^{\rm t}\in\mathbb{C}^n \mid v_j\neq0\ \text{for} \ j=1,\dots,n\bigr\}.
\]

From Lemma~\ref{r_diff_from_zero} and equation~\eqref{sing_valued_map_general}, it follows that $\mathcal{G}^{T}$ is single-valued. To prove the injectivity of $\mathcal{G}^T$, we let
\begin{equation}\label{def:Rgenmodel}
R=\begin{pmatrix}
 \rv_1^{\rm t}\\\vdots\\\rv^{\rm t}_n
\end{pmatrix}=(\Rv_1,\dots,\Rv_n)\qquad\text{and}\qquad R^{-1}=\begin{pmatrix}
 \mathbf{L}_1^{\rm t}\\\vdots\\\mathbf{L}_n^{\rm t}
\end{pmatrix}=(\lv_1,\dots,\lv_n),
\end{equation}
where $\mathbf{R}_j$ denotes the $j$-th column of $R$ and $\mathbf{r}^{\rm t}_j$ denotes its $j$-th row, with $\mathbf{r}_1^{\rm t}=\mathbf{r}^{\rm t}$ and $\mathbf{r}_n^{\rm t}=\mathbf{v}^{\rm t}$ from Lemma~\ref{r_diff_from_zero}. Similarly, let $\mathbf{L}_j^{\rm t}$ denote the $j$-th row of $R^{-1}$ and $\mathbf{l}_j$ denote the $j$-th column.\footnote{One should note that $\mathbf{R}_j$ and $\mathbf{L}_j^{\rm t}$ denote the right and left eigenvectors in $R$ and $R^{-1}$, respectively.}

Furthermore, as in the previous proofs, we will write
\begin{equation}\label{eq:biorthogonality}
\rv^{\rm t}_j \lv_{k} = r_{j1}l_{1k}+ \dotsb + r_{jn}l_{nk} = \delta_{jk} \qquad\text{for} \
j,k=1,\dots,n,
\end{equation}
where the notation on the left-hand side represents matrix multiplication between the row vector~$\rv_j^{\rm t}$ and the column vector $\lv_k$, rather than the usual scalar product in $\mathbb{C}^n$. Finally, let $(\mathbf{x})_1$ denote the first entry of the vector $\mathbf{x}$.
\begin{Theorem}\label{thm:roughbijection}
Let $\Lambda\in\Lambda(n)$, $\mathbf{r}^{\rm t}=(r_1,r_2,\dots,r_n)\in\mathcal{H}_n$ and $\mathbf{R}=(R_1,R_2,\dots,R_n)^{\rm t}\in\mathcal{V}_n$. Then, there exists a unique $R\in \widetilde{R}^{T}(n)$ whose first row is $\mathbf{r}^{\rm t}$ and first column is $\mathbf{R}$, and a unique $T\in\Lw'(n)$ such that $T=R\Lambda R^{-1}$, except for $\rv^{\rm t}$ belonging to an exceptional set $M\subset\mathcal{H}_n$ of Lebesgue measure zero.
\end{Theorem}
\begin{proof}
We start the proof by taking $R$ from the set of representatives of $\widetilde{\mathcal{R}}^{T}(n)$ defined in Lemma~\ref{r_diff_from_zero}.
Let us write the matrix equations
\begin{align*}
 TR=R\Lambda\qquad\text{and}\qquad
 R^{-1}T=\Lambda R^{-1},
\end{align*}
as a set of vector equations
\begin{align}\label{eq:genmodelrecrelation}
 &\mathbf{r}_j^{\rm t}\Lambda=b_{n-j+1}\mathbf{r}^{\rm t}_{j-1}+a_{n-j+1}\mathbf{r}_j^{\rm t}+c_{n-j}\mathbf{r}^{\rm t}_{j+1},\\
 &\Lambda\mathbf{l}_j=c_{n-j+1}\mathbf{l}_{j-1}+a_{n-j+1}\mathbf{l}_j+b_{n-j}\mathbf{l}_{j+1}\label{eq:genmodelrecrelation1}
\end{align}
for $j=1,2,\dots,n$, where we have adopted the notation in equation~\eqref{def:Rgenmodel}. Let\footnote{Lemma~\ref{r_diff_from_zero} implies that the entries of $\mathbf{R}_1$ are all different from zero. Thus, $\mathbf{R}_1\in\mathcal{V}_n$.} $\mathbf{R}=\mathbf{R}_1$ denote the first column of $R$ and $\mathbf{r}_1^{\rm t}=\mathbf{r}^{\rm t}$. We want to show that given $\Lambda$, $\mathbf{r}_1^{\rm t}$ and $\mathbf{R}_1$, with $\mathbf{r}_1^{\rm t}\mathbf{l}_1=1$, where
$\mathbf{l}_1=(1,\dots,1)^{\rm t}$,
we can reconstruct $T$, $R$ and $R^{-1}$ uniquely from the recurrence relations in equations~\eqref{eq:genmodelrecrelation} and~\eqref{eq:genmodelrecrelation1}, with boundary conditions $\mathbf{l}_0=\mathbf{r}_{0}^{\rm t}=\mathbf{0}$ and $b_0=c_0=0$.
For $j=1$, in equations~\eqref{eq:genmodelrecrelation} and~\eqref{eq:genmodelrecrelation1}, we thus have
\begin{align}
\label{cnot1recrelation}
 &\mathbf{r}^{\rm t}_1\Lambda=a_n\mathbf{r}^{\rm t}_1+c_{n-1}\mathbf{r}^{\rm t}_2,\\
 &\Lambda\mathbf{l}_1=a_n\mathbf{l}_1+b_{n-1}\mathbf{l}_2\label{cnot1recrelation1}.
\end{align}
Write
\begin{align}
 &a_n=\mathbf{r}_1^{\rm t}\Lambda\mathbf{l}_1,\label{eq:a_1}\\
 &c_{n-1}=\frac{\bigl(\mathbf{r}_1^{\rm t}\Lambda\bigr)_1-a_n\bigl(\mathbf{r}^{\rm t}_1\bigr)_1}{\bigl(\mathbf{r}_2^{\rm t}\bigr)_1}, \\
 &\label{eq:r_2}\mathbf{r}_2^{\rm t}=\frac{1}{c_{n-1}}\bigl(\mathbf{r}_1^{\rm t}\Lambda-a_n\mathbf{r}_1^{\rm t}\bigr),\\
 &b_{n-1}=\mathbf{r}_2^{\rm t}\Lambda\mathbf{l}_1,\\
 &\mathbf{l}_2=\frac{1}{b_{n-1}}(\Lambda\mathbf{l}_1-a_n\mathbf{l_1})\label{eq:l_2},
\end{align}
where in the above equations we have $\bigl(\rv_j^{\rm t}\bigr)_1=R_j$.
Notice that we successively express the left-hand sides in terms of known quantities from the previous equations, i.e., $a_n$ in terms of~$\Lambda$,~$\mathbf{r}_1^{\rm t}$ and $\mathbf{l}_1$, then $c_{n-1}$ in terms of these and $\bigl(\mathbf{r}_2^{\rm t}\bigr)_1=R_2$ etc.

In the case where $c_{n-1}=b_{n-1}=0$, the constructions in equations~\eqref{eq:r_2} and~\eqref{eq:l_2} fail. Therefore, we refer to Remark~\ref{rem:2} and decompose $T$ into two smaller tridiagonal matrices. Specifically, we express $T$ as the direct sum
\[
T=T_1\oplus U_{n-1},
\]
where $T_1$ and $U_1$ are defined as in Remark~\ref{rem:2}. This decomposition reduces the problem to two lower-dimensional matrices, which can be handled in two separate proofs. In general, if we encounter $j$ for which $c_{n-j}=b_{n-j}=0$, we use the decomposition in equation~\eqref{eq:decomp2} and carry out separate proofs for $T_j$ and $U_{n-j}$. The proofs will follow a similar structure to the one presented here, as $T_j$ and $U_{n-j}$ are tridiagonal matrices, albeit with lower degrees of freedom. Essentially, we restate the theorem for these lower-dimensional matrices and then perform the proofs for~$T_j$ and $U_{n-j}$ independently. Therefore, for the remainder of the proof, we assume that~$c_{n-j}$ and~$b_{n-j}$ are never simultaneously zero for any~$j$, thereby ensuring that our constructions avoid any singularities.
Additionally, since the probability distributions defined in equations~\eqref{mat_el_T} are absolutely continuous, for any fixed $\Lambda,$ the set of exceptional points $\rv_1^{\rm t}$ such that $b_{n-1}=0$ or $c_{n-1}=0$ has measure zero in $\Hp_n$ and can therefore be neglected.

The quantities in equations~\eqref{eq:a_1}--\eqref{eq:l_2} solve the recurrence relations in equations~$\eqref{cnot1recrelation}$ and $\eqref{cnot1recrelation1}$. We now check that the solutions are consistent with the equation $RR^{-1}=\mathbf{1}_n$. From the equations for $a_n$ and $\mathbf{r}_2^{\rm t}$, one is able to deduce that $\mathbf{r}_2^{\rm t}\mathbf{l}_1=0$. Then, from the equation for~$a_n$ and~$\mathbf{l}_2$, one can show that $\mathbf{r}_1^{\rm t}\mathbf{l}_2=0$. Furthermore,
\[
\mathbf{r}_2^{\rm t}\mathbf{l}_2=\frac{1}{b_{n-1}}\bigl(\mathbf{r}_2^{\rm t}\Lambda\mathbf{l}_1-a_n\mathbf{r}^{\rm t}_2\mathbf{l}_1\bigr)=1.
\]
Note that combining the orthogonality relations $\mathbf{r}_1^{\rm t}\mathbf{l}_2=0$ and $\mathbf{r}_2^{\rm t}\mathbf{l}_2=1$ gives
\[
\mathbf{r}_1^{\rm t}\Lambda\mathbf{l}_2=c_{n-1},
\]
as expected. Furthermore, by the orthogonality relations $\mathbf{r}_2^{\rm t}\mathbf{l}_2=1$ and $\mathbf{r}_2^{\rm t}\mathbf{l_1}=0$, we have
\[
\mathbf{r}_2^{\rm t}\Lambda\mathbf{l}_1=b_{n-1},
\]
as expected.

When $1<j<n$ take
$\mathbf{r}_{j-1}^{\rm t}$, $\mathbf{r}_j^{\rm t}$, $\mathbf{l}_{j-1}$, $\mathbf{l}_j$,
subject to the conditions
\begin{subequations}\label{eqs: Gen Cases}
\begin{align}
&\mathbf{r}_j^{\rm t}\mathbf{l}_j=\mathbf{r}^{\rm t}_{j-1}\mathbf{l}_{j-1}=1,\\
&\mathbf{r}_j^{\rm t}\mathbf{l}_{j-1}=\mathbf{r}^{\rm t}_{j-1}\mathbf{l}_j=0,\\
&c_{n-j+1}=\frac{\bigl(\mathbf{r}_{j-1}^{\rm t}\Lambda\bigr)_1-b_{n-j+2}\bigl(\mathbf{r}_{j-2}^{\rm t}\bigr)_1-a_{n-j+2}\bigl(\mathbf{r}_{j-1}^{\rm t}\bigr)_1}{\bigl(\mathbf{r}_{j}^{\rm t}\bigr)_1}\neq0,\\
&b_{n-j+1}=\mathbf{r}_j^{\rm t}\Lambda\mathbf{l}_{j-1}\neq0.
\end{align}
\end{subequations}
Define
\begin{subequations}\label{eqs: Gen cases 2}
\begin{align}
 &a_{n-j+1}=\mathbf{r}_j^{\rm t}\Lambda\mathbf{l}_j,\\
 &c_{n-j}=\frac{\bigl(\mathbf{r}_j^{\rm t}\Lambda\bigr)_1-b_{n-j+1}\bigl(\mathbf{r}_{j-1}^{\rm t}\bigr)_1-a_{n-j+1}\bigl(\mathbf{r}_j^{\rm t}\bigr)_1}{\bigl(\mathbf{r}_{j+1}^{\rm t}\bigl)_1},\\
 &\mathbf{r}_{j+1}^{\rm t}=\frac{1}{c_{n-j}}\bigl(\mathbf{r}_j^{\rm t}\Lambda-b_{n-j+1}\mathbf{r}^{\rm t}_{j-1}-a_{n-j+1}\mathbf{r}_j^{\rm t}\bigr),\\
 &b_{n-j}=\mathbf{r}_{j+1}^{\rm t}\Lambda\mathbf{l}_j,\\
 &\mathbf{l}_{j+1}=\frac{1}{b_{n-j}}(\Lambda\mathbf{l}_j-c_{n-j+1}\mathbf{l}_{j-1}-a_{n-j+1}\mathbf{l}_j).
\end{align}
\end{subequations}
For any fixed $\Lambda$, we have $b_{n-j}=0$ for $\rv_{j+1}^{\rm t}$ and $\lv_j$ belonging to a set of measure zero in $\C^n.$
In turn, $\rv_{j+1}^{\rm t}$ and $\lv_j$ are rational functions of the elements of $\rv_1^{\rm t}$; therefore, $b_{n-j}=0$ only for
$\rv_1^{\rm t}$ belonging to a set of measure zero in $\Hp_n.$ Note, a similar statement can be made when $c_{n-j}=0$. Combining equations~\eqref{eqs: Gen Cases} with equations~\eqref{eqs: Gen cases 2} and proceeding in a similar way too, when $j=1$, we prove that
\begin{align*}
&\mathbf{r}_{j+1}^{\rm t}\mathbf{l}_j=\frac{1}{c_{n-j}}\bigl(\mathbf{r}_j^{\rm t}\Lambda-b_{n-j+1}\mathbf{r}_{j-1}^{\rm t}-a_{n-j+1}\mathbf{r}_j^{\rm t}\bigr)\mathbf{l}_j=0,\\
&\mathbf{r}_j^{\rm t}\mathbf{l}_{j+1}=\frac{1}{b_{n-j}}\mathbf{r}_j^{\rm t}(\Lambda\mathbf{l}_{j}-c_{n-j+1}\mathbf{l}_{j-1}-a_{n-j+1}\mathbf{l}_j)=0,\\
&\mathbf{r}_{j+1}^{\rm t}\mathbf{l}_{j-1}=\frac{1}{c_{n-j}}\bigl(\mathbf{r}_j^{\rm t}\Lambda-b_{n-j+1}\mathbf{r}_{j-1}^{\rm t}-a_{n-j+1}\mathbf{r}_j^{\rm t}\bigr)\mathbf{l}_{j-1}=0,\\
&\mathbf{r}_{j-1}^{\rm t}\mathbf{l}_{j+1}=\frac{1}{b_{n-j}}\mathbf{r}_{j-1}^{\rm t}(\Lambda\mathbf{l}_j-c_{n-j+1}\mathbf{l}_{j-1}-a_{n-j+1}\mathbf{l}_j)=0,
\end{align*}
alongside
\begin{align*}
&\mathbf{r}_{j+1}^{\rm t}\mathbf{l}_{j+1}=\frac{1}{b_{n-j}}\mathbf{r}_{j+1}^{\rm t}(\Lambda\mathbf{l}_j-c_{n-j+1}\mathbf{l}_{j-1}-a_{n-j+1}\mathbf{l}_j)=1,\\
&\mathbf{r}_j^{\rm t}\Lambda\mathbf{l}_{j+1}=c_{n-j},\\
&\mathbf{r}_{j+1}^{\rm t}\Lambda\mathbf{l}_j=b_{n-j}.
\end{align*}
The orthogonality relations
\[
\mathbf{r}_{j+1}^{\rm t}\mathbf{l}_k=0,\qquad\mathbf{r}_k^{\rm t}\mathbf{l}_{j+1}=0
\]
for $k=1,2,\dots,j-2$, follow by induction and from the fact that by construction $\rv_j^{\rm t}\Lambda\lv_k=0$ if $|j-k|>1$.

The boundary conditions $b_0=c_0=0$ terminate the recurrence relations.
\end{proof}

\begin{Remark}
	 It is a well-known fact \cite{Chalker-Mehlig} that in non-Hermitian random matrix theory there exists a residual symmetry in rescaling the respective right and left eigenvectors simultaneously as follows:
		\begin{equation}
			r_j^{\rm t} \to d_j r_j^{\rm t},\qquad l_j\to d_j^{-1} l_j, \qquad 0\neq d_j\in\mathbb{C}, \qquad j=1,\dots,n.
			\label{invariance}
		\end{equation}
		This transformation does not change the biorthogonality of left and right eigenvectors in \eqref{eq:biorthogonality}, nor does it change the eigenvalues.
		It is this rescaling that fixes the $n-1$ additional parameters $c_1,\dots,c_{n-1}$ in ensemble $T$ given in \eqref{eq-gen_trid}, compared to ensemble $\tilde{T}$ given in \eqref{eq-T_tilde_def}. The matrix elements of the diagonal matrix $C$ in~\eqref{eq-Cdef} relating the two ensembles provides the rescaling in~\eqref{invariance}, with $d_1=1$, $d_2=c_{n-1}$, $d_3=c_{n-1}c_{n-2}$, $\dots$, $d_n= c_{n-1}\cdots c_1$. It is not difficult to see that this rescaling together with $b_j\to b_jc_j$ maps the recurrence~\eqref{cnot1recrelation}--\eqref{eq:l_2} above to the recurrence (2.20a)--(2.21d) in \cite{mezzadri-taylor}.
\end{Remark}

\subsection{Jacobian}
\label{eq:JacobianGenMatrix}
The generalisation of the volume element calculation to the spectral decomposition of $T\in\Lw(n)$ from the non-symmetric matrix calculations, as seen in~\cite{mezzadri-taylor}, is more computationally involved. We need to introduce another $n-1$ complex parameters, requiring a more careful adaptation of the previous theorems. To begin, we noted in Section~\ref{sec:Decomp} that the $6n-4$ real matrix parameters of $T$ are denoted by
\begin{equation*}
 \mathbf{a}:=(a_n,\dots,a_1),\qquad\mathbf{b}:=(b_{n-1},\dots,b_1),\qquad\mathbf{c}:=(c_{n-1},\dots,c_1),
\end{equation*}
and are in an injective correspondence with the $6n-4$ real spectral parameters
\begin{equation*}
 \blambda:=(\lambda_1,\dots,\lambda_n),\qquad\rv:=(r_1,\dots,r_n)\qquad\text{and}\qquad\mathbf{R}_1:=(R_1,\dots,R_n),
\end{equation*}
when $r_1+\dots+r_n=1$. Recall, we previously defined $\{\lambda_j\}_{j=1}^n$ to denote the $n$ distinct eigenvalues of $T$, $\{r_i\}_{i=1}^n$ as the first component of the right eigenvectors and $\{R_i\}_{i=1}^n$ as the elements of the first right eigenvector. There is a nuance in obtaining the $6n-4$ real spectral parameters since there is an overlap between the sets $\{r_i\}_{i=1}^n$ and $\{R_i\}_{i=1}^n$ that needs to be addressed. Specifically, we have $r_1=R_1$. Furthermore, only $n-1$ of $r_i$ are independent since $r_1+\dots+r_n=1$. Let the Vandermonde determinant be denoted by $\Delta_n(\blambda)=\prod_{1\leq j<k\leq n}(\lambda_j-\lambda_k)$.

\begin{Lemma}\label{lem:jac1}
Let $T=R\Lambda R^{-1}$ be the spectral decomposition of $T\in \Lw(n)$, where $R$ is defined in~\eqref{def:Rgenmodel} and
$\Lambda = \diag(\lambda_1, \dots,\lambda_n)$. We have
\begin{equation}\label{eq:Vandermonde determinant reduced}
\Delta_n(\boldsymbol{\lambda})^2=\frac{\prod_{j=1}^{n-1}(b_jc_j)^j}{\prod_{j=1}^{n}r_{j}},
\end{equation}
where $\rv = (r_1,\dots,r_n)$ is the first row of $R$ and
$r_1 + \dotsb + r_n =1$.
\end{Lemma}
\begin{proof}
 The proof follows analogously to the proof of~\cite[Lemma~3.3]{mezzadri-taylor} when the substitution $b_j\to b_jc_j$ is made.
\end{proof}

The following theorem generalises~\cite[Theorem~3.1]{mezzadri-taylor} to matrices of the form~\eqref{eq-gen_trid}.
\begin{Theorem}\label{thm:generalised Jacobian}
 We have
\begin{equation*}
\prod_{j=1}^n {\rm d}^2a_j \prod_{k=1}^{n-1}{\rm d}^2b_k
{\rm d}^2c_k = \left|\frac{\prod_{j=1}^{n-1}b_jc_j }{\prod_{k=1}^n r_k\prod_{w=2}^nR_w}\right|^2
\prod_{k=1}^{n-1}{\rm d}^2r_k\prod_{w=2}^n{\rm d}^2R_w\prod_{j=1}^n {\rm d}^2\lambda_j.
\end{equation*}
\end{Theorem}
\begin{proof}
Let us first repeat the domains of the complex variables. We have \smash{$\{a_j\}_{j=1}^n\in \mathbb{C}^n$}, \smash{$\{b_j\}_{j=1}^{n-1}\in \mathbb{C}^{n-1}$}, \smash{$\{c_j\}_{j=1}^{n-1}\in \mathbb{C}^{n-1}$} for the matrix parameters and \smash{$\{\lambda_j\}_{j=1}^n\in \mathbb{C}^n$}, \smash{$\{r_j\}_{j=1}^{n}\in \mathbb{C}^{n-1}$} with $\sum_{j=1}^n r_j=1$ and $\{R_j\}_{j=2}^n\in \mathbb{C}^{n-1}$ as spectral parameters.

Let $T=R\Lambda R^{-1}$ denote the spectral decomposition of $T$. Denote by $\mathbf{R}_1$ the first column of~$R$ and let its $j$-th component be denoted by $R_j$. We note that the first right eigenvector of $T$ satisfies the equation
\[
T\mathbf{R}_1=\lambda_1\mathbf{R}_1.
\]
Hence, if we equate row $j$ of this equation then
\begin{align}\label{eq:firstrighteigenvector}
 b_{n-j+1}R_{j-1}+a_{n-j+1}R_j+c_{n-j}R_{j+1}=\lambda_1R_j
\end{align}
for $1\leq j\leq n$ with boundary conditions $R_0=c_0=0$. The differential of this equation is
\begin{align}
&R_{j-1}{\rm d}b_{n-j+1}+R_j{\rm d}a_{n-j+1}+R_{j+1}{\rm d}c_{n-j}\nonumber\\
& \qquad =R_j{\rm d}\lambda_1+{\rm d}R_j[\lambda_1-a_{n-j+1}]-b_{n-j+1}{\rm d}R_{j-1}-c_{n-j}{\rm d}R_{j+1}.\label{eq:diff}
\end{align}
Consider the first element of the von Neumann series
\begin{align*}
\bigl((\mathbf{1}_n-zT)^{-1}\bigr)_{11}=\sum_{j=1}^n\frac{r_j}{1-z\lambda_j},
\end{align*}
where
right-hand side follows from the diagonalisation $T=R\Lambda R^{-1}$.
Taking the Taylor expansion of both sides of the equation and equating the coefficients of the powers of $z$ up to $z^{2n-1}$, interlaced with the rows of the eigenvector equation \eqref{eq:firstrighteigenvector}, we achieve\footnote{Note, for every even integer power $2j$ of $z$ we introduce $b_{n-j}$ and $c_{n-j+1}$. To counteract adding two new variables each even integer line we interlace in the rows of the eigenvector equation. This provides us with a~formula for the $b_j$.}
\begin{subequations}
 \label{eq:Taylor_exps_1}
\begin{align}
 \label{eq:ord_0_1}
 &O\bigl(z^0\bigr)\colon \quad 1=\sum_{j=1}^nr_j,\\
 \label{eq:ex1_1}
 &O\bigl(z^1\bigr)\colon \quad a_n=\sum_{j=1}^nr_j\lambda_j,\\
&j=1\colon \quad a_nR_1+c_{n-1}R_2=\lambda_1R_1,\\
&O\bigl(z^2\bigr)\colon \quad
 a_n^2+c_{n-1}b_{n-1}={*} +c_{n-1}b_{n-1}=\sum_{j=1}^nr_j\lambda_j^2
 \label{231d},\\
 &O\bigl(z^3\bigr)\colon \quad{} {*} +a_{n-1}b_{n-1}c_{n-1}=\sum_{j=1}^nr_j\lambda_j^3,\\
&j=2\colon \quad
{} {*} +c_{n-2}R_3=\lambda_1R_2,\\
&O\bigl(z^4\bigr)\colon \quad{} {*} +b_{n-1}c_{n-1}b_{n-2}c_{n-2}=\sum_{j=1}^nr_j\lambda_j^4,\\
 &O\bigl(z^5\bigr)\colon \quad{}{*} +a_{n-2}b_{n-2}c_{n-2}b_{n-1}c_{n-1}=\sum_{j=1}^nr_j\lambda_j^5,\\
 \nonumber & \quad \vdots \hspace{3cm}\vdots \hspace{3.5cm}\vdots\\
 &j=n-1\colon \quad
 {} * +c_{1}R_n=\lambda_1R_{n-1},\\
 &O\bigl(z^{2n-2}\bigr)\colon \quad{} {*} +b_1c_1\cdots b_{n-1}c_{n-1}=\sum_{j=1}^nr_j\lambda_j^{2n-2},\\
 \label{eq:ex2n-1_1}
 &O\bigl(z^{2n-1}\bigr)\colon \quad{} {*}+a_1b_{1}c_{1}\cdots b_{n-1}c_{n-1}=\sum_{j=1}^nr_j\lambda_j^{2n-1}.
\end{align}
\end{subequations}
Here, $*$ denotes terms that have appeared in the previous equations, as exemplified in \eqref{231d}.

We choose the ordering $a_n, c_{n-1}, b_{n-1}, a_{n-1}, c_{n-2}, b_{n-2},\dots$ such that the matrix of coefficients for the differentials of the left-hand side is triangular. Hence, taking the wedge product line by line gives
\begin{align}\label{eq:leftderivative}
(-1)^{(n-1)(n-2)/2}\prod_{j=1}^{n-1}R_{j+1}\prod_{w=1}^{n-1}(b_{w}c_{w})^{2w-1}\prod_{i=1}^{n-1}c_i {\rm d}\mathbf{a}\wedge {\rm d}\mathbf{b}\wedge {\rm d}\mathbf{c},
\end{align}
whilst noting that we have re-arranged equation~\eqref{eq:diff} such that the left-hand side of the differentials of the eigenvector equations~\eqref{eq:Taylor_exps_1} is
\begin{equation}\label{eq:lhsderivative}
R_{j-1}{\rm d}b_{n-j+1}+R_j{\rm d}a_{n-j+1}+R_{j+1}{\rm d}c_{n-j}.
\end{equation}
Differentiating the right-hand side of equations~(\ref{eq:Taylor_exps_1}), without the introduction of the rows of the eigenvector equations, leads to
\begin{equation*}
{\rm d}\left(\sum_{k=1}^nr_k\lambda_k^j\right)=\sum_{k=1}^{n-1}{\rm d}r_k\bigl(\lambda_k^j-\lambda_n^j\bigr)+j\sum_{k=1}^nr_k\lambda_k^{j-1}{\rm d}\lambda_k,
\end{equation*}
where we have used the differential relation
\[
{\rm d}r_n=-\sum_{j=1}^{n-1}{\rm d}r_j,
\]
as a substitution for ${\rm d}r_n$. Next, we note that the right-hand side of the differential of equation~\eqref{eq:firstrighteigenvector} is fixed by equation~\eqref{eq:lhsderivative} to be
\begin{equation*}
R_j{\rm d}\lambda_1+{\rm d}R_j[\lambda_1-a_{n-j+1}]-b_{n-j+1}{\rm d}R_{j-1}-c_{n-j}{\rm d}R_{j+1}.
\end{equation*}

Taking the wedge product, line by line, of the right-hand side differentials of equation~(\ref{eq:Taylor_exps_1}) produces
\begin{align}\label{eq:Jacrearrange1}
 \prod_{j=1}^{n-1}c_j\left(\prod_{k=1}^nr_k \right)
 \det\bigg[\bigl(\lambda_k^j-\lambda_n^j\bigr)_{
 \begin{subarray}{l}
 j=1,\dots,2n-1 \\ k=1,\dots,n-1
 \end{subarray}}
 \bigl(j\lambda_k^{j-1}\bigr)_{
 \begin{subarray}{l}
 j=1,\dots,2n-1 \\ k=1,\dots,n
 \end{subarray}}
 \bigg]
 {\rm d}\mathbf{r}\wedge {\rm d}\mathbf{R}_1\wedge {\rm d}\bm{\lambda}.
\end{align}
The determinant is calculated explicitly in the book by Forrester~\cite[equation~(1.175)]{PetersBook}.
We have
\begin{gather}
(-1)^{(n-1)(n-2)/2}\det\bigg[\bigl(\lambda_k^j-\lambda_n^j\bigr)_{
 \begin{subarray}{l}
 j=1,\dots,2n-1 \\ k=1,\dots,n-1
 \end{subarray}}
 \bigl(j\lambda_k^{j-1}\bigr)_{
 \begin{subarray}{l}
 j=1,\dots,2n-1 \\ k=1,\dots,n
 \end{subarray}}
 \bigg] \nonumber \\
 \label{eq:crazy_det}
 \qquad = \prod_{1 \le j < k \le n}\bigl(\lambda_k - \lambda_j\bigr)^4 = \Delta_n(\blambda)^4.
\end{gather}
Substituting equation~\eqref{eq:crazy_det} into equation~\eqref{eq:Jacrearrange1}, we produce
\begin{align}\label{eq:rhsderivative}
 &(-1)^{(n-1)}(-1)^{(n-1)(n-2)/2}\prod_{j=1}^{n-1}c_j\left(\prod_{k=1}^nr_k \right)\Delta_n(\bm{\lambda})^4
 {\rm d}\mathbf{r}\wedge {\rm d}\mathbf{R}_1\wedge {\rm d}\bm{\lambda}.
\end{align}
Note, to achieve the above result one needs to utilise that the wedge product is anti-commutative. This allows us to swap the order of the wedge products and therefore obtain the same wedge product that was seen in the non-symmetric matrix, wedged with the differentials of the interlaced eigenvector equation rows.
Combining the equations~(\ref{eq:leftderivative}) and~(\ref{eq:rhsderivative}), we obtain
\[
\prod_{j=1}^{n-1}R_{j+1}\prod_{i=1}^{n-1}(b_{i}c_{i})^{2i-1}{\rm d}\mathbf{a}\wedge {\rm d}\mathbf{b}\wedge {\rm d}\mathbf{c}=(-1)^{(n-1)}\left(\prod_{k=1}^nr_k \right)\Delta_n(\bm{\lambda})^4
 {\rm d}\mathbf{r}\wedge {\rm d}\mathbf{R}_1\wedge {\rm d}\bm{\lambda}.
\]
Rearranging this expression gives
\begin{equation}
\label{eq:J T proof}
{\rm d}\mathbf{a}\wedge {\rm d}\mathbf{b}\wedge {\rm d}\mathbf{c} =(-1)^{n-1} \frac{\bigl(\prod_{k=1}^n r_k\bigr) \Delta_n(\blambda)^4}{\prod_{j=1}^{n-1}(c_{j}b_{j})^{2j-1}\prod_{w=2}^{n}R_{w}} {\rm d}\mathbf{r}\wedge {\rm d}\mathbf{R}_1\wedge {\rm d}\bm{\lambda}.
\end{equation}
Using Lemma~\ref{lem:jac1}, specifically substituting equation~\eqref{eq:Vandermonde determinant reduced} into equation~\eqref{eq:J T proof}, and squaring yields the desired result.
\end{proof}

\subsection{Proof of Theorem~\ref{jpdf_theorem}}\label{Sec:UnBetaEnsemble}
\begin{proof}
The jpdf of the matrix elements of matrices $T\in\Lw'(n)$, distributed according to equation~\eqref{mat_el_T}, can be expressed as
\begin{equation}
 \label{jpdf_mat_el_1}
 p_{\beta}(\av,\bv,\cv) = \frac{1}{Z_\beta^{T}}\prod_{j=1}^{n-1}\abs{b_jc_j}^{\frac{\beta j}{2} -2}
 \exp\left[ -\frac12\left(\sum_{j=1}^n \abs{a_j}^2 + \sum_{j=1}^{n-1}
 \abs{b_j}^2+\sum_{j=1}^{n-1}|c_j|^2\right)\right],
\end{equation}
where the normalising partition function reads
\begin{equation*}
Z_{\beta}^{T}=\pi^{3n-2}2^{\frac{n(\beta n-\beta+4)}{4}}\prod_{k=1}^{n-1}\Gamma\left(\frac{\beta k}{4}\right)^2.
\end{equation*}
Notice one extra inverse power in the exponent $\frac{\beta j}{2}-2$ in \eqref{jpdf_mat_el_1}, compared to the $\chi$-dis\-tri\-bu\-tion~\eqref{mat_el_T}. This is because when going to polar coordinates, the differentials $d^2b_jd^2c_j$ will contribute an extra positive power in the radii.
By Lemma~\ref{lem:jac1} and Theorem~\ref{thm:generalised Jacobian}, we obtain
\begin{gather}
\prod_{j=1}^{n-1}\abs{b_jc_j}^{\frac{\beta j}{2} -2} \prod_{w=1}^n{\rm d}^2a_w\prod_{k=1}^{n-1}{\rm d}^2b_k{\rm d}^2c_k\nonumber\\
\qquad{} = \abs{\Delta_n(\blambda)}^\beta \prod_{j=1}^{n}\abs{r_j}^{\frac{\beta }{2}-2}\prod_{k=2}^n|R_k|^{-2}{\rm d}^2R_k{\rm d}^2r_{k-1}\prod_{w=1}^n{\rm d}^2\lambda_w.\label{eqmes_1}
\end{gather}
The exponent in equation~\eqref{jpdf_mat_el_1} can be expressed in terms of the squared Frobenius norm of $T$. Specifically,
\begin{equation}
\label{Frobenius_exp}
\norm{T}_{\mathrm{F}}^2= \Tr\bigl(T^\dag T\bigr)=\sum_{j=1}^n \abs{a_j}^2 + \sum_{j=1}^{n-1}\abs{b_j}^2+\sum_{j=1}^{n-1}\abs{c_j}^2.
\end{equation}
Next, we use the Schur decomposition to simplify the exponent further. Let
\begin{equation}
 \label{schur1}
 T = Q(\Lambda + N)Q^\dagger,
\end{equation}
where $\Lambda \in \Lambda(n)$, $Q$ is a unitary matrix and $N$ strictly upper triangular. Combining equation~\eqref{schur1} with Theorem~\ref{thm:roughbijection}, we write
\begin{equation}
 \label{exponent1}
 \norm{T}_{\mathrm{F}}^2 =
 \sum_{j=1}^n \abs{\lambda_j}^2 + 2g^{T}(\blambda,\rv,\Rv_1),
\end{equation}
where $g^{T}(\blambda,\rv,\Rv_1)$ is a positive function defined by equation~\eqref{gTdef}. Then by combining equations~\eqref{jpdf_mat_el_1}, \eqref{eqmes_1}, and~\eqref{exponent1}, we obtain
\begin{align*}
&p_\beta(\av,\bv,\cv)\prod_{j=1}^n{\rm d}^2a_j\prod_{k=1}^{n-1}{\rm d}^2b_k{\rm d}^2c_k=P_\beta^T(\blambda,\rv,\Rv_1)\prod_{j=1}{\rm d}^2\lambda_j\prod_{w=1}^{n-1}{\rm d}^2r_w{\rm d}^2R_{w+1}\nonumber\\
&\qquad = \frac{1}{Z_\beta^{T}}
{\rm e}^{-\frac12 \sum_{j=1}^n \abs{\lambda_j}^2 - g^{T}(\blambda,\rv,\Rv_1)}
\abs{\Delta_n(\blambda)}^\beta \prod_{j=1}^{n}\abs{r_j}^{\frac{\beta }{2}-2}\prod_{k=2}^n|R_k|^{-2}\prod_{j=1}^n{\rm d}^2\lambda_j\prod_{k=1}^{n-1}{\rm d}^2r_k{\rm d}^2R_{k+1}.
\end{align*}
When going to polar coordinates, there is an apparent logarithmic singularity in the integration over the radii $|R_k|$, with the jpdf being proportional to $1/|R_k|$ for each $k=2,\dots,n$.
However, when expressing the variables in the trace \eqref{Frobenius_exp} defining $g^{T}(\blambda,\rv,\Rv_1)$ in terms of variables~$R_j$,~$r_j$ and~$\lambda_j$,\footnote{For an explicit example at $n=2$, see Appendix~\ref{app-n=2}, equation~\eqref{eq-matrix elements T n=2}.}
in particular the $|c_j|^2$, it becomes apparent from equation~\eqref{eq:Taylor_exps_1} that for all $j=1,\dots,n-1$ we have $|c_{n-j}|^2\sim 1/|R_{j+1}|^2$. Therefore, we have an essential singularity from the exponent that makes the integral over the radii $|R_k|$ well-behaved at the origin.

Finally, integrating over $\rv$ and $\Rv_1$ leads to equations~\eqref{eq:EigDens1} and~\eqref{factor}.

The equality
 \[
 g^{T}(\blambda {\rm e}^{{\rm i}\phi},\rv,\Rv_1)=g^{T}(\blambda,\rv,\Rv_1)\qquad \text{for all} \ \phi\in[0,2\pi),
 \]
 follows from the observation that
 \[
 |a_{n-j+1}|=|\rv_j\Lambda\lv_j|,\qquad|c_{n-j}|=|\rv_j\Lambda\lv_{j+1}|,\qquad|b_{n-j+1}|=|\rv_{j+1}\Lambda\lv_j|
 \]
 for $j=1,2,\dots,n$. Furthermore, the mapping $\blambda\rightarrow\blambda {\rm e}^{{\rm i}\phi}$ leaves the joint probability density function of the eigenvalues invariant under rotation of the eigenvalues. To show that
 \[
g^{T}(\blambda,\rv,\Rv_1)=O\bigl(|\lambda_k|^2\bigr)\qquad \text{as} \ \lambda_k\rightarrow\infty,
 \]
 for $k=1,2,\dots,n$, we observe the following asymptotic behaviour:
 \begin{equation*}
 \begin{split}
a_{n-j+1}=O (\lambda_k ),\qquad c_{n-j}=O (\lambda_k ),\qquad b_{n-j+1}=O (\lambda_k ),
 \end{split}
 \end{equation*}
 as $\lambda_k\rightarrow\infty$ and $j=1,2,\dots,n$. This completes the proof.
\end{proof}

\section[Low temperature expansion beta>>1]{Low temperature expansion $\boldsymbol{\beta\gg1}$}\label{Sec:low temperature expansion}
Throughout this section, if $B$ denotes an $n\times n$ complex matrix, let its eigenvalues be given by $\lambda_1(B),\lambda_2(B),\dots,\lambda_n(B)$, labelled such that
\begin{equation*}
 |\lambda_n(B)|\leq|\lambda_{n-1}(B)|\leq\dots\leq|\lambda_1(B)|.
\end{equation*}
We study the asymptotics of the eigenvalues of the two non-Hermitian $\beta$-ensembles $T$ and \smash{$\widetilde{T}$}, focusing on fixed matrix dimension $n$ and the limit $\beta\rightarrow\infty$, known as the low temperature limit.

One might ask: what happens to the eigenvalues in this limit? For the Hermitian $\beta$-Hermite and $\beta$-Laguerre ensembles, Dumitriu and Edelman showed in~\cite{DE02B} that as $\beta\rightarrow\infty$, the eigenvalues crystallise, fixing their positions as the roots of the Hermite and Laguerre polynomials, respectively. In this low temperature regime, the randomness of the eigenvalues is effectively removed.

However, in contrast to the two-dimensional Coulomb gas, we have not observed numerically a condensation of complex eigenvalues at large values of $\beta$, in all three ensembles we consider here.
Nevertheless, the question can be asked if a simplification also occurs for our non-Hermitian $\beta$-ensembles at large $\beta$.

To obtain asymptotic results for the eigenvalues of the non-Hermitian $\beta$-ensembles in this low temperature limit, we first introduce two lemmas. The first, a perturbation theory lemma~\cite[Section 4.3]{Demmel}, is crucial for our analysis. A particular form of this lemma, applicable when $A$ and~$B$ are symmetric, will also be used and is detailed in~\cite[Lemma 2.1]{DE02B}.
\begin{Lemma}\label{lem:perturb}
 Let $A$ and $B$ be $n\times n$ matrices, and let $\epsilon>0$. Suppose the eigenvalues of $A$ are distinct. Define $M=A+\epsilon B+o(\epsilon)$, where $o(\epsilon)$ defines a matrix in which every entry goes to zero faster than $\epsilon$ goes to zero. Let $\lambda_i(A)$ denote the $i$-th eigenvalue of $A$ for $1\leq i \leq n$ and let~$\mathbf{x}_i$ and $\mathbf{y}_i^\dagger$ denote the right and left eigenvectors associated with~$\lambda_i(A)$, respectively. Then, we have
 \[
 \lim_{\epsilon\rightarrow0}\frac{1}{\epsilon}(\lambda_i(M)-\lambda_i(A))=\mathbf{y}_i^\dagger B\mathbf{x}_i.
 \]
 Equivalently, for every $1\leq i\leq n$, the eigenvalue of $\lambda_i(M)$ can be expanded as
 \[
 \lambda_i(M)=\lambda_i(A)+\epsilon \mathbf{y}_i^\dagger B\mathbf{x}_i.
 \]
\end{Lemma}
Succeeding this, we detail a second lemma~\cite[Lemma 2.3]{DE02B} in which one considers a $\chi_r$\nobreakdash-dis\-tri\-bu\-tion where the degree of freedom $r$ is growing.
\begin{Lemma}\label{lem:chilimit}
 Let $r>0$ and $X$ be a random variable with distribution $\chi_r$. Then, as $r\rightarrow\infty$ the distribution of $X-\sqrt{r}$ converges to a normal distribution of mean $0$ and variance $1/2$.
\end{Lemma}
\subsection{General ensemble}\label{ref:lowtemp}
Let $T_{\beta}$ be a random matrix from the general complex $\beta$-ensemble of $n\times n$ matrices defined in~\eqref{eq-gen_trid}, scaled by $\frac{1}{\sqrt{2n\beta}}$.
\begin{Theorem}\label{thm:largebetagen}
 Let $\lambda_i(T_{\beta})$ be the $i$-th eigenvalue of $T_{\beta}$ for $1\leq i\leq n$. Then,
 \[
 \lambda_i(T_{\beta})\rightarrow\frac{1}{\sqrt{2n}}\lambda_i(D_n)+o\left(\frac{1}{\sqrt{\beta}}\right),
 \]
 in distribution as $\beta\rightarrow\infty$, where
 \begin{equation}\label{eq:largebetasgen}
 D_{n}=\frac{1}{\sqrt{2}}\begin{pmatrix}
0 & \sqrt{n-1}{\rm e}^{{\rm i}\phi_{n-1}} & & & \\
\sqrt{n-1}{\rm e}^{{\rm i}\theta_{n-1}} & 0 & \ddots & & \\
 & \ddots & \ddots & \ddots & \\
 & & \sqrt{2}{\rm e}^{{\rm i}\theta_{2}} & 0 & \sqrt{1}{\rm e}^{{\rm i}\phi_1} \\
 & & & \sqrt{1}{\rm e}^{{\rm i}\theta_1} & 0
\end{pmatrix},
\end{equation}
and $\theta_i$, $\phi_i$ are independent random variables distributed uniformly, $\theta_i,
\phi_i\sim\mathcal{U}(0,2\pi)$ for $1\leq i \leq n-1$. All other elements vanish. Let
\begin{equation*}
 \mathcal{P}_k^{D}(z):=\det(z\mathbf{1}_k-D_k)
\end{equation*}
denote the characteristic polynomials of $D_k$ for $1\leq k\leq n$.
Then,
 \begin{align*}
 &\sqrt{\beta}\left(\lambda_i(L_{\beta})-\frac{1}{\sqrt{2n}}\lambda_i(D_n)\right)\\
 &=\frac{1}{\sqrt{2n}}\frac{\sum_{l=0}^{n-1}\! \mathcal{P}^D_{l}(\lambda_i)\mathcal{P}^{D}_{l}(\lambda_i^*)M_{l+1}+\!\sum_{l=1}^{n-1}\! [\mathcal{P}^{D}_{l-1}(\lambda_i^*)\mathcal{P}^D_l(\lambda_i)U_{l}+\mathcal{P}^{D}_{l} (\lambda_i^*)\mathcal{P}^D_{l-1}(\lambda_i)V_{l} ]}{\sum_{l=0}^{n-1}\!\mathcal{P}^D_l(\lambda_i)\mathcal{P}^{D}_l(\lambda_i^*)}\!+o(1),
 \end{align*}
in distribution as $\beta\rightarrow\infty$
for the independent random variables $\re(M_j), \im(M_j)\sim\mathcal{N}(0,1)$, $|U_w|,|V_w|\sim\mathcal{N}\bigl(0,\frac{1}{2}\bigr)$ and $\arg(U_w),\arg(V_w)\sim\mathcal{U}(0,2\pi)$ for $1\leq w\leq n-1$ and $1\leq j\leq n$.
\end{Theorem}

\begin{proof}
 Let $D_n$ denote an $n\times n$ non-Hermitian tridiagonal matrix defined by equation~\eqref{eq:largebetasgen}. We~restrict ourselves to the set of matrices $D_n$ with non-vanishing determinant and all eigenvalues distinct. As discussed in Section~\ref{sec:Decomp}, the exceptional set of matrices with degenerate spectra forms a set of measure zero. Associated to the $i$-th eigenvalue $\lambda_i:=\lambda_i(D_n)$, the matrix $D_n$ has right and left eigenvectors given by
 \begin{equation}\label{eq: right and left eigenvectors D}
 \mathbf{u}_i=\begin{pmatrix}
 \mathcal{P}^D_{n-1}(\lambda_i)\\
 \mathcal{P}^D_{n-2}(\lambda_i)\\
 \vdots\\
 \mathcal{P}^D_0(\lambda_i)
 \end{pmatrix}\qquad\text{and}\qquad
 \mathbf{v}_i^\dagger=\begin{pmatrix}
 \mathcal{P}^{D}_{n-1}(\lambda_i^*), \mathcal{P}^{D}_{n-2}(\lambda_i^*), \dots,\mathcal{P}^{D}_{0}(\lambda_i^*)
 \end{pmatrix}.
 \end{equation}
From equation~\eqref{eq:largebetasgen}, it follows that
\begin{equation}\label{eq:limbetagen}
\lim_{\beta\rightarrow\infty}\bigl[\sqrt{2n\beta}T_{\beta}-\sqrt{\beta}D_n\bigr]=Z^D,
\end{equation}
almost surely, where $Z^D$ is a random $n\times n$ complex tridiagonal matrix with independently distributed entries
\begin{equation}\label{eq:Zabc}
Z^D=\begin{pmatrix}
M_n & V_{n-1} & & \\
U_{n-1} & \ddots & \ddots & \\
 & \ddots & M_2 & V_1 \\
 & & U_1 & M_1
\end{pmatrix}.
\end{equation}
The matrix entries are distributed according to the following distributions: $\re(M_j)\sim\mathcal{N}(0,1)$ and $\im(M_j)\sim\mathcal{N}(0,1)$ for $1\leq j\leq n$, $|U_k|,|V_k|\sim\mathcal{N}\bigl(0,\frac{1}{2}\bigr)$ and $\arg(U_k),\arg(V_k)\sim\mathcal{U}(0,2\pi)$ for $1\leq k \leq n-1$.

Since we are considering $\chi_r$-distributions on the sub- and super-diagonals for increasing $r$, the limit~\eqref{eq:limbetagen} follows directly from the application of Lemma~\ref{lem:chilimit}. Hence,
\begin{equation*}
T_{\beta}=\frac{1}{\sqrt{2n}}D_n+\frac{1}{\sqrt{2n\beta}}Z^{D}+o\left(\frac{1}{\sqrt{\beta}}\right),
\end{equation*}
in distribution as $\beta\rightarrow\infty$.

Using this, we express the eigenvalues of $T_{\beta}$ as{\samepage
\[
\lambda_i(T_{\beta})=\lambda_i\left(\frac{1}{\sqrt{2n}}D_n +\frac{1}{\sqrt{2n\beta}}Z^{D}+o\left(\frac{1}{\sqrt{\beta}}\right)\right),
\]
in distribution for any $1\leq i\leq n$.}

Using Lemma~\ref{lem:perturb}, we expand the $i$-th eigenvalue of $T_{\beta}$ as
\begin{equation}\label{eq:eigpertgen}
\lambda_i(T_{\beta})=\frac{1}{\sqrt{2n}}\lambda_i(D_n)+\frac{1}{\sqrt{2n\beta}}\frac{\mathbf{v}_i^\dagger Z^{D}\mathbf{u}_i}{\mathbf{v}_i^\dagger\mathbf{u}_i}+o\left(\frac{1}{\sqrt{\beta}}\right),
\end{equation}
where $1\leq i\leq n$, in distribution as $\beta\rightarrow\infty$.

Next, substituting the definition of $Z^D$ from equation~\eqref{eq:Zabc} into equation~\eqref{eq:eigpertgen}, we obtain
\begin{align*}
 &\sqrt{\beta}\left(\lambda_i(T_{\beta})-\frac{1}{\sqrt{2n}}\lambda_i(D_n)\right)\\
 &=\frac{1}{\sqrt{2n}}\frac{\sum_{l=0}^{n-1}\mathcal{P}^D_{l} (\lambda_i)\mathcal{P}^{D}_{l}(\lambda_i^*)M_{l+1}\!+\!\sum_{l=1}^{n-1}\! [\mathcal{P}^{D}_{l-1}(\lambda_i^*)\mathcal{P}^D_l(\lambda_i)U_{l} +\mathcal{P}^{D}_{l}(\lambda_i^*)\mathcal{P}^D_{l-1}(\lambda_i)V_{l} ]}{\sum_{l=0}^{n-1}\mathcal{P}^D_l(\lambda_i)\mathcal{P}^{D}_l(\lambda_i^*)}\!+o(1),
 \end{align*}
in distribution as $\beta\rightarrow\infty$. This concludes the proof.
\end{proof}

\begin{Remark}
A similar result can be proven for a matrix $S_{\beta}$ from the symmetric complex $\beta$-ensemble, scaled by \smash{$\frac{1}{\sqrt{2n\beta}}$}. To do so, in equation~\eqref{eq:largebetasgen} we let $\phi_j=\theta_j$ for all $j$. Therefore, the resulting matrix is symmetric and we note that in equation~\eqref{eq: right and left eigenvectors D} we have \smash{$\mathbf{u}_j^{\rm t}=\mathbf{v}_j^\dagger$}. The proof then follows analogously.
\end{Remark}

\subsection{Non-symmetric ensemble}
Let $\widetilde{T}_{\beta}$ be a random matrix from the non-symmetric ensemble of $n\times n$ matrices in \eqref{eq-T_tilde_def}, scaled by $1/\sqrt{2n\beta}$.
We now outline and prove the following theorem.
\begin{Theorem}\label{thm:largebetanonsym}
 Let $\lambda_i\bigl(\widetilde{T}_{\beta}\bigr)$ be the $i$-th eigenvalue of $\widetilde{T}_{\beta}$. Then, as $\beta\rightarrow\infty$
 \[
 \lambda_i(\widetilde{T}_{\beta})\rightarrow\frac{1}{\sqrt{2n}}\lambda_i(G_n)+o\left(\frac{1}{\sqrt{\beta}}\right)
 \]
 in distribution, where $G_n$ is the matrix
 \begin{equation}\label{eq:largebetasNonSymH}
 G_n=\begin{pmatrix}
0 & \frac{1}{\sqrt{\beta}} & & & \\
\sqrt{\frac{n-1}{2}}{\rm e}^{{\rm i}\theta_{n-1}} & 0 & \ddots & & \\
 & \ddots & \ddots & \ddots & \\
 & & \sqrt{\frac{2}{2}}{\rm e}^{{\rm i}\theta_{2}} & 0 & \frac{1}{\sqrt{\beta}} \\
 & & & \sqrt{\frac{1}{2}}{\rm e}^{{\rm i}\theta_1} & 0
\end{pmatrix},
 \end{equation}
 and the $\theta_i$'s are independent random variables uniformly distributed in $[0,2\pi]$, for $1\leq i\leq n-1$. The remaining matrix elements vanish. Let
 \begin{equation}\label{eq: Characrteristic Polynomail G}
\Pc_k^G(z):=\det(z\mathbf{1}_k-G_k)
\end{equation}
denote the characteristic polynomials of $G_k$ for $1\leq k\leq n$. Then, as $\beta\rightarrow\infty$
 \begin{equation*}
 \begin{split}
 &\sqrt{\beta}\left(\lambda_i(\widetilde{T}_{\beta})-\frac{1}{\sqrt{2n}}\lambda_i(G_n)\right)\\
 &\qquad=\frac{1}{\sqrt{2n}}\frac{\sum_{l=0}^{n-1} \Pc_{l}^G(\lambda_i)\mathcal{P}^{G}_{l}(\lambda_i^*)M_{l+1}+\sum_{l=1}^{n-1} \Pc^{G}_{l-1}(\lambda_i^*)\Pc_l^G(\lambda_i)U_{l}}{\sum_{l=0}^{n-1}\Pc^G_l(\lambda_i)P^{G}_l(\lambda_i^*)}+o(1),
 \end{split}
 \end{equation*}
 in distribution, where $\re(M_j)\sim\mathcal{N}(0,1)$, $\im(M_j)\sim\mathcal{N}(0,1)$, $|U_j|\sim\mathcal{N}\bigl(0,\frac{1}{2}\bigr)$ and $\arg(U_j)\sim\mathcal{U}(0,2\pi)$ are independent random variables.
\end{Theorem}

\begin{proof}
 Let $G_n$ denote the $n\times n$ non-Hermitian tridiagonal matrix defined by equation~\eqref{eq:largebetasNonSymH}. Associated with the $i$-th eigenvalue of $G_n$, namely $\lambda_i:=\lambda_i(G_n)$, we have the right and left eigenvectors
 \begin{equation*}
 \mathbf{u}_i=\begin{pmatrix}
 \Pc^G_{n-1}(\lambda_i)\\
 \Pc^G_{n-2}(\lambda_i)\\
 \vdots\\
 \Pc^G_0(\lambda_i)
 \end{pmatrix}\qquad\text{and}\qquad
 \mathbf{v}_i^\dagger=\begin{pmatrix}
 \Pc^{G}_{n-1}(\lambda_i^*), \Pc^{G}_{n-2}(\lambda_i^*), \dots,\Pc^{G}_{0}(\lambda_i^*)
 \end{pmatrix},
 \end{equation*}
respectively, for $1\leq k\leq n$.

Now, by application of Lemma~\ref{lem:chilimit}, almost surely we have
\begin{equation}\label{eq:Zlargebeta}
\lim_{\beta\rightarrow\infty}\bigl[\sqrt{2n\beta}\widetilde{T}_{\beta}-\sqrt{\beta}G_n\bigr]=Z^G=\begin{pmatrix}
M_n & 0 & & \\
U_{n-1} & \ddots & \ddots & \\
 & \ddots & M_2 & 0 \\
 & & U_1 & M_1
\end{pmatrix},
\end{equation}
where $Z^G$ is a random matrix whose elements are distributed according to $\re(M_j)\sim\mathcal{N}(0,1)$, $\im(M_j)\sim\mathcal{N}(0,1)$, $|U_w|\sim\mathcal{N}\bigl(0,\frac{1}{2}\bigr)$ and $\arg(U_w)\sim\mathcal{U}(0,2\pi)$ are independent random variables for $1\leq j\leq n$ and $1\leq w\leq n-1$.

Hence, as $\beta\rightarrow\infty$, we can express the tridiagonal matrix $\widetilde{T}_{\beta}$ as
\begin{equation*}
\widetilde{T}_{\beta}=\frac{1}{\sqrt{2n}}G_n+\frac{1}{\sqrt{2n\beta}}Z^G+o\left(\frac{1}{\sqrt{\beta}}\right),
\end{equation*}
in distribution. In particular, for any $1\leq i\leq n$, we have
\begin{equation*}
\lambda_i(\widetilde{T}_{\beta})=\lambda_i\left(\frac{1}{\sqrt{2n}}G_n +\frac{1}{\sqrt{2n\beta}}Z^G\right)+o\left(\frac{1}{\sqrt{\beta}}\right),
\end{equation*}
 in distribution as $\beta\rightarrow\infty$. Applying the perturbation theory result from Lemma~\ref{lem:perturb}, we obtain
\begin{equation}\label{eq:eigpertb}
\lambda_i(\widetilde{T}_{\beta})=\frac{1}{\sqrt{2n}}\lambda_i(G_n)+\frac{1}{\sqrt{2n\beta}}\frac{\mathbf{v}_i^\dagger Z^G\mathbf{u}_i}{\mathbf{v}_i^\dagger\mathbf{u}_i}+o\left(\frac{1}{\sqrt{\beta}}\right).
\end{equation}
Using the form of $Z^G$, defined by~\eqref{eq:Zlargebeta}, and~\eqref{eq:eigpertb}, we then have the desired result
 \begin{equation*}
 \begin{split}
 &\sqrt{\beta}\left(\lambda_i(\widetilde{T}_{\beta})-\frac{1}{\sqrt{2n}}\lambda_i(G_n)\right)\\
 &\qquad =\frac{1}{\sqrt{2n}}\frac{\sum_{l=0}^{n-1}\Pc^G_{l}(\lambda_i)\Pc^{G}_{l} (\lambda_i^*)M_{l+1}+\sum_{l=1}^{n-1}\Pc^{G}_{l-1}(\lambda_i^*)\Pc^G_l (\lambda_i)U_{l}}{\sum_{l=0}^{n-1}\Pc^G_l(\lambda_i)\Pc^{G}_l(\lambda_i^*)}+o(1),
 \end{split}
\end{equation*}
in distribution as $\beta\rightarrow\infty$, which completes the proof of the theorem.
\end{proof}

In the case of the non-symmetric tridiagonal ensemble, we apply the low temperature limit to the continuous density function of the diagonal and sub-diagonal elements to illustrate a similar limiting behaviour as in the general and symmetric ensembles. However, we keep the order $\mathcal{O}\bigl(1/\sqrt{\beta}\bigr)$ on the super-diagonal at the same time. The reason is the following: If we neglect this order at this stage, we find a matrix with all zero elements up to the sub-diagonal which yields directly only zero eigenvalues.

\section{Characteristic polynomials of centred Jacobi matrices}\label{Sec:characteristic polynomial}

In this section, we study the characteristic polynomials of the random matrices we derived in the low temperature limit $\beta\gg1$ in the previous section. As they both have only non-zero entries on their sub-and super-diagonal, we keep the discussion
general and treat the characteristic polynomial of the general centred matrix $T|_{a_j=0}$ denoted by $\mathcal{P}_n(z)$ (cf.\ equation~\eqref{eq: Characteristic Polynomial}).
This covers both cases $D_n$ in equation~\eqref{eq:largebetasgen} and $G_n$ equation~\eqref{eq:largebetasNonSymH} to which we will apply the results of this section subsequently.
\begin{proof}[Proof of Proposition~\ref{prop-formula coefficients P}]
The three-term recurrence \eqref{eq-3 term recurrence} follows from the form of $T$ with vanishing diagonal elements $a_j$ and expansion of the determinant twice, for instance first with respect to the first row and second with respect to the first column.

Recall, that we set $\mathcal{P}_0(z):=1$. The first characteristic polynomials for $n=0,1,2,3$ can be calculated straightforwardly from determining the determinant and read
\begin{align}
\begin{split}
& \mathcal{P}_0(z)=1,\\
& \mathcal{P}_1(z)=z,\\
& \mathcal{P}_2(z)=z^2-\widetilde{b}_1, \\
& \mathcal{P}_3(z)=z^3-\bigl(\widetilde{b}_2+\widetilde{b}_1\bigr)z.
\end{split}
\label{eq-characteristic polynomial small n examples}
\end{align}
From these explicit forms of the characteristic polynomials for small $n$ and the three-term recurrence \eqref{eq-3 term recurrence} without diagonal term it follows directly that the polynomials $\mathcal{P}_n(z)$ are either even or odd, depending on the parity of their degree $n$.

Let us emphasise again that the polynomials $\mathcal{P}_n(z)$ are not orthogonal polynomials on $\mathbb{R}$ as the coefficients \smash{$\widetilde{b}_j$} are complex and do not correspond to some real positive norm of orthogonal polynomials.

We start with the observation that the three-term recurrence~\eqref{eq-3 term recurrence} yields a recurrence relation for the coefficients. Inserting the notation from equation~\eqref{eq: Characteristic Polynomial} for $\mathcal{P}_{n+1}(z)$, $\mathcal{P}_{n}(z)$ and $\mathcal{P}_{n-1}(z)$ into~\eqref{eq-3 term recurrence} gives
\begin{align*}
 z^{n+1}+\sum_{\ell=1}^{\lfloor \frac{n+1}{2} \rfloor}\kappa_\ell^{(n+1)}z^{n+1-2\ell}=z^{n+1}+\sum_{\ell=1}^{\lfloor \frac{n}{2} \rfloor}\kappa_\ell^{(n)}z^{n+1-2\ell}-\widetilde{b}_{n}z^{n-1}-\widetilde{b}_{n}\sum_{\ell=1}^{\lfloor \frac{n-1}{2} \rfloor}\kappa_\ell^{(n-1)}z^{n-1-2\ell}.
\end{align*}
Comparing the coefficients in orders $z^{n+1-2\ell}$ for $\ell=1,\dots, \bigl\lfloor \frac{n+1}{2}\bigr\rfloor-1$ yields for the coefficients of $\mathcal{P}_{n+1}(z)$
\begin{align}\label{eq-recursion kappas}
 \kappa_\ell^{(n+1)}=\begin{cases}
 \kappa_1^{(n)}-\widetilde{b}_{n} & \text{for }\ell=1,\\
 \kappa_\ell^{(n)}-\widetilde{b}_{n}\kappa_{\ell-1}^{(n-1)} & \text{for } \ell=2,\dots, \lfloor \frac{n+1}{2}\rfloor-1,\\
 \kappa_m^{(n)}-\widetilde{b}_{n}\kappa_{m-1}^{(n-1)} & \text{for }\ell= \bigl\lfloor \frac{n+1}{2}\bigr\rfloor=m,\ \text{with}\ n+1=2m+1\ \text{odd},\\
\hphantom{\kappa_m^{(n)}}{}-\widetilde{b}_{n}\kappa_{m}^{(n-1)} & \text{for }\ell= \bigl\lfloor \frac{n+1}{2}\bigr\rfloor=m+1,\ \text{with}\ n+1=2m+2\ \text{even}.
 \end{cases}
\end{align}
We will treat the coefficient of the lowest order $z^0$ resulting from the last equation separately below. From \eqref{eq-characteristic polynomial small n examples}, we have that equation~\eqref{eq-formula coefficients P} is true for $n=0,1,2,3$. We prove via induction that this also holds for some $n+1$. Therefore, we insert equation~\eqref{eq-formula coefficients P} in the right-hand side of equation~\eqref{eq-recursion kappas} and find for $\ell=1$
\begin{align*}
 \kappa_1^{(n+1)}&=-\sum_{\gamma=1}^{n-1}\widetilde{b}_\gamma-\widetilde{b}_{n} =-\sum_{\gamma=1}^{n}\widetilde{b}_\gamma
\end{align*}
as claimed. We have for $\ell=2,\dots, \bigl\lfloor \frac{n+1}{2}\bigr\rfloor-1$ for $n+1$ even, and for $\ell=2,\dots, \bigl\lfloor \frac{n+1}{2}\bigr\rfloor$ for $n+1$ odd
\begin{align*}
 \kappa_\ell^{(n+1)}&=
 (-1)^\ell \sum_{\gamma_1=2\ell-1}^{n-1}\widetilde{b}_{\gamma_1} \sum_{\gamma_2=2\ell-3}^{\gamma_1-2}\widetilde{b}_{\gamma_2}
		\cdots \sum_{\gamma_\ell=1}^{\gamma_{\ell-1}-2}\widetilde{b}_{\gamma_\ell}\ -\ \widetilde{b}_{n}(-1)^{\ell-1} \sum_{\gamma_1=2\ell-3}^{n-2}\widetilde{b}_{\gamma_1}
		\cdots \sum_{\gamma_{\ell-1}=1}^{\gamma_{\ell-2}-2}\widetilde{b}_{\gamma_{\ell-1}} \\
 &=(-1)^\ell \left(\sum_{\gamma_1=2\ell-1}^{n-1}\widetilde{b}_{\gamma_1}+\widetilde{b}_n\right)\sum_{\gamma_2=2\ell-3}^{\gamma_1-2}\widetilde{b}_{\gamma_2}
		\cdots \sum_{\gamma_\ell=1}^{\gamma_{\ell-1}-2}\widetilde{b}_{\gamma_\ell}\\
 &=(-1)^\ell \sum_{\gamma_1=2\ell-1}^{n}\widetilde{b}_{\gamma_1} \sum_{\gamma_2=2\ell-3}^{\gamma_1-2}\widetilde{b}_{\gamma_2}
		\cdots \sum_{\gamma_\ell=1}^{\gamma_{\ell-1}-2}\widetilde{b}_{\gamma_\ell},
\end{align*}
which proves equation~\eqref{eq-formula coefficients P} via induction for all
coefficient resulting from order in $z$ down to~$z^2$ for $n+1=2m+2$ even, and down to $z^1$ for $n+1=2m+1$ odd.
For the coefficient of the lowest order $z^0$, we have to treat even $n+1=2m+2$ differently, as the recursion is different.
Here, we obtain
 \begin{align*}
			\kappa_{m+1}^{(2m+2)}&=-\widetilde{b}_{2m+1} \kappa_{m}^{(2m)}.
		\end{align*}
Iterating this recurrence, using the initial condition $\kappa_1^{(2)}=-\widetilde{b}_1$, and induction assumption
		\begin{align*}
			\kappa_{m}^{(2m)}=(-1)^m\prod_{\ell=1}^{m} \widetilde{b}_{2\ell-1},
		\end{align*}
we complete the proof of equation~\eqref{eq-formula coefficients P}.
We can further simplify the expression as follows: We relax each lower bound of the nested sums to $1$. One intermediate step reads as
	\begin{align*}
		\sum_{\gamma_p=2\ell-2p+1}^{\gamma_{p-1}-2}\!\widetilde{b}_{\gamma_p}\sum_{\gamma_{p+1}=2\ell-2p-1}^{\gamma_p-2}\!\widetilde{b}_{\gamma_{p+1}}=	\sum_{\gamma_p=1}^{\gamma_{p-1}-2}\widetilde{b}_{\gamma_p}\sum_{\gamma_{p+1}=2\ell-2p-1}^{\gamma_p-2}\!\widetilde{b}_{\gamma_{p+1}}-	\sum_{\gamma_p=1}^{2\ell-2p}\widetilde{b}_{\gamma_p}\sum_{\gamma_{p+1}=2\ell-2p-1}^{\gamma_p-2}\!\widetilde{b}_{\gamma_{p+1}}
	\end{align*}
	for $p=1,2,\dots,\ell$ with $\gamma_{0}=n+1$. Following our convention for the sums stated before Proposition~\ref{prop-formula coefficients P}, the second term vanishes as the upper bound of the inner sum is always smaller than the lower bound, i.e., 	\begin{align*}
		\gamma_p-2<2\ell -2p-1\qquad \text{for}\ 1\leq \gamma_p \leq 2\ell-2p. \tag*{\qed}
	\end{align*}
\renewcommand{\qed}{}
\end{proof}

We observe that the coefficients of the characteristic polynomial $\Pc_n(z)$ depend only on the product of the sub- and super-diagonal entries. This matches the well-known fact that the eigenvalues of a tridiagonal matrix depend only on the diagonal elements and the product of the sub- and super-diagonal elements (cf.\ \cite{cullum-willoghby}), and the study of Jacobi matrices~\cite{GK05}.

As mentioned already in the introduction, the expression for the coefficients equation~\eqref{eq-formula coefficients P} is given in \cite{PS06} for orthogonal polynomials on $\mathbb{R}$ and the same objects are also studied in \cite{LRR91} without proof. The latter authors found the following expression:
\begin{align}\label{eq-coeff LRR91}
		\kappa_{\ell}^{(n)}=(-1)^\ell \sum _{\{i_n\}\in \{0,1,\dots,n-1\}}^\prime \widetilde{b}_{i_1}\cdots \widetilde{b}_{i_\ell},
	\end{align}
	where the prime on the sum indicates that no two adjacent indices ever occur within ${i_n}$. Note that we adapted their result to our notation. The sum is defined in the following way: We start to choose $i_1$ in the full range $i_1=0,1,\dots,n-1$ and choose the next summation index such that it runs only up to $i_1-2$ to have disjoint indices. Similar to the earlier remark, the first summand will vanish for $i_2>i_1$. From this, it is straightforward to see that equations~\eqref{eq-formula coefficients P} and~\eqref{eq-coeff LRR91} agree.
\begin{Lemma} \label{lemma-mean and variance kappa's}
	Assume that the random variables $\bigl\{\widetilde{b}_i \bigl\}_{i=1}^{n-1}$ are independent of each other and have all zero mean. From this, it follows that the coefficients \smash{$\kappa_\ell^{(n)}$} \eqref{eq-formula coefficients P} of the polynomial $\mathcal{P}_n(z)$ are pairwise-uncorrelated of each other, have zero mean and variance
	\begin{align}\label{eq-variance kappas}	\Var\bigl(\kappa_\ell^{(n)}\bigr) =\sum_{\gamma_1=1}^{n-1}\Var\bigl(\widetilde{b}_{\gamma_1}\bigr) \sum_{\gamma_2=1}^{\gamma_1-2}\Var\bigl(\widetilde{b}_{\gamma_2}\bigr)\cdots \sum_{\gamma_\ell=1}^{\gamma_{\ell-1}-2}\Var\bigl(\widetilde{b}_{\gamma_\ell}\bigr),
	\end{align}
	$\forall \ell=1,2,\dots, \bigl\lfloor \frac{n}{2}\bigr\rfloor$.
\end{Lemma}
 \begin{proof}
 1.~The random variables $\bigl\{\widetilde{b}_i \bigr\}_{i=1}^{n-1}$ are independent by assumption. Let us recall the notions of independent and uncorrelated variables. For independent (complex) variables, the joint probability distribution factorises, especially all moments (of possibly complex conjugate powers) factorise. On the other hand, for uncorrelated variables only the covariance vanishes. The latter is defined as
		\begin{align*}
			\Cov(X,Z):=\mathbb{E}(X,Z^*),
		\end{align*}
		for complex random variables $X,Z$. We show that $\Cov\bigl(\kappa_{\ell}^{(n)},\kappa_{\ell'}^{(n)}\bigr)=0$ with $\ell<\ell'$. From the derived explicit form of the \smash{$\kappa_{\ell}^{(n)}$}'s in equation~\eqref{eq-formula coefficients P}, we find that it consists of $\ell$ nested sums, with each of the intersections of index sets of the sums being empty. Hence, each coefficient contains sums of a product of $\ell$ different \smash{$\widetilde{b}_j$}'s. As the latter are independent complex random variables with zero mean, the expectation of them can only yield a non-zero contribution if we find contributions~$|\widetilde{b}_j|^2$. We observe that the expectation value \smash{$\mathbb{E}\bigl(\kappa_{\ell}^{(n)},\kappa_{\ell'}^{(n)}\bigr)$} for $\ell<\ell'$ will vanish as we have always one factor of a \smash{$\widetilde{b}_j$} without a matching part $\widetilde{b}_j^*$.

2.~The random variables $\bigl\{\widetilde{b}_i \bigr\}_{i=1}^{n-1}$ have zero mean by assumption.

3.~Next, we study the sum $Z=Z_1+Z_2$ of two independent complex random variables with vanishing mean. The mean of $Z$ is given by the sum of the means of $Z_1$ and $Z_2$ and, hence vanishes. For the variance, we find
		\begin{align*}
			\Var(Z)=\Var(Z_1+Z_2) =\mathbf{E}\bigl(|Z_1|^2\bigr)+\mathbf{E}\bigl(|Z_2|^2\bigr).
		\end{align*}

4.~Let $Z_1$ and $Z_2$ be two independent complex random variables with zero mean. The mean of the product $Z=Z_1 Z_2$ vanishes and the variance is given by
		\begin{align*}
			\Var(Z)=\Var(Z_1)\Var(Z_2).
		\end{align*}

Hence, the mean of $\kappa_\ell^{(n)}$ vanishes for all $\ell$ and the variance is given by \eqref{eq-variance kappas}.
 \end{proof}

We will see in the next section how we can make use of the explicit formula for the coefficients of the characteristic polynomials, using the results from Section~\ref{Sec:low temperature expansion} in the low temperature limit to derive the equilibrium density of our non-Hermitian tridiagonal matrix ensembles.

\section[Equilibrium density via free probability for beta>>1 and n to infty]{Equilibrium density via free probability for $\boldsymbol{\beta \gg 1}$ and $\boldsymbol{n\rightarrow \infty}$} \label{Sec:equilibrium density}

We determine the equilibrium density of all three tridiagonal matrix ensembles under consideration in the low temperature limit $\beta \gg 1$. We find that the limiting density of the eigenvalues agrees for the general and the symmetric tridiagonal matrix ensemble, leading to a logarithmic density, whereas the leading order contribution for the non-symmetric ensemble is given by a~Dirac delta. For this, we study their characteristic polynomials and apply the following theorem from free probability, which determines the equilibrium density of the zeros of polynomials with random coefficients. We slightly rephrase the result from \cite[Theorem~1.2]{Barbarino-Noferini} for the case of random variables instead of random matrices as coefficients, replacing the condition of independent random variables by pairwise-uncorrelated ones.
\begin{Theorem}\label{thm-Barbarino Noferini}
	Suppose that for any $n$ we have a vector of pairwise-uncorrelated random variables \smash{$X_0^{(n)}, X_1^{(n)},\dots,X_n^{(n)}$} with zero mean, unit variance and continuous distributions with densities uniformly bounded by a constant $M>0$ not depending on~$n$. Let also \smash{$\alpha_0^{(n)},\alpha_1^{(n)},\dots, \alpha_n^{(n)}$} be sequences of complex numbers.
	We consider the polynomial of degree~$n$
	\begin{align}\label{eq-polynomial thm Barbarino Noferini}
		\Pc_n(z)=\sum_{j=0}^n \alpha_j^{(n)}X_j^{(n)}z^j.
	\end{align}
	Then, the empirical spectral distribution of zeros of the random polynomials $\Pc_n$ converges almost surely to a probability measure $\mu$ as $n\rightarrow \infty$, when the coefficients \smash{$\alpha_j^{(n)}$} satisfy the following conditions:
	\begin{itemize}\itemsep=0pt
		\item $\alpha_n^{(n)}\neq 0$ for every $n$,
		\item for almost every $z$, let
 \begin{align}\label{eq:Qfunc}
 Q_n(z)=\sum_{i=0}^n\frac{\bigl|\alpha_i^{(n)}\bigr|^2}{\bigl|\alpha_n^{(n)}\bigr|^2}z^i,
 \end{align}
 such that the limit $U(z)=\lim_{n\rightarrow \infty} \frac{1}{2n} \ln \bigl(Q_n\bigl(|z|^2\bigr)\bigr)$ exists as a locally integrable function or a distribution,
		\item the probability measure $\mu$ satisfies
		\begin{align*}
			\mu=\frac{1}{2\pi}\Delta U(z) \mathbf{1}_{\rm sup},
		\end{align*}
		in the sense of distributions, on the support $($sup$)$ of $\mu$. Here, $\Delta=\partial_z\partial_z^*$ denotes the two-dimensional Laplacian.
	\end{itemize}
\end{Theorem}
\begin{proof}
We relaxed the stronger condition of independence of the coefficients $X_k^{(n)}$ to pairwise-uncorrelateness. In \cite[Theorem 1.2]{Barbarino-Noferini}, the authors allow for random $k\times k$ matrices \smash{$C_j^{(n)}$} in the polynomials which contain independent copies of the random variables \smash{$X_j^{(n)}$}. Hence, the condition of independence appears twice in their theorem, i.e., that the complex random variables~\smash{$X_j^{(n)}$} are independent and the random matrices \smash{$C_j^{(n)}$} are built up by independent copies of the variables each. As we need only the case of $k=1$ (in the notation of~\cite{Barbarino-Noferini} $k$ and $n$ are interchanged), we briefly revisit the key points of their proof.

We can assume without loss of generality that $\bigl|\alpha_n^{(n)}\bigr|=1$ for each~$n$. If we fix $z$ and analyse~$\Pc_n(z)$ from \eqref{eq-polynomial thm Barbarino Noferini}, we find that it can be treated as a random variable with mean 0 and variance \smash{$\sigma^2=\sum_{i=0}^n\bigl|\alpha_i^{(n)}\bigr|^2|z|^{2i}$}. This follows from the fact that the \smash{$X_j^{(n)}$}'s are pairwise-uncorrelated as
\begin{align*} \Var\biggl(\sum_{i=0}^n\alpha_i^{(n)}X_i^{(n)}z^i\biggr)&=\sum_{i=0}^n\Var\bigl(\alpha_i^{(n)}X_i^{(n)}z^i\bigr)+2\sum_{i<j}^n\Cov\bigl(\alpha_i^{(n)}X_i^{(n)}z^i,\alpha_j^{(n)}X_j^{(n)}z^j\bigr)\\
			&=\sum_{i=0}^n\bigl|\alpha_i^{(n)}\bigr|^2|z|^{2i},
		\end{align*}
		where we used that the $X_j^{(n)}$'s have unit variance. Next, we need to prove that the random variable $\Pc_n(z)/\sigma$, which has zero mean and unit variance, has a bounded density function. We have that $\Pc_n(z)/\sigma$ is of the type
		\begin{align*}
			Y_n=X_n^{(n)}\frac{\alpha_n^{(n)}z^n}{\sigma}+X_{n-1}^{(n)}\frac{\alpha_{n-1}^{(n)}z^{n-1}}{\sigma}+\dots+X_1^{(n)}\frac{\alpha_1^{(n)}z}{\sigma}+X_0^{(n)}\frac{\alpha_0^{(n)}}{\sigma},
		\end{align*}
		where all $X_i^{(n)}$ have a density function uniformly bounded by a constant $M>0$ (by assumption). Here comes now the crucial change in the proof for uncorrelated variables. In the work of Barbarino--Noferini, where $Y_n$ is the sum of independent uniformly bounded random variables with zero mean and unit variance, they applied a result by Bobkov--Chistyakov~\cite{Bobkov-Chistyakov} to prove that~$Y_n$ is bounded by the very same constant $M$. Then, a matrix built up by independent copies of this $Y_n$ can be further studied. However, we can study this complex random variable directly and do not need this step here. In order to complete the proof, the following two bounds are needed:
		\begin{align*}
			&\mathbb{P}\bigl(\sigma(\Pc_n(z)/\sigma)\leq u\bigr)\leq c u^2,\\
			&\mathbb{P}\bigl(\lVert \Pc_n(z)/\sigma\rVert\geq v\bigr)\leq \frac{1}{v^2}\qquad \forall u,v>0,
		\end{align*}
		where $\sigma(A)$ denotes the largest singular value of the matrix $A$.
		Let us recall that we deal here with a much simpler case, i.e., $\Pc_n(z)/\sigma$ is a complex random variable for each fixed $z$ rather then a random matrix. The first inequality holds as we replace the squared singular value of~$\Pc_n(z)/\sigma$ in this case of a $1\times 1$ matrix by the absolute square. The second relation follows from Chebyshev's inequality. The rest of the proof stays the same as in the case discussed in~\cite{Barbarino-Noferini}.
\end{proof}

The remarkable point is that the coefficients of the polynomial $\Pc_n(z)$ in this theorem are not necessarily identically distributed. The condition of unit variance can be full-filled by pulling the non-unit variance of the coefficients into the complex coefficients \smash{$\alpha_j^{(k)}$}. We apply this theorem on random polynomials to the characteristic polynomials of the non-Hermitian tridiagonal matrix ensembles, to obtain the limiting distribution of the eigenvalues of the ensembles which equals the distribution of the zeros of their characteristic polynomials. In the low temperature limit $\beta \rightarrow \infty$, we determined an explicit expression of the characteristic polynomials in the previous Section~\ref{Sec:characteristic polynomial} and derived in Lemma~\ref{lemma-mean and variance kappa's} the variance and mean of the coefficients. Using these results one can in principle calculate the equilibrium density of the zeros of the characteristic polynomial~$\Pc_n(z)$, which correspond to the eigenvalues of the centred random Jacobi matrix $T|_{a_j=0}$ (full-filling the conditions of Lemma~\ref{lemma-mean and variance kappa's}) via Theorem~\ref{thm-Barbarino Noferini} by choosing \smash{$\bigl|\alpha_j^{(n)}\bigr|^2=\Var(b_jc_j)$}. However, as the above expression of the variances still contains the product of nested sums, it is difficult to proceed in this general setting. Therefore, we will now turn to the explicit low temperature limits we derived in Section~\ref{Sec:low temperature expansion}.

\subsection{General ensemble}\label{Sec:SpecDensUcon}
To derive the limiting spectral density of a matrix drawn from the general complex $\beta$-ensemble in the low temperature limit, consider the characteristic polynomial of the scaled matrix $D_n$ from Section~\ref{ref:lowtemp}. Specifically,
\begin{align}\label{eq: D_n rescaled}
		D_n=\begin{pmatrix}
			0& \sqrt{n-1}{\rm e}^{{\rm i}\phi_{n-1}} & & & \\
			\sqrt{n-1}{\rm e}^{{\rm i}\theta_{n-1}}& 0 &\sqrt{n-2}{\rm e}^{{\rm i}\phi_{n-2}} & &\\
			& \ddots & \ddots &\ddots &\\
			& & \sqrt{2}{\rm e}^{{\rm i}\theta_2}& 0 & \sqrt{1}{\rm e}^{{\rm i}\phi_1}\\
			& & & \sqrt{1}{\rm e}^{{\rm i}\theta_1} & 0
		\end{pmatrix},
	\end{align}
	where $\theta_\ell,\phi_\ell\sim \mathcal{U}(0,2\pi)$ are independent and identically distributed random variables. To apply Theorem~\ref{thm-Barbarino Noferini} we incorporate a factor of $\sqrt{2}$ in front of $D_n$.\footnote{In particular, we need a factor of $\sqrt{2}$ in front of each random variable of the form ${\rm e}^{{\rm i}\theta}$, with $\theta\sim \mathcal{U}(0,2\pi)$.} This is because after applying the low temperature limit, where we effectively performed the one-dimensional limit of the chi-distribution while keeping the angular component fixed, we are left with a one-dimensional random variable in each sub- and super-diagonal entry. However, due to their explicit form, these lead to complex-valued entries. Thus, we effectively have a one-dimensional random variable, in complex representation as a phase. According to the methodology of the authors of \cite{Barbarino-Noferini}, this should be interpreted as a one-dimensional random variable multiplied by a complex coefficient. However, we cannot do this because the real and imaginary parts of the variable ${\rm e}^{{\rm i}\theta}$, with $\theta\sim \mathcal{U}(0,2\pi)$, are not independent. Therefore, we interpret the matrix entries as complex random variables, whose real and imaginary parts contribute equally, i.e., each contributes to half of the total variance. By multiplying with the factor $\sqrt{2}$, we take this into account.

 Let $\mathcal{P}_n^D(z)$ denote the characteristic polynomial of $D_n$ with $b_\ell c_\ell=\ell {\rm e}^{{\rm i}(\theta_\ell+\phi_\ell)}$, $\ell=1,\dots,n-1$.
 \begin{Lemma}\label{lemma-var coeff P D}
	The coefficients $\kappa_\ell^{(n,D)}$ of the polynomial $\mathcal{P}_n^D(z)$ are independent, have zero mean and variance given by
	\begin{align}\label{eq-var coeff P D}
		\Var\bigl(\kappa_\ell^{(n,D)}\bigr)= \sum_{\gamma_1=1}^{n-1}\gamma_1^2\sum_{\gamma_2=1}^{\gamma_1-2}\gamma_2^2\cdots \sum_{\gamma_\ell=1}^{\gamma_{\ell-1}-2}\gamma_\ell^2
	\end{align}
	for all $\ell=1,2,\dots,\bigl\lfloor \frac{n}{2}\bigr\rfloor$.
 \end{Lemma}
 \begin{proof}
 We utilise Lemma~\ref{lemma-mean and variance kappa's} with $b_j
 c _j=j{\rm e}^{{\rm i}(\theta_j+\phi_j)}$ for $j=1,\dots,n-1$, where $\theta_j\sim \mathcal{U}(0,2\pi)$. First, we check whether these random variables have zero mean. We study a complex random variable of the form $Z:=a{\rm e}^{{\rm i}m\theta}$, with $a\in\mathbb{R}$, $\theta\sim \mathcal{U}(0,2\pi)$ and $m\in\mathbb{N}$. We assume that $m=1$ without loss of generality as the exponential function is periodic. The mean of $Z$ vanishes since
 \begin{equation*}
 \mathbb{E}[Z]=\frac{a}{2\pi}\int_{0}^{2\pi}{\rm e}^{{\rm i}\theta}\mathrm{d}\theta=0.
 \end{equation*}
 Next, note the variance of a complex random variable $W$ is defined by
			\begin{align*}
				\Var(W)=\mathbb{E}\bigl[|W|^2\bigr]-|\mathbb{E}[W]|^2.
			\end{align*}
Applying this definition to the random variable $Z$, we find
			\begin{align*}
				\Var(Z)&=\mathbb{E}\bigl[\bigl|a{\rm e}^{{\rm i}\theta}\bigr|^2\bigr]-\bigr|\mathbb{E}\bigl[a{\rm e}^{{\rm i}\theta}\bigr]\bigr|^2=a^2.
		\end{align*}
Hence, we have $\Var(b_jc_j)=j^2$ for $j=1,\dots,n-1$. Utilising Lemma~\ref{lemma-mean and variance kappa's} allows us to conclude the result stated in equation~\eqref{eq-var coeff P D} as desired.
 \end{proof}

As established in the previous sections, the random coefficients of the characteristic polynomial $\mathcal{P}_n^D(z)$ satisfy the necessary conditions for applying Theorem~\ref{thm-Barbarino Noferini}. To simplify the polynomial $Q_n^D(x)$, constructed from the coefficients \smash{$\bigl|\alpha_\ell^{(n)}\bigr|^2=\Var\bigl(\kappa_\ell^{(n,D)}\bigr)$}, we analyse these polynomials using generating functions. Consequently, we state the following lemma.
 \begin{Lemma}\label{lemma-Q D polynomials with var kappa}
 The polynomials
 \begin{align*}
 Q_n^D(x)=\sum_{\ell=0}^{\lfloor \frac{n}{2}\rfloor}\Var\bigl(\kappa_\ell^{(n,D)}\bigr)x^{n-2\ell},
 \end{align*}
 with $\Var\bigl(\kappa_\ell^{(n,D)}\bigr)$ given by equation~\eqref{eq-var coeff P D} for $\ell=1,\dots,\left\lfloor \frac{n}{2}\right\rfloor$ and $\Var\bigl(\kappa_0^{(n,D)}\bigr)=1$, satisfy the following three-term recurrence relation:
 \begin{align}\label{eq-3 term recurrence Q D}
 Q_{n+1}^D(x)=xQ_n^D(x)+n^2Q_{n-1}^D(x), \qquad n\geq1.
 \end{align}
 \end{Lemma}
 \begin{proof}
 From the variance at $\ell=0$, we have $Q_0^D(x)=1$ and $Q_1^D(x)=x$. The rest
 follows from Proposition~\ref{prop-formula coefficients P} with $\Var(-b_\ell c_\ell)=\ell^2$ for $\ell=1,\dots,n-1$.
 \end{proof}

 \begin{Proposition}\label{prop-explicit form Q H}
 We find the following rewriting of the polynomials $Q_n^D(x)$:
 \begin{align*}
 Q_n^D(x)=\sum_{k=0}^n\frac{(-1)^{n-k}(n!)^22^k}{(k!)^2(n-k)!}\left(\frac{1}{2}+\frac{x}{2}\right)_k,
 \end{align*}
 where $(a)_b:=\Gamma(a+b)/\Gamma(a)$ denotes the Pochhammer symbol.
 \end{Proposition}
 \begin{proof}
	We use the fact that the polynomials $Q_n^D(x)$ satisfy the three-term recurrence relation~\eqref{eq-3 term recurrence Q D} and study them via the method of generating functions. The latter is defined by
	\begin{align}\label{eq-generating function Q D}
		q(z,x)=\sum_{n=0}^\infty \frac{Q_n^D(x)}{(n!)^2}z^n.
	\end{align}
	We study a differential equation for $q(z,x)$, with differential operator $D_z:=\partial_z z \partial_z=z\partial_z^2+\partial_z$,
	\begin{align*}
		D_zq(z,x) =x\sum_{n=0}^\infty \frac{Q_n^D(x)}{n!n!}z^n+\sum_{n=0}^\infty \frac{Q_{n-1}^D(x)}{(n-1)!(n-1)!}z^n
		 =xq(z,x)+z q(z,x).
	\end{align*}
	The definition
	$q(z,x)={\rm e}^{-z}r(z,x)$ leads to
	\begin{align*}
		z\partial_z^2 r(z,x)+(1-2z)\partial_z r(z,x)-(x+1)r(z,x)=0,
	\end{align*}
	which is Kummer's differential equation in variable $2z$, cf. \cite{NIST}. We can thus immediately write down the two independent solutions
		\begin{align*}
			q(z,x)=c_1{\rm e}^{-z}U\left(\frac{1}{2}+\frac{x}{2},1;2z\right)+c_2{\rm e}^{-z}
M\left(\frac{1}{2}+\frac{x}{2},1;2z\right)	,
		\end{align*}
		where $U(a,b;y)$ denotes the confluent or Kummer hypergeometric function and $M(a,b;y)$ Tricomi's hypergeometric function.
The initial conditions $q(0,x)=1$ and $\partial_zq(z,x)|_{z=0}=x$ follow from the definition \eqref{eq-generating function Q D} and the two lowest order polynomials which determine the constants~$c_1$,~$c_2$ in the solution of the differential equation.
Because $U(a,1,z)$ has a logarithmic singularity for $z\rightarrow 0$ \cite[equation~(13.2.19)]{NIST}, we have to set $c_1=0$ and $c_2=1$, as $M(a,b;y)={}_1F_1(a;b;y)$ is unity at the origin. The second boundary condition can be easily checked from the derivative of
$M(a,b;y)$ at the origin using~\cite{NIST}.
Finally, the explicit form \eqref{prop-explicit form Q H} follows from writing out the hypergeometric function ${}_1F_1(a,b;z)$ as a series and using the Cauchy product
	\begin{align*}
		{\rm e}^{-z}{}_1F_1\left(\frac{1}{2}+\frac{x}{2},1;2z\right) =\sum_{k=0}^\infty \frac{(-1)^kz^k}{k!}\sum_{n=0}^\infty \frac{(a)_n2^n}{n!n!}z^n
		 =\sum_{n=0}^\infty \sum_{k=0}^n (-1)^{n-k}\frac{(a)_k2^{k}}{k!k!(n-k)!}z^n,
	\end{align*}
	with $a=\frac{1}{2}+\frac{x}{2}$.
	\end{proof}

We observe that the polynomials $Q_n^D(x)$ in \eqref{prop-explicit form Q H} are the \textit{Meixner--Pollaczek}-polynomials, denoted by \smash{$P_n^{(\lambda)}(y;\phi)$} in its hypergeometric representation \cite[Section~9.7]{Hypergeometric-Orthogonal-Polynomials}
\begin{align*}
 P_n^{(\lambda)}(y;\phi)=\frac{(\lambda)_n}{n!}{\rm e}^{-{\rm i}n\phi}{}_2F_1\bigl( -n, \lambda+{\rm i}y;2\lambda;1-{\rm e}^{-2{\rm i}\phi}\bigr).
\end{align*}
We obtain the polynomials $Q_n^D(y)$ with the choice $\lambda=\frac{1}{2}$, $\phi=-\frac{\pi}{2}$ and $y=-\frac{{\rm i}}{2}$. The Meixner--Pollaczek polynomials appeared already in another context of random matrices, namely the study of moments and hypergeometric orthogonal polynomials in \cite{Cunden-et-al}.

Combining all results, we are now in the position to compute the equilibrium eigenvalue density \eqref{eq-limiting density T,S} of eigenvalues of the general complex tridiagonal matrix \smash{$T_\beta$} in the low temperature.

\subsubsection{Proof of Theorem~\ref{thm-density beta}, equation~(\ref{eq-limiting density T,S})}
We derived Lemma~\ref{lemma-var coeff P D} for the variances of the coefficients of the characteristic polynomial of matrix $D_n$. To apply Theorem~\ref{thm-Barbarino Noferini}, we observe that we found an expression for the characteristic polynomial of~$D_n$, such that the coefficients of this polynomial in $z$ are independent, have zero mean and the variance is given by~\eqref{eq-var coeff P D}. Moreover, we proved in Lemma~\ref{lemma-Q D polynomials with var kappa} that the expressions for the variances of the coefficients $\kappa_k^{(n,D)}$ are the coefficients of the polynomial $Q_n^D(x)$. We define the parameters $\alpha_k^{(n)}$ in Theorem~\ref{thm-Barbarino Noferini} by the square root of the variances \smash{$\sqrt{\vphantom{\big|}\smash{\Var\bigl(\kappa_k^{(n,D)}\bigr)}}$} and find that we have to study those polynomial, whose coefficients are the variances of the coefficients of the characteristic polynomial of $D_n$ as a polynomial in~$|z|^2$, i.e., we study $Q_n^D\bigl(|z|^2\bigr)$. To find the corresponding limit for the matrix \smash{$T_\beta$}, we have to take the scaling due to Theorem~\ref{thm:largebetagen} into account.
	Recall that we denote the characteristic polynomial of the matrix \smash{$T_\beta$} by
	\begin{align*}
		\Pc^T_n(z)=\det(z\mathbf{1}_n-T_\beta),
	\end{align*}
whereas the characteristic polynomial of $D_n$ from equation~\eqref{eq:largebetasgen} is defined as
\begin{align*}
	\Pc^D_n(z)=\det(z\mathbf{1}_n-D_n).
\end{align*}
We conclude that these two characteristic polynomials are related due to Theorem~\ref{thm:largebetagen} as
\begin{align}
	\Pc^T_n(z)&=\prod_{i=1}^n(z-\lambda_i(T_\beta))
	 =\prod_{i=1}^n\left(z-\frac{1}{\sqrt{2n}}\lambda_i(D_n)+o\left(\frac{1}{\sqrt{\beta}}\right)\right)\nonumber\\
	&=\frac{1}{(2n)^{\frac{n}{2}}}\Pc^D_n(\sqrt{2n}z) +o\left(\frac{1}{\sqrt{\beta}}\right).
	\label{eq-charac poly with beta correction}
\end{align}
We will only take the leading order of the low temperature limit $\beta \gg 1$ into account and comment on possible corrections later. Recall that for calculating the limiting density we have to take the variance of the coefficients of the characteristic polynomials into account which yields the following result due to Proposition~\ref{prop-explicit form Q H} and Theorem~\ref{thm-Barbarino Noferini}:
\begin{align*}
	U(z)&=\lim_{n\rightarrow \infty}\frac{1}{2n}\ln \left(\frac{1}{(2n)^{n}}Q_n^D\left((2n)|z|^2+o\left( \frac{1}{\beta}\right)\right)\right)\\
	&=\lim_{n\rightarrow \infty}\frac{1}{2n}\ln \left(\frac{1}{(2n)^{n}}\sum_{k=0}^n\frac{(-1)^{n-k}(n!)^22^k}{(k!)^2(n-k)!}\left(\frac{1}{2}+n|z|^2\right)_k+o\left( \frac{1}{\beta}\right)\right)\\
	&=\lim_{n\rightarrow \infty}\frac{1}{2n}\ln \left((-1)^n \frac{n!}{(2n)^n}\, {}_2F_1\left(-n,\frac{1}{2}+n|z|^2;1;2\right)+o\left( \frac{1}{\beta}\right)\right). \label{eq-expression via 2F1 symm model}
\end{align*}
In the second step, we have replaced the sum by Gauss' hypergeometric function which truncates. Notice that the prefactor in \eqref{eq-charac poly with beta correction} gets squared as we now consider the polynomial $Q_n^D\bigl((2n)|z|^2\bigr)$ of the variances.
To compute the limit of the hypergeometric function, we assume that we can first expand in the first argument as it diverges faster than the second argument, because we have $|z|<1$ in the corresponding scaling limit. We apply the large $n$ asymptotics given in \cite[equation~15.7.3]{Abramowitz-Stegun} with $\arg(-2n)=\pi$ for $n\in \mathbb{N}$:
	\begin{align*}
		{}_2F_1\left(-n,\frac{1}{2}+n|{z}|^2;1;2\right)
		=&\frac{{\rm i}{\rm e}^{{\rm i}\pi n|{z}|^2}}{\Gamma\bigl(\frac{1}{2}-n|{z}|^2\bigr)}(-2n)^{-\frac{1}{2}-n|{z}|^2}
		\left(1+\mathcal{O}\left(\frac{1}{2n}\right)\right)\\
		&+\frac{1}{\Gamma\bigl(\frac{1}{2}+n|{z}|^2\bigr)}{\rm e}^{-2n}(-2n)^{-\frac{1}{2}+n|{z}|^2}\left(1+\mathcal{O}\left(\frac{1}{2n}\right)\right).
	\end{align*}
Using the Stirling formula to approximate $\Gamma(x)$ for large $x$ as \cite[equation~(8.327)]{gradshteyn}
	\begin{equation*}
		\Gamma(x)\sim x^{x-1/2}{\rm e}^{-x}\sqrt{2\pi}\left(1+\mathcal{O}\left(\frac{1}{x}\right)\right),
	\end{equation*}
we	conclude
	\begin{align}
		{}_2F_1\left(-n,\frac{1}{2}+n|z|^2;1;2\right)={}& \frac{{\rm i}}{\sqrt{2\pi}}{\rm e}^{{\rm i}\pi n|z|^2}(-n|z|^2)^{n|z|^2}{\rm e}^{-n|z|^2}
		(-2n)^{-\frac{1}{2}-n|z|^2}
		\left(1+\mathcal{O}\left(\frac{1}{2n}\right)\right)\nonumber\\
		&+\frac{1}{\sqrt{2\pi}}(n|z|^2)^{-n|z|^2}{\rm e}^{n|z|^2}{\rm e}^{-2n} (-2n)^{-\frac{1}{2}+n|z|^2}\left(1+\mathcal{O}\left(\frac{1}{2n}\right)\right)\nonumber\\
		\approx{}& \frac{{\rm i}}{\sqrt{2\pi}}{\rm e}^{n|z|^2(i\pi-1)}(|z|^2/2)^{n|z|^2}(-2n)^{-\frac12}
\nonumber\\
		&+\frac{1}{\sqrt{2\pi}}{\rm e}^{n(-2+|z|^2)}(-2/{|z|^2})^{n|z|^2}(-2n)^{-\frac12}. \label{eq:hypergeomsimp}
	\end{align}
	To calculate the limit $U(z)$, we have to find an expression for the logarithm of this approximation. We use the following rewriting of the logarithm of the sum of two complex variables $a$, $b$ with $\frac{|b|}{|a|}\gg 1$:
	\begin{align}\label{eq:logsimp}
		\ln(a+b)&=\ln(|a+b|)+{\rm i}\arg(a+b)
 \approx\ln |b|+{\rm i}\arg(a+b)
	\end{align}
	for \smash{$\frac{|b|}{|a|}\gg 1$}, where $\arg(a+b)\in [0,2\pi)$ is the argument of the complex variable $a+b$. It vanishes in the limit $n\rightarrow \infty$ because of the factor $1/(2n)$. We proceed by implementing the approximation of the hypergeometric function into the logarithm and we set
	\begin{align}\label{def:aandb}
		&a=\frac{{\rm i}}{\sqrt{2\pi}}{\rm e}^{n|z|^2({\rm i}\pi-1)}2^{-n|z|^2}\bigl(|z|^2\bigr)^{n|z|^2}\frac{1}{\sqrt{-2n}},\\
		&b=\frac{1}{\sqrt{2\pi}}{\rm e}^{n(-2+|z|^2)}\bigl(|z|^2\bigr)^{-n|z|^2}2^{n|z|^2}\frac{1}{\sqrt{-2n}},\nonumber
	\end{align}
	and observe that
	\begin{align}\label{eq:ratio}
		\frac{|a|}{|b|}={\rm e}^{-2n|z|^2+2n}\bigl(|z|^2/2\bigr)^{2n|z|^2}\ll 1 \qquad \text{for } n\rightarrow \infty.
	\end{align}
Note that the value of $\bigl|\frac{a}{b}\bigr|$ is not always close to zero for all values of $z\in\mathbb{\mathbb{C}}$ as $n\to\infty$. This is because the exponent in equation~\eqref{eq:ratio} can be positive or negative, depending on the value of~$z$, for example, when $z$ is close to zero.

For values of $z$ where the exponent is positive, we instead consider $\left|\frac{b}{a}\right|$, which decays to zero as $n\to\infty$. In this case, we rewrite equation~\eqref{eq:logsimp} as
\begin{align}\label{eq:ratio1}
		\ln(a+b)&\approx\ln(|a|)+{\rm i}\arg(a+b),
\end{align}
and carry on with the proof using this approximation.

To proceed, we therefore need to consider two cases, depending on the specific value of $z$, since the sign of the exponent in equation~\eqref{eq:ratio} depends on~$z$. This will determine whether \smash{$\bigl|\frac{a}{b}\bigr|$} or \smash{$\bigl|\frac{b}{a}\bigr|$} decays to zero. These two cases are:
\begin{enumerate}\itemsep=0pt
 \item For values of $z$ where the exponent in equation~\eqref{eq:ratio} is negative, $\left|\frac{a}{b}\right|\ll1$ as $n\to\infty$. Therefore, we use equation~\eqref{eq:ratio} to approximate the logarithm of equation~\eqref{eq:hypergeomsimp}.
 \item For values of $z$ where the exponent in equation~\eqref{eq:ratio} is positive, $\left|\frac{b}{a}\right|\ll1$ as $n\to\infty$. Therefore, we use equation~\eqref{eq:ratio1} to approximate the logarithm of equation~\eqref{eq:hypergeomsimp}.
\end{enumerate}
However, in both cases, equation~\eqref{eq:U(Z)unsimp} will lead to the same spectral density in equation~\eqref{eq:final spec dens}, regardless of whether we take $a$ or $b$ in equation~\eqref{def:aandb}. Thus, for the remainder of the proof, we focus on the case described in equation~\eqref{eq:ratio}, substituting equation~\eqref{eq:logsimp} into equation~\eqref{eq:U(Z)unsimp}. The proof for the other case follows analogously and is therefore omitted. Nevertheless, more care must be taken because equations~\eqref{eq:M def} and~\eqref{eq:eta def} are slightly different, as the signs of $m_1$ and $m_2$ will be flipped.

Next, we approximate the prefactor of the hypergeometric function using the Stirling formula again,
	\begin{align*}
		\frac{(-1)^nn!}{(2n)^n}\approx
		&\sqrt{2\pi n}\frac{(-1)^n{\rm e}^{-n}}{2^n}.
	\end{align*}
	Bringing all together, we find
	\begin{align}
		U(z)&=\lim_{n\rightarrow \infty}\frac{1}{2n}\ln\left(\frac{(-1)^nn!}{(2n)^n}{}_2F_1\left(-n,\frac{1}{2}+n|z|^2;1;2\right)\right)\nonumber\\
		&\approx \lim_{n\rightarrow \infty}\frac{1}{2n}\ln\bigl((-1)^{n}{\rm e}^{n(-3+|z|^2)}2^{n(|z|^2-1)}\bigl(|z|^2\bigr)^{-n|z|^2}\bigr)\nonumber\\
		&=m_0+m_1|z|^2+m_2|z|^2\ln\bigl(|z|^2\bigr),\label{eq:U(Z)unsimp}
	\end{align}
	where we obtain
	\begin{align}\label{eq:M def}
			m_0=\frac{{\rm i}\pi}{2}-\frac{3}{2}-\frac{1}{2}\ln(2),\qquad
			m_1=\frac{1}{2}+\frac{\ln(2)}{2},\qquad
			m_2=-\frac{1}{2}.
	\end{align}
	Calculating the Laplacian yields
	\begin{align*}
		\Delta_z U(z)\approx m_1+2m_2+m_2\ln\bigl(|z|^2\bigr).
	\end{align*}
We conclude that the equilibrium density of the general tridiagonal matrix ensemble is given by
\begin{align*}
	\widetilde{\rho}^T(z)=\eta_0 +\eta_1\ln\bigl(|z|^2\bigr),
\end{align*}
valid on its support to be determined, where the constants $\eta_0$, $\eta_1$ are defined by
\begin{align}\label{eq:eta def}
	\eta_0= m_1+2m_2=-\frac{1}{2}+\frac{\ln 2}{2},\qquad
	\eta_1= m_2=-\frac{1}{2}.
\end{align}
The function $\widetilde{\rho}^T(z)$ is a density as long as $\widetilde{\rho}^T(z)\geq 0$. We define the radius of the complex variable $z$ by $r:=|z|$, and determine the support also called droplet via the following condition
\begin{align*}
	\eta_0+\eta_1\log\bigl(r^2\bigr)\geq 0
	\ \Leftrightarrow \ 	r \leq \exp\left(-\frac{\eta_0}{2\eta_1}\right),
\end{align*}
where we used that $\eta_1<0$. We denote
\begin{align*}
	r_0:=\exp\left(-\frac{\eta_0}{2\eta_1}\right)=\sqrt{\frac{2}{e}}\approx 0.858.
\end{align*}
The support $[0,r_0]$ is determined by the ratio of the two constants $\eta_0$, $\eta_1$ and therefore independent of the normalisation.
We still have to normalise the density properly. Therefore, we define
\begin{align*}
A:= \int_{\mathbb{C}}\mathrm{d}^2z\, \widetilde{\rho}^T(z)=
2\pi	\int_0^{r_0} \mathrm{d}r\, r\widetilde{\rho}^T(z)
=2\pi	\int_0^{r_0} \mathrm{d}r\, r\bigl(\eta_0+\eta_1\ln\bigl(r^2\bigr)\bigr)
	 = -\pi\eta_1{\rm e}^{-\frac{\eta_0}{\eta_1}}=\frac{\pi}{{\rm e}}.
\end{align*}
Our final result is the following normalised equilibrium density
\begin{align}\label{eq:final spec dens}
	\rho^T(z) = \frac{{\rm e}}{\pi}\widetilde{\rho}^T(z)=
	\frac{{\rm e}^1}{2\pi}\bigl(\ln(2)-1-\ln\bigl(|z|^2\bigr)\bigr)
	 \approx -0.133-0.433\ln\bigl(|z|^2\bigr).
\end{align}

To prove that this result also holds for the symmetric ensemble $S$, we choose $\phi_j=\theta_j$ for all~$j$ in equation~\eqref{eq: D_n rescaled}. The result then follows a similar calculation to the one presented above.

\begin{Remark}[corrections in $1/\beta$]
 As we calculated the leading order of the equilibrium density in the low temperature limit $\beta \gg 1$, we will comment on the order of the sub-leading terms.
 We start from \eqref{eq-charac poly with beta correction} and perform the product as
	\begin{align*}
		\mathcal{P}_n^T(z) ={}&\frac{1}{(2n)^{\frac{n}{2}}}\mathcal{P}_n^D\bigl(\sqrt{2n}z\bigr)+o\left(\frac{1}{\sqrt{\beta}}\right)\left(\sum_{j=1}^n\prod_{k\neq j}\left(z-\frac{1}{\sqrt{2n}}\lambda_i(H_n)\right)\right)\\
		&+o\left(\frac{1}{\beta}\right)\left(\sum_{i=1}^n\sum_{j=1}^n\prod_{k\neq j,i}\left(z-\frac{1}{\sqrt{2n}}\lambda_i(H_n)\right)\right)+\cdots.
	\end{align*}
We observe that we can express the lower order terms in terms of derivatives of the characteristic polynomial $\mathcal{P}_n^D\bigl(\sqrt{2n}z\bigr)$ as
\begin{align}\label{eq-expansion in order of beta} \mathcal{P}_n^T(z)=\frac{1}{(2n)^{\frac{n}{2}}}\mathcal{P}_n^D\bigl(\sqrt{2n}z\bigr)+\frac{1}{(2n)^{\frac{n}{2}}}\sum_{i=1}^n o\left(\frac{1}{\beta^{\frac{i}{2}}}\right)\partial_z^i \mathcal{P}_n^D\bigl(\sqrt{2n}z\bigr).
\end{align}
We normalised the matrix \smash{$T_\beta$} by $1/\sqrt{2n\beta}$ such that the eigenvalues are supported in a compact disc. Therefore, we expect the characteristic polynomial $\mathcal{P}_n^T(z)$ to be bounded on the support of~\smash{$T_\beta$}. Hence, the quantity \smash{$\frac{1}{(2n)^{\frac{n}{2}}}\mathcal{P}_n^D\bigl(\sqrt{2n}z\bigr)$} is bounded, too. Furthermore, we use the fact that~the derivative of the characteristic polynomial \smash{$\mathcal{P}_n^D\bigl(\sqrt{2n}z\bigr)$} is given by the sum of the characteristic polynomials of all $((n-1)\times (n-1))$-dimensional sub-matrices of $D$. The latter are normalised with respect to a factor \smash{$1/(2n-1)^{\frac{n-1}{2}}$} (by an induction hypothesis, this normalisation is of the same order for all $n$). Notice that we get another factor of $\sqrt{2n}$ for each derivative as we act on the variable $z'=\sqrt{2n}z$.

It follows that the first correction reads
\begin{align*}
	o\left(\frac{1}{\sqrt{\beta}}\right)\frac{\sqrt{2n}}{(2n)^{\frac{n}{2}}}(2n-1)^{\frac{n-1}{2}}\times \text{norm.}=o\left(\frac{1}{\sqrt{\beta}}\right)\left(1-\frac{1}{2n}\right)^{\frac{n-1}{2}}\times \text{norm.},
\end{align*}
where ``norm.'' denotes the fixed norm of the characteristic polynomial of $D_n$. We observe that the leading order is independent of $n$. For the $i$-th order correction in $\beta$, we find
\begin{align*}
	o\left(\frac{1}{\beta^{\frac{i}{2}}}\right)\left(1-\frac{i}{2n}\right)^{\frac{n-i}{2}}\times \text{norm.}
\end{align*}
for $i=1,\dots,n$. For the calculation of the equilibrium density, we need the polynomial built by the variances of the coefficients of $\mathcal{P}_n^T(z)$. Therefore, we have to estimate the correction to each coefficient separately and calculate its variance. For $z^{n-1}$, we have only one contribution coming from the first correction term in \eqref{eq-expansion in order of beta}. Using the independence of the matrix elements of $D$, as we did previously, it is clear that the variance splits and we obtain a correction of order $o\bigl(\frac{1}{\beta}\bigr)$ for the second highest coefficient in~$z$. Hence, we conclude that
\begin{align*}
	U(z)=\lim_{n\rightarrow \infty}\frac{1}{2n}\ln \left(\frac{1}{(2n)^n}Q_n^D\bigl(2n|z|^2\bigr)+\frac{1}{(2n)^n}\sum_{i=1}^n o\left(\frac{1}{\beta^i}\right)Q_{n-i}^{D_i}\bigl(2n|z|^2\bigr)\right),
\end{align*}
where the $Q_{n-i}^{D_i}\bigl(2n|z|^2\bigr)$ is the polynomial of the variances of all $((n-i)\times (n-i))$-dimensional sub-matrices of $H$ for $i=1,\dots,n$. We estimate the leading order in $n$ by
\begin{align*}
\begin{split}
 U(z)&=\lim_{n\rightarrow \infty}\frac{1}{2n}\ln \left(\frac{1}{(2n)^n}Q_n^D\bigl(2n|z|^2\bigr)\left(1+\sum_{i=1}^n o\left( \frac{1}{\beta^i} \right) \right)\right)\\
 &=\lim_{n\rightarrow \infty}\frac{1}{2n}\ln \left(\frac{1}{(2n)^n}Q_n^D\bigl(2n|z|^2\bigr)\right)+\lim_{n\rightarrow \infty}\frac{1}{2n}\ln \left(1+\sum_{i=1}^n o\left( \frac{1}{\beta^i} \right) \right),
\end{split}
\end{align*}
where we find that the first correction term in orders of $\beta$ is massively suppressed in the large $n$-limit. This is the reason that the radial density of the general and the symmetric tridiagonal ensemble shows universality in $\beta$ in the large $n$-limit. This contrasts with the behaviour for $n=2$, which is discussed in Appendix~\ref{app-n=2}.

In the non-symmetric ensemble, we have that the coefficients of the polynomial $Q^G_n\bigl(|z|^2\bigr)$ still contain an explicit $\beta$-dependence (cf.\ Proposition~\ref{prop-explicit form Q G}) which makes the study of the correction terms much more evolved.
\end{Remark}

We conclude by determining the moments of our radially symmetric logarithmic density.
\begin{Lemma}
 The radial moments of the rotational symmetric density $\rho^{L}(z)$ for the general and symmetric ensemble $L=T,S$ are defined by
\begin{align}\label{eq-def rad moments}
 m_k:=\int_0^{r_0} \mathrm{d}r\, r^{2k+1}\rho^{L}(z)
\end{align}
for $k=1,2,\dots$ with $|z|=r$.
They are given as
\begin{align}\label{eq-rad moments}
 m_k=
 \frac{{\rm e}^1 }{4\pi }\frac{1}{(k+1)^2}r_0^{2k+2}= \frac{1}{2\pi(k+1)^2}\left(\frac{2}{e}\right)^k.
\end{align}
\end{Lemma}
\begin{proof}
 Due to the radial symmetry of density $\rho^{L}(z)$ we study only the non-vanishing radial moments defined in equation~\eqref{eq-def rad moments}. We parametrise $|z|=r$ and insert equation~\eqref{eq-limiting density T,S}, which leads to \eqref{eq-rad moments} using elementary integrals.
\end{proof}

The moments seemingly do not correspond to any interesting combinatorial objects, unlike the semi-circle, which relates to the Catalan numbers.

\subsection{Non-symmetric ensemble}\label{Subsec:equilibrium density non symmetric}
To derive the limiting spectral density of the non-symmetric tridiagonal matrix ensemble in the low temperature limit $(\beta \gg 1)$, consider the characteristic polynomial of the scaled matrix $G_n$ from Section~\ref{Sec:low temperature expansion}. Specifically,
\begin{align*}
		G_n=\begin{pmatrix}
			0& \frac{1}{\sqrt{\beta}} & & & \\
			\sqrt{n-1}{\rm e}^{{\rm i}\theta_{n-1}}& 0 &\frac{1}{\sqrt{\beta}} & &\\
			& \ddots & \ddots &\ddots &\\
			& & \sqrt{2}{\rm e}^{{\rm i}\theta_2}& 0 & \frac{1}{\sqrt{\beta}}\\
			& & & \sqrt{1}{\rm e}^{{\rm i}\theta_1} & 0
		\end{pmatrix},
	\end{align*}
where $\theta_i\sim \mathcal{U}(0,2\pi)$. We introduced here a factor of $\sqrt{2}$ in front of $G_n$ compared with the notation in Section~\ref{Sec:low temperature expansion}, for the same reasons as discussed for the non-symmetric ensemble in Section~\ref{Sec:SpecDensUcon}. We now consider the polynomial
$\mathcal{P}_n^{G}(z)$, defined in analogy to equation~\eqref{eq: Characteristic Polynomial}, associated to the rescaled random matrix $G_n$.
 \begin{Lemma}\label{lemma-var coeff P G}
	The coefficients $\kappa_\ell^{(n,G)}$ of the polynomial $\mathcal{P}_n^G(z)$ are independent, have zero mean and variance given by
	\begin{align}\label{eq-var coeff P G}
		\Var\bigl(\kappa_\ell^{(n,G)}\bigr)= \frac{1}{\beta^\ell} \sum_{\gamma_1=1}^{n-1}\gamma_1\sum_{\gamma_2=1}^{\gamma_1-2}\gamma_2\cdots \sum_{\gamma_\ell=1}^{\gamma_{\ell-1}-2}\gamma_\ell
	\end{align}
for all $\ell=1,2,\dots,\bigl\lfloor \frac{n}{2}\bigr\rfloor$.
 \end{Lemma}
 \begin{proof}
 We have \smash{$b_j c_j=\sqrt{\frac{j}{\beta}}{\rm e}^{{\rm i}\theta_j}$} for $j=1,\dots,n-1$, with $\theta_j\sim \mathcal{U}(0,2\pi)$. The rest of the proof is completely in parallel to the one of Lemma~\ref{lemma-var coeff P D} for the general ensemble.
 \end{proof}

 As in the previous subsection, we conclude that we have all conditions of the random coefficients of the characteristic polynomial $\mathcal{P}^G_n(z)$ to apply Theorem~\ref{thm-Barbarino Noferini}. Therefore, we define%
 \begin{equation}\label{eq:alphanonsym}
 |\alpha_j^{(n)}|^2=\Var\bigl(\kappa_j^{(n,H)}\bigr).
 \end{equation}
 To simplify our calculations, we will evaluate the function $Q_n^{G}(x)$, defined in~\eqref{eq:Qfunc}, in terms of the Hermite polynomials.
 \begin{Lemma}\label{lemma-Q G polynomials with var kappa}
 The polynomials
 \begin{align*}
 Q_n^G(x)=\sum_{k=0}^{\lfloor \frac{n}{2}\rfloor}\Var\bigl(\kappa_k^{(n,G)}\bigr)x^{n-2k},
 \end{align*}
 with $\Var\bigl(\kappa_k^{(n,G)}\bigr)$ given by equation~\eqref{eq-var coeff P G} for $k=1,\dots,\bigl\lfloor \frac{n}{2}\bigr\rfloor$ and $\Var\bigl(\kappa_0^{(n,G)}\bigr)=1$, satisfy the following three-term recurrence relation:
 \begin{align*}
 Q_{n+1}^G(x)=xQ_n^G(x)+\frac{n}{\beta}Q_{n-1}^G(x).
 \end{align*}
 \end{Lemma}
 \begin{proof}
 This follows from Proposition~\ref{prop-formula coefficients P} with $\Var(-b_j c_j)=\frac{j}{\beta}$ for $j=1,\dots,n-1$.
 \end{proof}

 \begin{Proposition}\label{prop-explicit form Q G}
 The polynomial $Q_n^G(x)$ is given by
 \begin{align*}
 Q_n^G(x)=\sum_{j=0}^{\lfloor\frac{n}{2}\rfloor}\frac{n!}{j!(n-2j)!(2\beta)^j}\
 x^{n-2j}.
 \end{align*}
 \end{Proposition}
 \begin{proof}
We begin by simplifying the expression for the nested summation of
$\Var\bigl(\kappa_k^{(n,G)}\bigr)$ defined in equation~\eqref{eq-var coeff P G}.

Let $B_{n}$ denote an $n\times n$ Jacobi matrix defined as follows:
\[
B_n:=\frac{1}{\sqrt{2}}\begin{pmatrix}
0 & \sqrt{n-1} & & \\
\sqrt{n-1} & 0 & \ddots & \\
 & \ddots & \ddots & \sqrt{1} \\
 & & \sqrt{1} & 0
\end{pmatrix}.
 \]
Utilising the notation of equation~\eqref{eq-3 term recurrence}, the characteristic polynomial of $B_n$ is given by
 \begin{equation}\label{eq:CharacteristicB}
 \mathcal{P}_n^B(x)=\sum_{m=0}^{\lfloor\frac{n}{2}\rfloor}\kappa_m^{(n,B)}x^{n-2m},
 \end{equation}
where
 \begin{equation}\label{eq:closedformB}
 \kappa_m^{(n,B)}=\frac{(-1)^m}{2^m}\sum_{\gamma_1=1}^{n-1}\gamma_1\dots\sum_{\gamma_m=1}^{\gamma_{m-1}-2}\gamma_m.
 \end{equation}
Next, we note that defined in~\cite[equation~(18.7.11)]{NIST} the probabilist's Hermite polynomials $\operatorname{He}_n(x)$ satisfy the three-term recurrence relation
\begin{equation*}
 \operatorname{He}_{n+1}(x)=x\operatorname{He}_n(x)-n\operatorname{He}_{n-1}(x).
\end{equation*}
Identifying the recurrence coefficients $\tilde{b}_n$
in equation~\eqref{eq-3 term recurrence}
for Hermite polynomials and using Proposition~\ref{prop-formula coefficients P}, alongside the explicit form of the Hermite polynomials~\cite[equation~(18.5.13)]{NIST}, we obtain
\begin{align}\label{eq:ChebyshevPoly}
\operatorname{He}_n(x)=\sum_{m=0}^{\lfloor\frac{n}{2}\rfloor}\kappa_{m}^{(n,{\rm He})}x^{n-2m}=n!\sum_{m=0}^{\lfloor\frac{n}{2}\rfloor}\frac{(-1)^m}{m!(n-2m)!2^m} x^{n-2m}
\end{align}
for
\begin{align}\label{eq:ChebyshevCoeff}
 \kappa_{l}^{(n,{\rm He})}=(-1)^l\sum_{\gamma_1=1}^{n-1}\gamma_1\dots\sum_{\gamma_l=1}^{\gamma_{l-1}-2}\gamma_{l}.
\end{align}
Now, we can derive an expression for $\kappa_m^{(n,G)}$ by comparing the closed forms of the characteristic polynomial in~\eqref{eq:CharacteristicB} and the Hermite polynomial. Specifically, by equating the equations~\eqref{eq:CharacteristicB} and~\eqref{eq:closedformB} with~\eqref{eq:ChebyshevPoly} and~\eqref{eq:ChebyshevCoeff}, we find
 \[
 \frac{1}{2^{2m}}\frac{n!}{m!(n-2m)!}=\frac{1}{2^m}\sum_{\gamma_1=1}^{n-1} \gamma_1\sum_{\gamma_2=1}^{\gamma_1-2}\gamma_2\dots\sum_{\gamma_m=1}^{\gamma_{m-1}-2}\gamma_m.
 \]
 Thus, we can express $\Var\bigl(\kappa_l^{(n,G)}\bigr)$ as
 \begin{align*}
 \Var\bigl(\kappa_l^{(n,G)}\bigr)=\frac{1}{\beta^l}\sum_{\gamma_1=1}^{n-1}\gamma_1\sum_{\gamma_2=1}^{\gamma_1-2}\gamma_2\dots\sum_{\gamma_l=1}^{\gamma_{l-1}-2}\gamma_l=\frac{1}{2^{l}\beta^l}\frac{n!}{l!(n-2l)!},
 \end{align*}
 as desired.
 \end{proof}

Combining the results discussed above, we are now ready to compute the limiting spectral density of the eigenvalues of the non-symmetric tridiagonal matrix $\widetilde{T}_\beta$ in the low temperature limit.

\subsubsection{Proof of Theorem~\ref{thm-density beta}, equation~(\ref{eq-limiting density T tilde})}
\begin{proof}
We begin by defining the parameters $\alpha_k^{(n)}$ in Theorem~\ref{thm-Barbarino Noferini} using equation~\eqref{eq:alphanonsym}. By Lemma~\ref{lemma-var coeff P G} we know that these parameters are independent, have zero mean and variance \smash{$\sqrt{\vphantom{\big|}\smash{\Var\bigl(\kappa_k^{(n,G)}\bigr)}}$}. Thus, our polynomial $\mathcal{P}^G(z)$, as defined by equation~\eqref{eq-3 term recurrence}, takes the form outlined in~\eqref{eq-polynomial thm Barbarino Noferini}.

Next, we consider the polynomial $Q_n^G\bigl(|z|^2\bigr)$, which is associated with the matrix $G_n$. To determine the corresponding limit for the matrix \smash{$\widetilde{T}_\beta$}, we must also account for the scaling determined in Theorem~\ref{thm:largebetanonsym}.
Recall that the characteristic polynomial of the matrix \smash{$\widetilde{T}_\beta$} is denoted by \smash{$\mathcal{P}^{\widetilde{T}}_n(z)$} and that of $G_n$, denoted by $\mathcal{P}_n^G(z)$, is given by equation~\eqref{eq: Characrteristic Polynomail G}.
Therefore, we conclude that
\begin{align*}
\mathcal{P}^{\widetilde{T}}_n(z)&=\prod_{i=1}^n(z-\lambda_i(\widetilde{T}_\beta))
	 =\prod_{i=1}^n\left(z-\frac{1}{\sqrt{2n}}\lambda_i(G_n)+o\left(\frac{1}{\sqrt{\beta}}\right)\right)
	 \approx\frac{1}{(2n)^{\frac{n}{2}}}\mathcal{P}^G_n\bigl(\sqrt{2n}z\bigr).
\end{align*}
in distribution as $\beta\rightarrow\infty$. Utilising Proposition~\ref{prop-explicit form Q G}, we establish the following:
\begin{align}
\frac{1}{(2n)^{n}}Q_n^G\bigl(2n|z|^2\bigr)
	 =\sum_{j=0}^{\lfloor\frac{n}{2}\rfloor}\frac{n!}{j!(n-2j)!}\frac{|z|^{2(n-2j)}}{2^{3j}\beta^jn^{2j}}
 =
\frac{1}{\bigl(2{\rm i}n\sqrt{\beta}\bigr)^n}\operatorname{He}_{n}\bigl(2{\rm i}n\sqrt{\beta}|z|^2\bigr).
\label{eq:RelationtoHermite}
\end{align}
However, the difficulty is that the scaling parameter $n$ appears in both the argument and the coefficient of the Hermite polynomial in~\eqref{eq:RelationtoHermite}. Therefore, before taking the limit of the logarithm
we apply the multiplication theorem for Hermite polynomials~\cite[equation~(18.18.13)]{NIST}, applied to~\eqref{eq:RelationtoHermite} to produce
\begin{align*}
 \frac{1}{\bigl(2{\rm i}n\sqrt{\beta}\bigr)^n}\operatorname{He}_{n} \bigl(2{\rm i}n\sqrt{\beta}|z|^2\bigr)&=\sum_{j=0}^{\lfloor\frac{n}{2}\rfloor}\left(1+\frac{1}{4\beta n^2}\right)^j\frac{n!}{j!(n-2j)!2^j}\operatorname{He}_{n-2j}\bigl(|z|^2\bigr)\nonumber\\
 &\approx \sum_{j=0}^{\lfloor\frac{n}{2}\rfloor}\frac{n!}{j!(n-2j)!2^j}\operatorname{He}_{n-2j}\bigl(|z|^2\bigr).
\end{align*}
In the second step, we are neglecting higher order terms in $1/\beta$. For the latter sum, there is the following identity \cite[equation~(18.18.20)]{NIST}:
\begin{equation*}
|z|^{2n}=\sum_{j=0}^{\lfloor\frac{n}{2}\rfloor}\frac{n!}{j!(n-2j)!2^j}\operatorname{He}_{n-2j}\bigl(|z|^2\bigr),
\end{equation*}
and thus we arrive at the following approximation for the limit
\begin{equation*}
 U(z)\approx\lim_{n\rightarrow\infty}\frac{1}{2n}\log\bigl( |z|^{2n}\bigr)=\log|z|.
\end{equation*}
It is well known that the logarithm is the Green's function for the Laplacian in two dimensions, and thus we obtain from Theorem~\ref{thm-Barbarino Noferini}
\begin{align*}
 \widetilde{\rho}^{\widetilde{T}}(z)=\frac{1}{2\pi}\Delta U(z)=\frac 14\delta^{(2)}(z).
\end{align*}
 After appropriately normalising this is equation~\eqref{eq-limiting density T tilde}, in coordinates $z=x+{\rm i}y$.
\end{proof}

\section[Numerical simulations, pseudospectrum, and local nearest-neighbour spacing]{Numerical simulations, pseudospectrum,\\ and local nearest-neighbour spacing}\label{Sec:Num}

\begin{figure}[ht]\centering
\begin{subfigure}{0.46\textwidth}\centering
\includegraphics[width=1\linewidth]{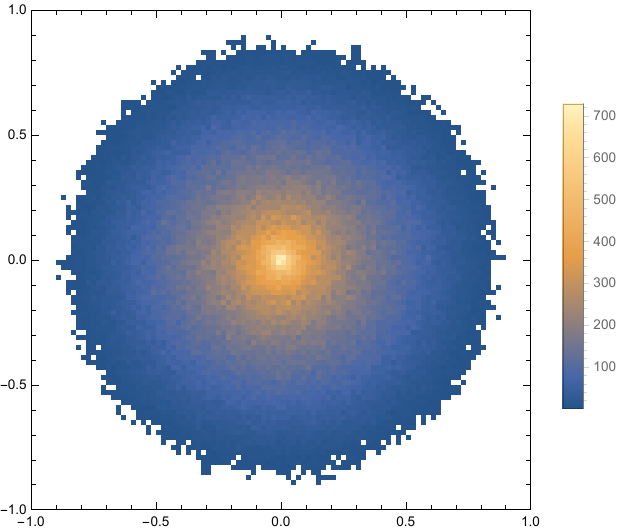}
\caption{$T_\beta$}
\end{subfigure}\quad
\begin{subfigure}{0.46\textwidth}\centering
\includegraphics[width=1\linewidth]{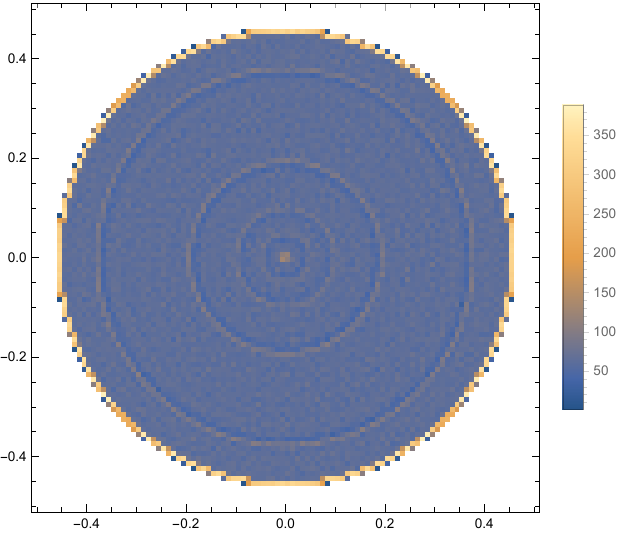}
\caption{$\widetilde{T}_\beta$}
\end{subfigure}
\caption{Spectral density of the complex eigenvalues of \smash{$T_\beta$} (left) and $\widetilde{T}_\beta$ (right), for $100$ matrices of dimension $n=5000$, both with $\beta=6$. The spectral density is normalised by $\sqrt{2n\beta}$.\label{fig-scatter plots tridiag models n5000}}
\end{figure}
In this section, we study the tridiagonal matrix ensembles numerically and compare the equilibrium densities for \smash{$T_\beta$} \eqref{eq-limiting density T,S} and for $\widetilde{T}_\beta$ \eqref{eq-limiting density T tilde} derived in the previous section in the limit of low temperature $\beta \gg 1$ and large matrices $n\rightarrow \infty$ with the numerics.

We begin with two scatter plots of the complex eigenvalues of the general tridiagonal random matrices \smash{$T_\beta$} and $\widetilde{T}_\beta$ in Figure~\ref{fig-scatter plots tridiag models n5000} left respectively right, for the same number of realisations, matrix size and $\beta=6$. In both cases the spectrum is rotationally invariant, as expected. For this reason, below we will only present plots for the radial density of the modulus $|z|$.
However, already for a single value of $\beta$ the ensembles show very different features.
While the spectrum of \smash{$T_\beta$} in Figure~\ref{fig-scatter plots tridiag models n5000} (left) is smooth, with a strong peak at the origin, that of $\widetilde{T}_\beta$ looks almost flat, with pronounced concentric rings, including at the origin and at the edge.

\begin{figure}[!ht]\centering
 \begin{subfigure}{0.46\textwidth}\centering
		\includegraphics[width=\textwidth]{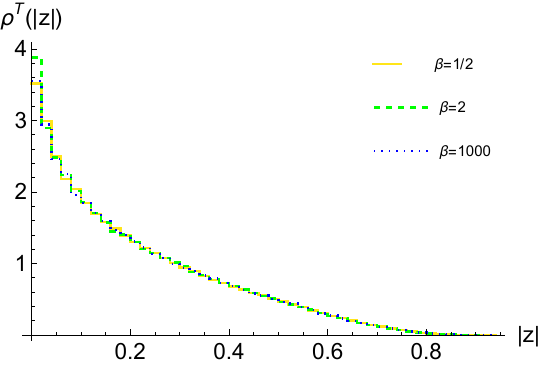}
		\subcaption{\smash{$T_\beta$} \label{fig-radial dist T n=5000}}	
	\end{subfigure}\quad
 \begin{subfigure}{0.46\textwidth}\centering
		\includegraphics[width=\textwidth]{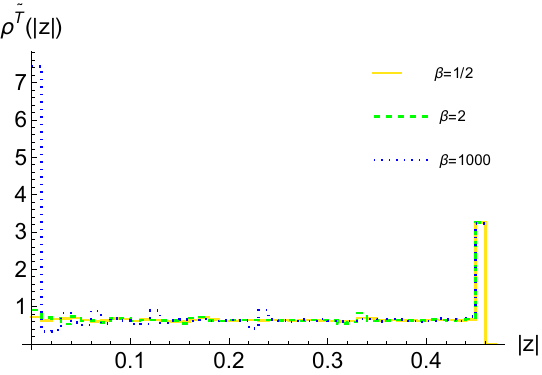}
		\subcaption{$\widetilde{T}_\beta$\label{fig-radial dist T tilde n=5000}}	
	\end{subfigure}
	\caption{Histograms of the radial density of the complex eigenvalues for $100$ matrices of size $n=5000$ of the ensembles \smash{$T_\beta$} (left) and $\widetilde{T}_\beta$ (right), at $\beta=1/2$ (yellow full curve), $\beta=2$ (green dashed curve) and $\beta=1000$ (magenta dotted curve).}\label{fig-radial dist T and T tilde beta 1/2,2,1000 n=5000}
\end{figure}

In Figure~\ref{fig-radial dist T and T tilde beta 1/2,2,1000 n=5000},
the radial density $\rho^L(|z|)$ of the complex eigenvalues $z$ is plotted for $L=T,\widetilde{T}$ for different values of $\beta=1/2,2,1000$ in different colours, for the same size and number of realisations as in Figure~\ref{fig-scatter plots tridiag models n5000}. For the general tridiagonal model \smash{$T_\beta$} (left), we observe that the histograms of the radii $|z|$ look very similar for the plotted $\beta$-values over a range of 3 orders of magnitude. Only near the origin small differences can be seen. We conclude that the spectral density of the general tridiagonal ensembles \smash{$T_\beta$} seems to be universal for a large range of $\beta$ in the limit of large matrix size $n$. This shows that although we are only able to compute the limiting spectral density for matrices \smash{$T_\beta$} in the limit of low temperature $\beta \gg 1$, due to universality the result will be a good approximation also for small $\beta$ down to $\beta=1/2$, when $n$ is large, too.

For the radial density of the ensemble $\widetilde{T}_\beta$, the situation is different.
We observe that all three curves for different values of $\beta=1/2,2,1000$ share a
flat density and a peak at the edge of the spectrum, which also appears to be universal.
However, at large $\beta$~-- which is the limit in which we derived the global density~-- a strong peak at the origin develops, in addition to small peaks within the spectrum.
We conclude for the non-symmetric tridiagonal model the equilibrium density \smash{$\rho^{\widetilde{T}}(z)=\delta(x)\delta(y)$} \eqref{eq-limiting density T tilde}
completely dominates the spectrum in the limit of large $\beta$ and~$n$, and that we cannot resolve the subleading, potentially universal features of the density in this limit.
The tails of the peak at the edge may be different for small to intermediate $\beta$ too, but such a local analysis is currently not feasible.
We will discuss in the next subsection what could be the reason of this singular behaviour, introducing the notation of a pseudo-spectrum, which plays an important role here.

\begin{figure}[!ht]
			\ \begin{minipage}{0.45\textwidth}
\centering					\includegraphics[width=\linewidth,angle=0]{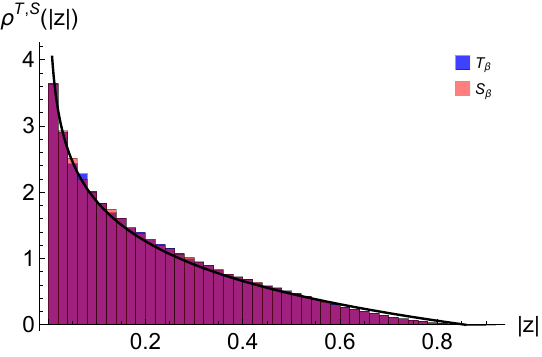}
				\end{minipage}\quad
			\begin{minipage}{0.50\textwidth}\small\centering
		\begin{tabular}{|c|l|c|c|c|c|}
		\hline
		&\diagbox{$n(m)$}{$\beta$} & $\frac{1}{2}$ & 2 & 100 & 1000 \\
		\hline
 \multirow{3}{*}{$T_\beta$}&100 (5000) & 7.60 & 4.57& 10.70& 11.73\\
		\cline{2-6}
		&1000 (1000) & 3.25 & 1.35 & 3.60& 4.02\\
		\cline{2-6}
		 &5000 (100) & 1.38 & 1.41& 1.39& 1.38\\ \hline \hline
		\multirow{3}{*}{$S_\beta$}&100 (5000) & 9.84 & 5.74&10.62 & 11.82\\
		\cline{2-6}
		&1000 (1000) & 4.10 & 2.24&3.59 & 4.15\\
		\cline{2-6}
		&5000 (100) & 2.22 & 1.38& 1.37& 1.34\\
	 \hline
		
	\end{tabular}
	\label{tab-KS dist T,S beta}
		\end{minipage}
		\caption[Mathematica Plot, \emph{Radial density with analytics}]{Comparison between analytics and numerics for \smash{$T_\beta$} and $S_\beta$. The histogram ({left}) shows the distribution of the radii for an ensemble of $100$ matrices of size $n=5000$ for the general ensemble \smash{$T_\beta$} (magenta) and symmetric ensemble $S_\beta$ (light red) for $\beta=100$. The region where both distributions agree is purple. The black curve gives the analytical expression for the density \eqref{eq-limiting density T,S}.
		We also give the Kolmogorov--Smirnov distances ({right}) in units $10^{-2}$ between the numerics and the analytical result~\eqref{eq-limiting density T,S} for both ensembles for various matrix sizes $n$, ensemble sizes $m$ and $\beta$ values.}\label{fig-comparison T,S beta}
	\end{figure}

 In Figure~\ref{fig-comparison T,S beta}, we give a visual (left) and quantitative comparison (right) between the radial distribution of the numerically generated general tridiagonal ensemble \smash{$T_\beta$} and symmetric ensemble $S_\beta$ at $\beta=100$, keeping the same matrix and ensembles size. The left plot shows coloured histograms in magenta and red, respectively, versus the
analytical solution \eqref{eq-limiting density T,S}, which we derived in the low temperature limit for $n\rightarrow \infty$ (magenta curve). We observe an excellent agreement within both matrix ensembles. Moreover, the distributions of the two models show only small departures from each other and from the limiting spacing density. To estimate the systematic error between the analytical result \eqref{eq-limiting density T,S} and the numerics of the two matrix ensembles \smash{$T_\beta$} and~$S_\beta$,
we calculate the Kolmogorov--Smirnov distance $d$ for 4 different values of $\beta$ in units of $10^{-2}$ in the table of Figure~\ref{fig-comparison T,S beta} (right).
The Kolmogorov--Smirnov distance $d$ between two distributions $f$ and $g$ is defined as
\begin{equation*}
	d=\max_{x\geq0}|F(x)-G(x)|\in[0,1],
\end{equation*}
given in terms of the respective cumulative distributions $F$ and $G$ of $f$ and $g$. The Kolmogorov--Smirnov distance $d$ is independent of the binning into histograms.

As a reference, we give the same measure of error for the constant spectral density of the complex Ginibre ensemble GinUE with the same ensemble size $n=5000$ and realisations ${m=100}$, which produces $d=0.17\cdot 10^{-2}$. Note that the range of error for the equilibrium density $\rho^{T,S}(z)$ of both ensembles is always comparable and shows the same systematic when increasing $\beta$, $n$ and $m$. We observe that the best agreement between the derived analytical result and numerics is given for large $n=5000$. However, the Kolmogorov--Smirnov distance is still by one order larger than for the same ensemble of Ginibre matrices.
For the largest $n$, where we see the best agreement in Kolmogorov--Smirnov distance, we can go to even larger $\beta= 5000, 7000, 50000$ and find that $d$ is not reduced further, i.e., we have $d=1.37, 1.39, 1.37$ in units $10^{-2}$, respectively.

\begin{Remark}
	As a technical remark, we would like to point out that for the calculation of the Kolmogorov--Smirnov distances, the data is normalised with respect to a flat measure, which gives a much simpler numerical computation of the cumulative distribution. Therefore, we compare with the limiting density, which is also normalised with respect to the flat measure, i.e.,\looseness=-1
	\begin{align*}
		1=\int_0^{r_0} \mathrm{d}r\, \rho^{S,\text{flat}}(z).
	\end{align*}
Since this only affects the overall scaling of the function, both for analytics and numerics, the calculated Kolmogorov--Smirnov distances still provide a good estimate of the systematic error.
\end{Remark}

\subsection[Contours of epsilon-pseudospectra for non-Hermitian tridiagonal matrices]{Contours of $\boldsymbol{\epsilon}$-pseudospectra for non-Hermitian tridiagonal matrices}

It is well known that the numerically determined eigenvalues of a non-normal matrix do not always provide a complete picture of the matrix properties, particularly its resonances and stability. In many cases, the spectrum alone is insufficient to capture the sensitivity of the matrix to perturbations. This limitation motivates the study of the pseudospectrum, a concept introduced in~\cite{Trefethen-Embree}. The pseudospectrum extends the analysis of the eigenvalue spectrum to incorporate the effects of perturbations and the behaviour of the resolvent.

The pseudospectrum can be formally defined as follows.

\begin{Definition}[\text{\cite[equation~(2.1)]{Trefethen-Embree}}]
	Let $A\in \mathbb{C}^{n\times n}$ and $\epsilon >0$. The \textit{$\epsilon$-pseudospectrum} $\Lambda_\epsilon(A)$ is the set of $z\in \mathbb{C}$ such that
		\begin{align*}
			\bigl\|(z\boldsymbol{1}_n-A)^{-1}\bigr\|>\epsilon^{-1}.
		\end{align*}
\end{Definition}	
The matrix $(z\boldsymbol{1}_n-A)^{-1}$ is known as the resolvent of $A$, and the spectrum $\sigma(A)$ consists of those $z\in\mathbb{{C}}$ where the resolvent does not exist (i.e., the eigenvalues of $A$). Following the convention in~\cite{Trefethen-Embree}, for $z\in\sigma(A)$, we set
\begin{align*}
	\bigl\|(z\boldsymbol{1}_n-A)^{-1}\bigr\|=\infty.
\end{align*}
An alternative definition of the $\epsilon$-pseudospectrum is given in terms of the $2$-norm and the smallest singular value of $(z\boldsymbol{1}_n-A)$, denoted by $s_{\min}(z\boldsymbol{1}_n-A)$. Specifically, when $\|\cdot \|=\|\cdot \|_2$, it holds that
\[
s_{\min}(z\boldsymbol{1}_n-A)=\frac{1}{\|(z\boldsymbol{1}_n-A)^{-1}\|_2}.
\]
\begin{Definition}\label{def:singpseudo}
	For $\|\cdot \|=\|\cdot \|_2$, the $\epsilon$-pseudospectrum $\Lambda_\epsilon(A)$ is the set of all $z\in\mathbb{C}$ such that
	\begin{align*}
		s_{\min}(z\boldsymbol{1}_n-A)=\inf_{\|\mathbf{x}\|_2=1}\|(z\boldsymbol{1}_n-A)\mathbf{x}\|_2< \epsilon.
	\end{align*}
\end{Definition}

For normal matrices (matrices which satisfy $AA^*=A^*A$), the behaviour of the resolvent is relatively straightforward: the $\epsilon$-pseudospectrum consists of the union of balls of radius $\epsilon$ centred at each eigenvalue in the complex plane. This reflects that the eigenvalues of normal matrices are stable under small perturbations, and the matrix properties can be understood purely from the spectrum.

However, for non-normal matrices, the situation is more complicated. Even when $z$ is far from the spectrum, the norm of the resolvent $\bigl\|(z\boldsymbol{1}_n-A)^{-1}\bigr\|$ can be large. This occurs because non-normal matrices can have eigenvalues that are extremely sensitive to small perturbations. As a result, the $\epsilon$-pseudospectrum of a non-normal matrix does not simply form $\epsilon$-balls around the eigenvalues. This behaviour is described in~\cite[Theorem~2.2]{Trefethen-Embree}.

A significant result that connects the spectrum and the pseudospectrum is the Bauer--Fike theorem, which provides bounds on how eigenvalues behave under perturbations.
\begin{Theorem}[Bauer--Fike theorem]\label{thm:Bauer-Fike}
 Let $A\in\mathbb{C}^{N\times N}$ be diagonalisable, such that $A=V\Lambda V^{-1}$. Then, for any $\epsilon>0$, with $\|\cdot \|=\|\cdot \|_2$,
 \begin{equation*}
 \sigma_A+\Delta_\epsilon\subseteq \Lambda_\epsilon(A)\subseteq \sigma_A+\Delta_{\epsilon\kappa(V)},
 \end{equation*}
 where
 \[
 \Delta_\epsilon=\{z\in\mathbb{C}\mid |z|<\epsilon\}\qquad \text{and}\qquad \kappa(V)=\frac{\|V\|}{\|V^{-1}\|}\ \text{is the condition number of}\ V.
 \]
\end{Theorem}
This result shows that for diagonalisable matrices, the pseudospectrum $\Lambda_\epsilon(A)$ is contained within a neighbourhood controlled by the condition number $\kappa(V)$ of the eigenvector matrix~$V$. In the case of non-normal matrices, where $\kappa(V)$ can be large (see Section~\ref{Sec:Measures of non-normality}), the pseudospectrum provides an enlarged region around the spectrum.

\begin{figure}[!ht]\centering
\begin{subfigure}{0.46\textwidth}\centering
		\includegraphics[width=\textwidth]{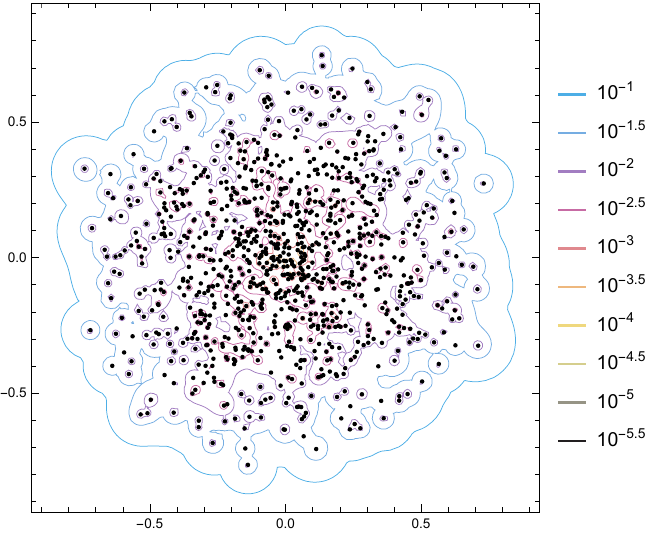}
		\subcaption{$T_\beta$ }	
	\end{subfigure}\quad
\begin{subfigure}{0.46\textwidth}\centering
		\includegraphics[width=\textwidth]{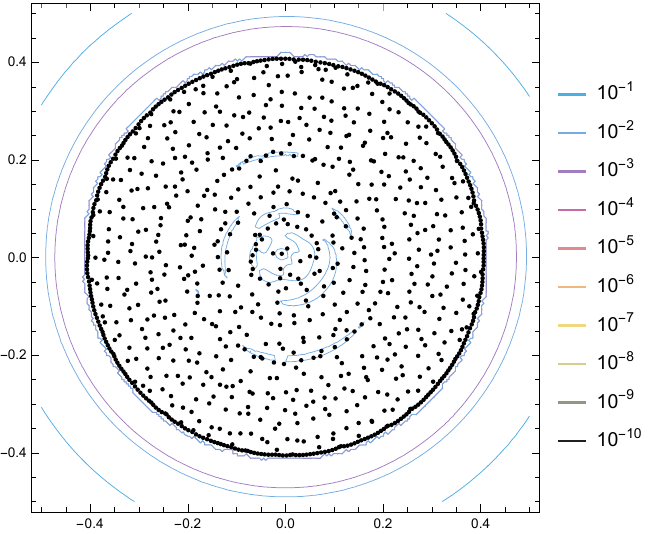}
		\subcaption{$\widetilde{T}_\beta$}
	\end{subfigure}
	\caption{The plots show the $\epsilon$-pseudospectra of $T_{\beta}$ (left) and $\widetilde{T}_{\beta}$ (right), for $n=1000$, compared to its eigenvalues (dots) for $\beta=2$. The contour scale represents $10^{-\epsilon}$.}
 \label{fig:pseudospectrum}
\end{figure}
\begin{Proposition}
	For arbitrary $\epsilon>0$, the $\epsilon$-pseudospectrum is given as
		\begin{gather}\label{eq:pseudo}
			\Lambda_\epsilon(L_\beta)=\left\{z\in \mathbb{C} \, \bigg| \left|z -\frac{a_1}{\sqrt{2n\beta}}\right|^2+\frac{|c_1|^2}{2n\beta}<\epsilon^2\right\}\quad \text{with} \ \begin{cases}
			 |c_1|\sim \chi_{\beta/2} & \text{for } L_\beta=T_\beta,\\
 |c_1|\sim \frac{1}{\sqrt{2}}\chi_{\beta} & \text{for } L_\beta=S_\beta,\\
 c_1=1& \text{for } L_\beta=\widetilde{T}_\beta
			\end{cases}
		\end{gather} and $\operatorname{Re}(a_1),\operatorname{Im}(a_1)\sim\mathcal{N}(0,1)$.

\end{Proposition}
\begin{proof}
First, we will use Definition~\ref{def:singpseudo} to compute the $\epsilon$-pseudospectrum of the general tridiagonal matrix $T_{\beta}$. Thus, to analyse the 2-norm of the vector $(zI-T_\beta)\mathbf{e}_n$ consider the equation
	\begin{align*}
		||(z\boldsymbol{1}_n-T_\beta)\mathbf{e_n}||_2&=\left\lVert \begin{pmatrix}
			0,&
			\dots,&
			0,&
			\dfrac{c_1}{\sqrt{2n\beta}},&
			z-\dfrac{a_1}{\sqrt{2n\beta}}
		\end{pmatrix}^{\rm t} \right\rVert_2.
	\end{align*}
	This yields
	\begin{equation*} ||(z\boldsymbol{1}_n-T_\beta)\mathbf{e_n}||_2=\sqrt{\frac{|c_1|^2}{2n\beta}+|z|^2-\frac{2|za_1^*|}{\sqrt{2n\beta}}\cos(\arg(za_1^*))+\frac{|a_1|^2}{2n\beta}}.
	\end{equation*}
	Next, we must establish that the vector $\mathbf{e}_n$ minimises this expression over all unit vectors \mbox{$\mathbf{x}\in\mathbb{C}^n$}. To do so we must use the Cauchy-Schwartz inequality since $\Var(|c_j|)\geq\Var(|c_1|)$ for $j=2,\dots,n-1$ when $|c_w|\sim \chi_{\beta w/2}$ for general tridiagonal matrices \smash{$T_\beta$} and $|c_w|\sim \frac{1}{\sqrt{2}}\chi_{\beta w}$ for symmetric tridiagonal matrices $S_\beta$.
	Note that the entries on the sub- and super-diagonals are of the same order for the general tridiagonal ensemble.
	
	Assuming we have a vector $\mathbf{x}$ with $\|\mathbf{x}\|_2=1$ one obtains
\begin{align*}
		||(z\boldsymbol{1}_n-T_\beta)\mathbf{x}||_2
		&\geq \sqrt{\left|\left(z-\frac{a_1}{\sqrt{2n\beta}}\right)x_n\right|^2+2\operatorname{Re} \left(\frac{c_1}{\sqrt{2n\beta}}x_{n-1}\left(\frac{a_1^*}{\sqrt{2n\beta}}-z^*\right)x_n^*\right)}.
	\end{align*}
	We can choose a representative, such that the variable $\operatorname{Re} \bigl(\frac{c_1}{\sqrt{2n\beta}}x_{n-1}\bigl(\frac{a_1^*}{\sqrt{2n\beta}}-z^*\bigr)x_n^*\bigr)$ is positive, without loss of generality. Therefore,
	\begin{align*}
		||(z\boldsymbol{1}_n-T_\beta)\mathbf{x}||_2&
		\geq \sqrt{\left|\left(z-\frac{a_1}{\sqrt{2n\beta}}\right)x_n\right|^2},
	\end{align*}
which implies that $\mathbf{e}_n$ is the infimum. For the non-symmetric ensemble, we have $c_j=1$ for $j=1,\dots,n-1$, which is independent of $\beta$.
\end{proof}

As $n\rightarrow\infty$ or $\beta\rightarrow\infty$, the $\epsilon$-pseudospectrum converges to $\epsilon$-balls centred at the origin for all three matrix ensembles. However, for small $n$ and $\beta$ the randomness displayed in $|c_1|$ altering the radius of the contours results in a significantly different situation for the general and symmetric tridiagonal ensemble, cf.\ equation~\eqref{eq:pseudo}.

\subsection{Measures of non-normality}\label{Sec:Measures of non-normality}
As discussed in the previous section, the $\epsilon$-pseudospectrum becomes more important as the degree of non-normality grows. However, the latter is too complex to be determined by one parameter completely. There are several possible measures of non-normality, which are discussed in \cite[Section~48]{Trefethen-Embree}. To show that the degree of non-normality is much larger for the non-symmetric tridiagonal matrices than for general/symmetric tridiagonal matrices, we consider the condition number of the diagonalising matrix as a measure. This quantity can be easily derived numerically.

As previously, we assume that the matrix $L_\beta=T_\beta, S_\beta, \widetilde{T}_\beta$, can be diagonalised by matrix~$R_L$~as
\begin{align*}
	L_\beta=R_L\Lambda_L R_L^{-1},
\end{align*}
where $\Lambda_L$ is a diagonal matrix containing the complex eigenvalues of $L_\beta$. The \textit{condition number} of the diagonalising matrix $R_L$ (which contains the eigenvectors) is also defined as
\begin{align*}
	\kappa(R_L)=\frac{s_{\max}(R_L)}{s_{\min}(R_L)},
\end{align*}
where $s_{\max}(R_L)$ and $s_{\min}(R_L)$ are the maximal and minimal singular value of $R_L$. Recall, that the diagonalising matrix $R_L$ is not determined uniquely and one has to fix a representative. The condition number is bounded by $1<\kappa(R_L)<\infty$, where the condition number $\kappa(R_L)=1$ is possible if and only if the matrix $L_\beta$ is normal. Therefore, we will give the condition number of the diagonalising matrix as a measure of non-normality for matrix ensembles of $m$ matrices of size $n=100,500,1000$ and values $\beta=1/2,2,100,1000$ each in the Table~\ref{tab-condition number}. We observe that the condition numbers are of order at least $10^{12}$ times larger for the non-symmetric ensemble compared with the general and the symmetric one. This confirms the observations we have made thus far, as the spectrum of the general tridiagonal matrix ensemble is quite stable for large $n$ and different $\beta$ and we calculated the equilibrium density which describes the numerics for a large range of ensemble sizes. For a better comparison, we also give the condition numbers for the same ensemble sizes of complex Ginibre matrices (normalised such that the droplet is the unit disk) in the last three rows of Table~\ref{tab-condition number}.
\begin{table}[!ht]\centering\small
	\caption[Condition number of the matrix of eigenvectors]{Condition numbers. We give the condition numbers of the matrices of eigenvectors for matrix ensembles of $m$ matrices of size $n=100,500,1000$ and $\beta=1/2,2,100,1000$ each for the general tridiagonal matrix \smash{$T_\beta$}, the symmetric matrix $S_\beta$, the non-symmetric matrix $\widetilde{T}_\beta$ and the Ginibre ensemble corresponding to $\beta=2$.}\label{tab-condition number}
	\begin{tabular}{|c|c|c|c|c|c|}
		\hline
		&\diagbox{$n(m)$}{$\beta$} & 1/2 & 2 & 100 & 1000\\
		\hline
 \hline
 \multirow{3}{*}{$T_\beta$}&
		100 (1000) &38.83 & 42.09 & 41.20 & 41.67\\
		\cline{2-6}
		&500 (500) &105.81 & 116.81 & 124.03 &114.04\\
		\cline{2-6}
		&1000 (100) & 156.14& 177.47 & 178.90 &181.97\\
		\hline
		\hline

 \multirow{3}{*}{$S_\beta$}&
		100 (1000) & 26.47 & 36.33 & 42.40 & 41.70\\
		\cline{2-6}
		&500 (500) & 77.95 & 101.37 & 114.83 &111.19\\
		\cline{2-6}
		&1000 (100) & 118.67&154.42 & 188.21 &177.21\\
		\hline
		\hline
 \multirow{3}{*}{$\widetilde{T}_\beta$}&100 (1000) & 2.57$\cdot10^{13}$ & 3.12$\cdot 10^{13}$ & 3.10$\cdot 10^{13}$ & 2.99$\cdot 10^{13}$\\
		\cline{2-6}
		&500 (500) & 3.90$\cdot 10^{13}$ & 3.87$\cdot 10^{13}$ &3.87$\cdot 10^{13}$ & 3.79$\cdot 10^{13}$\\
		\cline{2-6}
		&1000 (100) & 4.06$\cdot 10^{13}$ & 4.06$\cdot 10^{13}$ & 4.12$\cdot 10^{13}$ & 4.08$\cdot 10^{13}$\\
		\hline
		\hline

		\multirow{3}{*}{Ginibre}&
		100 (1000)& & 98.13 & & \\
		\cline{2-6}
		&500 (500)& & 405.45 & &\\
		\cline{2-6}
		&1000 (100)& & 704.47 & &\\
		
		\hline
	\end{tabular}
	
\end{table}

It is known that Ginibre matrices are not normal but their scatter plots can be described very well by the spectrum. We observe that the values of the condition numbers are of the same order as for the general and the symmetric tridiagonal matrix ensemble. Moreover, Chalker and Mehlig \cite{Chalker-Mehlig} derived that the eigenvalue condition numbers scale as $\mathcal{O}\bigl(\sqrt{N}\bigr)$ as $N\rightarrow \infty$, which indicates a mild degree of non-normality (cf.\ discussion in \cite[Section~35]{Trefethen-Embree}). Edelman proved in~\cite{Edelman} that the slope of the linearly growing condition number of a dense random matrix in the dimension is given by $C_E=2{\rm e}^{\gamma /2}$, with Euler's constant $\gamma$. However, we want to emphasise that this condition number, studied by Edelman and others, refers directly to the random matrix~$X$, rather than to the matrix of eigenvectors. We refer the reader to the literature~\cite{Edelman} for a formal definition. We checked that the calculation of this measure of non-normality yields qualitatively the same picture as we have seen in Table~\ref{tab-condition number}, i.e., the values of this quantity are of the same order for the general and the general tridiagonal ensemble while the results for the non-symmetric model are of orders of magnitudes larger.

There exist also other measures of non-normality such as the departure of normality (introduced by Henrici~\cite{Henrici}, see also \cite[Section~48]{Trefethen-Embree}). We observed that the computation of this quantity is in the case of our matrix ensembles extremely sensitive to the machine precision and we therefore cannot observe large differences between the tridiagonal matrix ensembles.\looseness=1

\subsection{Local nearest-neighbour spacing}

\begin{figure}[!ht]\centering
		\hspace{50pt}{$\beta=1/2$} \hspace{100pt}{$\beta=2$}\hspace{100pt}{$\beta=1000$} \hfill\\
		\begin{turn}{90}\hspace{40pt}{$T_\beta$}\end{turn}	
		\includegraphics[width=0.3\textwidth,angle=0]{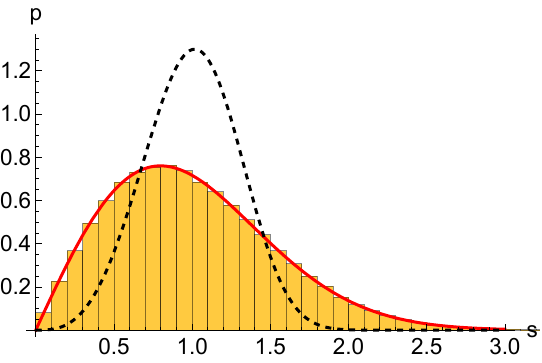}
		\includegraphics[width=0.3\textwidth,angle=0]{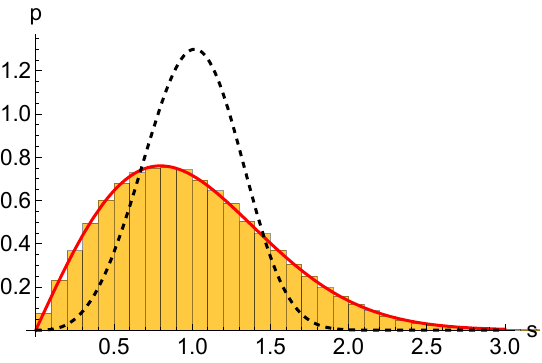}
		\includegraphics[width=0.3\textwidth,angle=0]{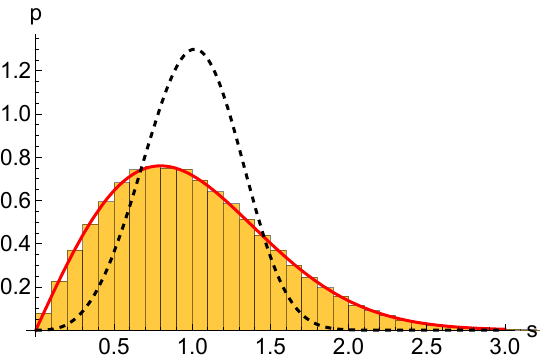}\\
		\begin{turn}{90}\hspace{40pt}{$S_\beta$}\end{turn}
		\includegraphics[width=0.3\textwidth,angle=0]{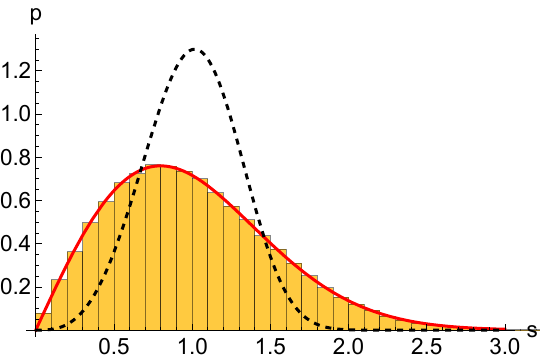}
		\includegraphics[width=0.3\textwidth,angle=0]{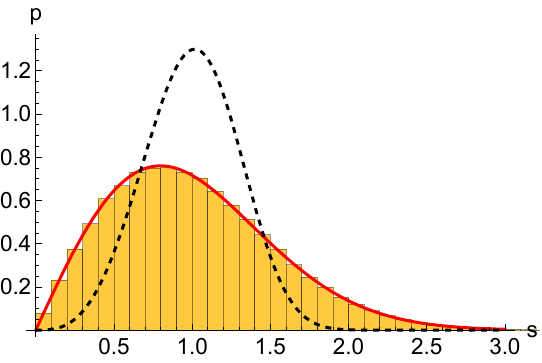}
		\includegraphics[width=0.3\textwidth,angle=0]{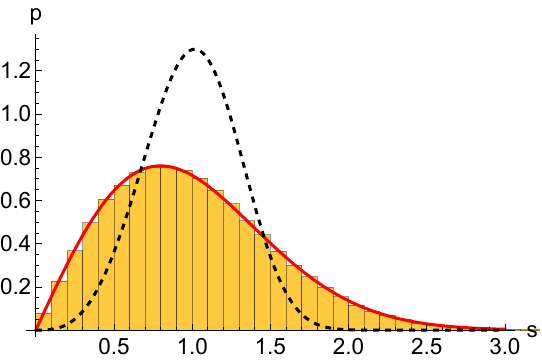}\\
		\begin{turn}{90}\hspace{40pt}{$\widetilde{T}_\beta$}\end{turn}
		\includegraphics[width=0.3\textwidth,angle=0]{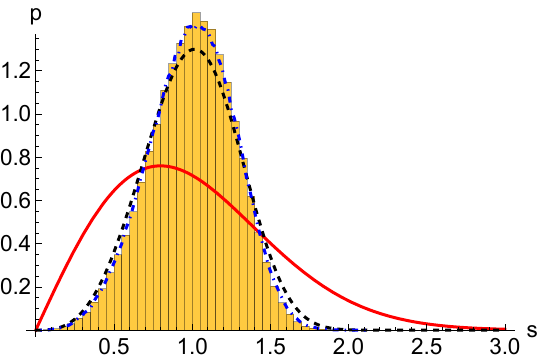}
		\includegraphics[width=0.3\textwidth,angle=0]{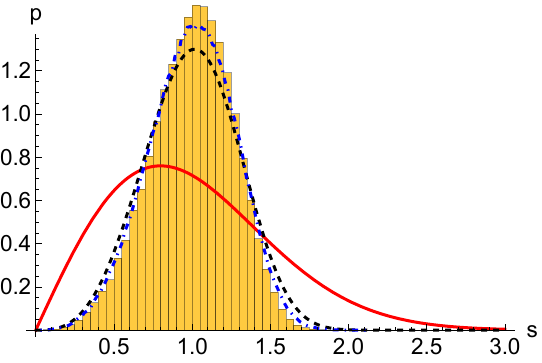}
		\includegraphics[width=0.3\textwidth,angle=0]{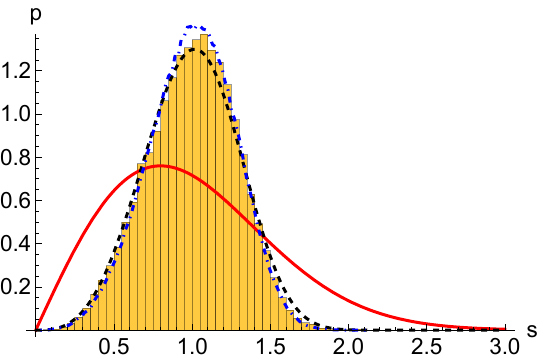}
\caption{We present histograms of the NN statistics for the ensembles \smash{$T_\beta$}, $S_\beta$, and $\widetilde{T}_\beta$ from top to bottom row. For all ensembles we consider $\beta=1/2,2,1000$ in the plots from left to right. Additionally, we give the analytical known results for 2D Poisson (${\beta=0}$)~\eqref{Poi-NN} (red full) and for complex Ginibre (${\beta=2}$)~\eqref{Gin-NN} (black dashed). We observe an agreement of the NN statistics of the ensembles~\smash{$T_\beta$} and~$S_\beta$ with 2D Poisson independently of $\beta$. In the last row, we also plot the numerical distribution of the symmetry class AII$^\dag$ (dot-dashed blue) which agrees well with $\beta=1/2$. The ensemble $\widetilde{T}_\beta$ shows a $\beta$-dependence, and it overshoots the maximum of the complex Ginibre ensemble for $\beta=1/2,2$, whereas it is in better agreement for $\beta=1000$.} \label{fig-nn}
	\end{figure}

We study the local statistics of the discussed tridiagonal matrix ensembles. For a set of $n$ complex eigenvalues $\lambda_1,\dots, \lambda_n$, we define the NN of the eigenvalue $\lambda_i$ as
	$s_i=\min_{i\neq j}|\lambda_i-\lambda_j|$ and study the distribution $p(s)$. It is convenient to normalise $p(s)$ and its first moment to unity.
We consider 500 matrices of size $N=5000$ for each of the three discussed tridiagonal matrix ensembles. For the general and symmetric matrix ensemble, we study the bulk of the spectrum, i.e., radii in the interval $ [0.1,0.6]\subset\bigl[0,\sqrt{2/{\rm e}}\bigr]$ with $\sqrt{2/{\rm e}}\approx0.858$ to stay away from the edge (cf.\ Figure~\ref{fig-radial dist T and T tilde beta 1/2,2,1000 n=5000}\,(a)) and unfold the spectrum with the square root of the analytically known spectral density \eqref{eq-limiting density T,S} (discussed in detail in~\cite{AKMP}). In both cases we find that the local NN spacing distribution (after rescaling the first moment to unity) looks for various $\beta=1/2,2,1000$ similar to the statistics of a 2D Poisson process of independent particles ($\beta=0$) with \eqref{Poi-NN}~\cite{Haake}. This is clearly not the NN spacing of a 2D Coulomb gas at the inverse temperature $\beta$ and we neither expected this nor understand this behaviour at the moment. For the non-symmetric model, we cut off the large peaks at the origin and the edge of the spectrum, i.e., consider radii in $[0.02,0.44]$ (cf.\ Figure~\ref{fig-radial dist T and T tilde beta 1/2,2,1000 n=2}\,(b)) and find a different behaviour compared to the other ensembles. The maximum of the distribution is shifted to the right with respect to Poisson, and we observe $\beta$-dependence of the distributions. For comparison we plot the analytically known NN distribution~\eqref{Gin-NN}~\cite{GHS} of the complex Ginibre ensemble ($\beta=2$) as a black dashed line, see Figure~\ref{fig-nn},
\begin{align}
	\label{Poi-NN}
	&p^{\textsc{nn}}_{\textrm{Pois}}(s) = 2 \Gamma\left(1+\frac{1}{2}\right)^2 s {\rm e}^{-\Gamma (1+\frac{1}{2} )^2 s^2}\\
	\label{Gin-NN}
	&p_{\textsc{nn}}^{\textrm{Ginibre}}(s) =
	E_{\textsc{nn}}(s)
	\sum_{j=1}^{\infty} \frac{2s^{2j+1}{\rm e}^{-s^2}}{\Gamma(1 + j, s^2)} \qquad \text{with }E_{\textsc{nn}}(s)= 	\prod_{j=1}^{\infty}\frac{\Gamma\bigl(1 + j, s^2\bigr)}{\Gamma(j+1)},
\end{align}
where the first moment of $p_{\textsc{nn}}^{\textrm{Ginibre}}(s)$ still needs to be normalized via the numerically computed factor 1.1429 cf.\ \cite{GHS}. The first moment of $p^{\textsc{nn}}_{\textrm{Pois}}(s)$ is one. In the last row of Figure~\ref{fig-nn}, we give also the local NN distribution in the bulk of class AII$^\dag$ as representative of the supposedly three universal local bulk statistics in non-Hermitian random matrices (see for further discussion~\cite{Hamazaki-et-al}). We numerically generate an ensemble of 500 matrices of size $2N=5000$, where each eigenvalue is doubly degenerate. We illustrate the NN distributions for the simplest representatives of each separate universality class in Figure~\ref{fig:NN symmetry classes}. Here, class AI$^\dag$ (red dashed) consists of complex symmetric matrices, class A (black full) is the complex Ginibre ensemble (cf. equation~\eqref{Gin-NN}) and class AII$^\dag$ (blue dot-dashed) contains complex self-dual matrices. It was found in \cite{AMP} that the local NN statistics of the classes AI$^\dag$ and AII$^\dag$ can be effectively described via the statistics of 2D Coulomb gases at $\beta=1.4$ and $\beta=2.6$, respectively. For class A, the exact correspondence with a 2D Coulomb gas at $\beta=2$ is well established. Additionally, we give also the 2D Poisson statistics (magenta dotted) from equation~\eqref{Poi-NN} in Figure~\ref{fig:NN symmetry classes}.
\begin{figure}[t]\centering
\includegraphics[width=0.45\linewidth]{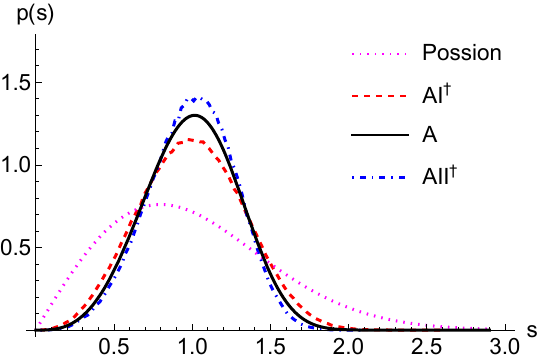}
\caption{We compare the NN distributions of the numerically generated statistics of 500 matrices of size $N=5000$, respectively $2N=5000$ of classes AI$^\dag$ (red dashed), respectively AII$^\dag$ (blue dot-dashed) with the analytically known results of class A (black full) and a 2D Poisson process (magenta dotted). We normalised the distribution with respect to the first moment of the NN statistics each.}\label{fig:NN symmetry classes}
\end{figure}

In Figure~\ref{fig-nn}, we find a surprisingly good agreement with the bottom left plot, i.e., $\beta=1/2$, which we currently do not understand. Moreover, we do not understand the shown convergence towards Ginibre for larger $\beta$ at the moment. Hence, the coupling of the eigenvectors plays an important role and needs further study in the future.

\appendix

 \section[Results for n=2]{Results for $\boldsymbol{n=2}$}\label{app-n=2}
 \begin{figure}[!ht]\centering
\begin{subfigure}{0.46\textwidth}\centering
		\includegraphics[width=\textwidth]{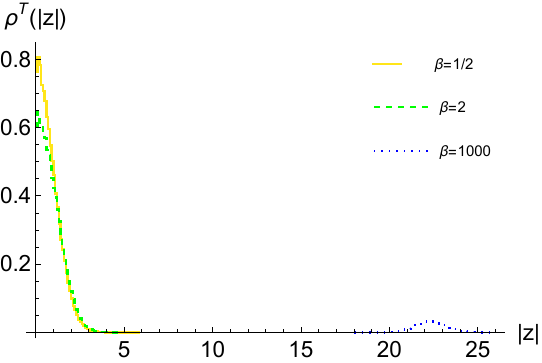}
		\subcaption{$T_\beta$\label{fig-radial dist T n=2}}	
	\end{subfigure}\quad
	\begin{subfigure}{0.46\textwidth}\centering
		\includegraphics[width=\textwidth]{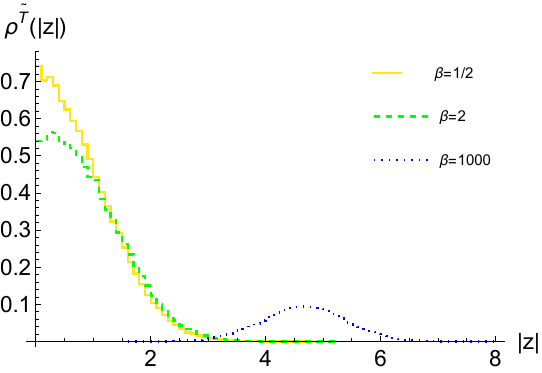}
		\subcaption{$\widetilde{T}_\beta$ \label{fig-radial dist T tilde n=2}}	
	\end{subfigure}
	\caption{Histograms of the radial density of complex eigenvalues for $50000$ matrices of size $n=2$ of the ensemble \smash{$T_\beta$} (left) and $\widetilde{T}_\beta$ (right) at $\beta=1/2$ (yellow full), $\beta=2$ (green dashed) and $\beta=1000$ (magenta dotted).} \label{fig-radial dist T and T tilde beta 1/2,2,1000 n=2}
\end{figure}
 It is well known, that the small $n$ statistics of non-Hermitian random matrices is not always a good approximation for the large $n$-limits, as pointed out in~\cite{GHS} for Ginibre matrices. This is also what we observe here. For the general ensemble \smash{$T_\beta$}, it is clear from Figure~\ref{fig-radial dist T and T tilde beta 1/2,2,1000 n=2}\,(a) that we do not have $\beta$-universality here, i.e., we see a distinct behaviour of the radial distribution for various~$\beta$ at $n=2$. For the non-symmetric tridiagonal ensemble $\widetilde{T}_\beta$, we see in Figure~\ref{fig-radial dist T and T tilde beta 1/2,2,1000 n=2}\,(b) that for small $\beta=1/2,2$ we have a relatively large support of the eigenvalues where a large amount concentrate near the origin. For $\beta=1000$ we have one large peak which is not placed at the origin and the eigenvalues are supported in a small neighbourhood of this radius. As $n=2$ is the simplest case of these non-Hermitian ensembles where we can study the jpdf and average spectral density, we will show the results and comment on the differences compared to the previously derived large $n$-results.
\begin{Lemma}\label{lemma-n=2 jpdf}
	In the low temperature limit $\beta \gg 1$, the jpdf of the eigenvalues of the $2\times 2$ matrix $L=T,S,\widetilde{T}$, is given by
	\begin{align}
			P_\beta^T(\boldsymbol{\lambda})&\approx C_\beta ^T |\lambda_1-\lambda_2|^\beta {\rm e}^{\frac{|\lambda_1-\lambda_2|^2}{4}-\frac{|\lambda_1|^2}{2}-\frac{|\lambda_2|^2}{2}}K_0 \left(\frac{|\lambda_1-\lambda_2|^2}{4}\right) \qquad \text{with }
			C_\beta^T=\frac{2^{4-\beta}\pi^2}{\beta Z_\beta^T},\!\!\label{eq-n=2 jpdf T}\\
		P_\beta^S(\boldsymbol{\lambda})&\approx C_\beta^S |\lambda_1-\lambda_2|^\beta {\rm e}^{-\frac{|\lambda_1|^2}{2}-\frac{|\lambda_2|^2}{2}} \qquad \text{with }C_\beta^S=\frac{2^{1-\beta}\pi}{\beta Z_\beta^S},\label{eq-n=2 jpdf S}\\
		P_\beta^{\widetilde{T}}(\boldsymbol{\lambda})&\approx C_\beta^{\widetilde{T}}|\lambda_1-\lambda_2|^\beta {\rm e}^{\frac{|\lambda_1-\lambda_2|^2}{4}-\frac{1}{32}|\lambda_1-\lambda_2|^4-\frac{|\lambda_1|^2}{2}-\frac{|\lambda_2|^2}{2}} \qquad \text{with }C_\beta^{\widetilde{T}}=\frac{2^{1-\beta}\pi}{\beta Z_\beta^{\widetilde{T}}},\label{eq-n=2 jpdf tilde T}
	\end{align}
where $K_0(x)$ denotes the modified Bessel function.
\end{Lemma}
\begin{Remark}
	We observe in the symmetric ensemble with $n=2$ that the function $g_S(\boldsymbol{\lambda},\mathbf{r})$ contributes only as a normalisation factor. This can be understood because a $2\times2$ matrix which is complex symmetric can be made normal by fixing only one phase. As the Frobenius norm is invariant under this change one can assume that a $2\times 2$ complex symmetric matrix is also normal, i.e., it can be diagonalised via a unitary transformation.

	Let $S=\left(\begin{smallmatrix}
			a_1 & c_1\\
			c_1 & a_2
		\end{smallmatrix}\right)$
	be a complex symmetric matrix with $a_1, a_2, c_1\in \mathbb{C}$. The normality condition $SS^\dagger =S^\dagger S$ yields one equation
	\begin{align*}
		a_1c_1^*+c_1a_2^*=a_1^*c_1+c_1^*a_2.
	\end{align*}
	We choose polar coordinates $a_1=r_1{\rm e}^{{\rm i}\Theta_1}$, $a_2=r_2{\rm e}^{{\rm i}\Theta_2}$ and $c_1=s{\rm e}^{{\rm i}\varphi}$ with $r_1,r_2,s\in [0,\infty)$, $\Theta_1,\Theta_2, \varphi\in [0,2\pi)$. After some basic algebraic computations, one finds the condition
	\begin{align*}
		{\rm e}^{2{\rm i}\varphi}=\frac{r_2{\rm e}^{{\rm i}\Theta_2}-r_1{\rm e}^{{\rm i}\Theta_2}}{r_2{\rm e}^{-{\rm i}\Theta_2}-r_1{\rm e}^{-{\rm i}\Theta_1}},
	\end{align*}
	i.e., the normality condition for a $2\times 2$ complex symmetric matrix can be expressed by fixing one phase.
\end{Remark}
\begin{proof}
	We apply the method of steepest descent for large $\beta \gg 1$ to calculate the function $f_L(\Lambda)$ for $L=T,S,\widetilde{T}$. As the $\beta$-dependent part in the expression of $f_L(\Lambda)$ is the same for all three ensembles under consideration, the saddle point analysis can be performed before specifying the ensemble. We denote the $\beta$-dependent part of $f_L(\Lambda)$ as
	\begin{align*}
		f_L^\beta(\Lambda,\mathbf{R}_1)=\int_{\mathbb{C}}\exp(-g_L(\Lambda,\mathbf{r},\mathbf{R}_1)) |r_1|^{\frac{\beta}{2}-2}|1-r_1|^{\frac{\beta}{2}-2}\mathrm{d}^2r_1,
	\end{align*}
where the $\mathbf{R}_1$-dependence drops out for the ensembles $S$ and \smash{$\widetilde{T}$}. Note that the function $g_L(\Lambda,\mathbf{r},\mathbf{R}_1)$ explicitly depends on the choice of the ensemble via its definition in equation~\eqref{gTdef} and the expressions of the matrix elements in terms of the spectral parameters. We use the following formula for a two-dimensional integration via steepest descent:
\begin{align*}
	\int h(\mathbf{x}){\rm e}^{Mp(\mathbf{x})}\mathrm{d}\mathbf{x}\approx \frac{2\pi}{M}\frac{h(\mathbf{x_0}){\rm e}^{Mp(\mathbf{x_0})}}{\det(-H_p(\mathbf{x_0}))^{1/2}} \qquad \text{as} \ M\rightarrow \infty,
\end{align*}
where $H_p(\mathbf{x_0})$ is the Hessian matrix of $p$ evaluated at $\mathbf{x_0}$, which is a global maximum/ minimum of $p(\mathbf{x})$. We denote the real and imaginary parts of $r_1$ as $(x,y)$ and rewrite $f_L^\beta(\Lambda,\mathbf{R}_1)$ as follows:
\begin{align*}
	f_L(\Lambda,\mathbf{R}_1)=\int_{-\infty}^\infty \mathrm{d}x\int_{-\infty}^\infty \mathrm{d}y \frac{\exp(-g_L(\Lambda,\mathbf{r},\mathbf{R}_1))}{\bigl(x^2+y^2\bigr)\bigl((1-x)^2+y^2\bigr)} {\rm e}^{\frac{\beta}{4}\log[(x^2+y^2)((1-x)^2+y^2)]}.
\end{align*}
Hence, we are looking for the maximum/minimum of the function
\begin{align*}
	p(x,y)=\frac{1}{4}\log\bigl[\bigl(x^2+y^2\bigr)\bigl((1-x)^2+y^2\bigr)\bigr].
\end{align*}
We find that only $(x,y)=(1/2,0)$ yields $\frac{\partial p}{\partial x}=\frac{\partial p}{\partial y}=0$. The determinant of the Hessian at~$\mathbf{x_0}$ is negative. Hence, we have a saddle point. Moreover, we observe that $p(x,y)\rightarrow -\infty$ for $(x,y)=(0,0)$ and $(x,y)=(1,0)$. Hence, the contribution of $p(x,y)$ vanishes at these points. For $x,y\rightarrow \infty$, we have a logarithmic divergence of $p(x,y)$. It follows that
\begin{align*}
	f_L(\Lambda,\mathbf{R}_1)\approx \frac{2^{4-\beta}\pi}{\beta}\exp(-g_L(\Lambda,1/2,\mathbf{R}_1)),
\end{align*}
for $\beta \gg1$. To specify this function further, we have now to specify the function $g_L(\Lambda,\mathbf{R}_1)$ for each matrix ensemble $L=T,S,\widetilde{T}$, where the $\mathbf{R}_1$ dependence only appears for the general ensemble $T$, which makes another integration over the complex variable $R_2$ necessary.

\textit{Case $L=T$:} Using the recurrence relations which were constructed in the proof of Theorem~\ref{thm:roughbijection}, we find the following expressions of the matrix entries in terms of the spectral coordinates:
\begin{alignat}{3}
	&a_2=r_1(\lambda_1-\lambda_2)+\lambda_2,\qquad&&
	c_1=\frac{r_1(1-r_1)(\lambda_1-\lambda_2)}{R_2},&\nonumber\\
	&a_1=-r_1(\lambda_1-\lambda_2)+\lambda_1,\qquad&&
	b_1=R_2(\lambda_1-\lambda_2),&\label{eq-matrix elements T n=2}
\end{alignat}
which simplify as follows for $r_1=1/2$ to
\begin{align*}
	a_1=a_2=\frac{\lambda_1+\lambda_2}{2}, \qquad c_1=\frac{\lambda_1-\lambda_2}{4R_2}, \qquad b_1=R_2(\lambda_1-\lambda_2).
\end{align*}
With these expressions and the definition of the function $g_T(\Lambda,\mathbf{r},\mathbf{R}_1)$ in equation~\eqref{gTdef}, we find at $r_1=1/2$
\begin{align*}	g_T(\Lambda,1/2,R_2)=-\frac{|\lambda_1-\lambda_2|^2}{4}+\frac{|\lambda_1-\lambda_2|^2}{2} \left(|R_2|^2+\frac{1}{16|R_2|^2}\right).
\end{align*}
We can now calculate the function $f_T(\Lambda)$ (cf.\ equation~\eqref{factor}) as follows:
\begin{align*}
	f_T(\Lambda)&=\int_{\mathbb{C}\backslash \{0\}}\frac{f_T^\beta(\Lambda,\mathbf{R}_1)}{|R_2|^2}\mathrm{d}^2R_2\nonumber\\
	&\approx \frac{2^{5-\beta}\pi^2}{\beta}{\rm e}^{\frac{|\lambda_1-\lambda_2|^2}{4}}\int_{0}^\infty \mathrm{d} |R_2|\, |R_2|^{-1}{\rm e}^{-\frac{|\lambda_1-\lambda_2|^2}{2}(|R_2|^2+\frac{1}{16}|R_2|^{-2})}\nonumber \\
	&= \frac{2^{5-\beta}\pi^2}{\beta}{\rm e}^{\frac{|\lambda_1-\lambda_2|^2}{4}}K_0\left(\frac{|\lambda_1-\lambda_2|^2}{4}\right),
\end{align*}
where we used the integral representation \cite[equation~(8.432)\,(7)]{gradshteyn} of the modified Bessel function $K_0(x)$. The joint probability density in equation~\eqref{eq-n=2 jpdf T} follows now from equation~\eqref{eq:EigDens1}.

\textit{Case $L=S$:} Using similar recurrence relations as presented in the proof of Theorem~\ref{thm:roughbijection} adapted to the symmetric ensemble described in equation~\eqref{factorS} and below, we find that $g_S(\Lambda,\mathbf{r}_1)=1$ at $r_1=1/2$. Hence, the factor $f_S(\Lambda,\mathbf{r}_1)$ gives only a constant contribution yielding the result in equation~\eqref{eq-n=2 jpdf S}.

\textit{Case $L=\widetilde{T}$:} Via the recurrence relations detailed in the proof of \cite[Theorem~2.2]{mezzadri-taylor}, we find at $r_1=1/2$
\begin{align*}
	g_{\widetilde{T}}(\Lambda,1/2)=-\frac{|\lambda_1-\lambda_2|^2}{4}+\frac{|\lambda_1-\lambda_2|^4}{32},
\end{align*}
which yields the joint probability density in equation~\eqref{eq-n=2 jpdf tilde T}.
\end{proof}

Using the results of Lemma~\ref{lemma-n=2 jpdf}, we compute the average spectral density, which is defined~as
\begin{align*}
	\rho^L(\lambda)=\int_\mathbb{C}\mathrm{d}\lambda_2\, P^L(\lambda,\lambda_2).
\end{align*}
\begin{Proposition}
	In the low temperature limit $\beta \gg 1$, the spectral density of the $2\times 2$ matrix $L=T,S,\widetilde{T}$, is given by
	\begin{align}
&		\rho^T(\lambda)=\frac{1}{\pi{}_2F_1\bigl(1+\frac{\beta}{2},1+\frac{\beta}{2}; \frac{3}{2}+\frac{\beta}{2};\frac{1}{2}\bigr)}{\rm e}^{-|\lambda|^2}\,{}_2F_2\left(1+\frac{\beta}{2},1+\frac{\beta}{2}; 1,\frac{3}{2}+\frac{\beta}{2};\frac{|\lambda|^2}{2}\right),\label{eq-n=2 density T}\\
&		\rho^S(\lambda)=\frac{1}{\pi 2^{\frac{\beta}{2}+1}}\,{}_1F_1\left(\frac{\beta}{2}+1,1,\frac{|\lambda|^2}{2}\right){\rm e}^{-|\lambda|^2},\label{eq-n=2 density S}\\
&\rho^{\widetilde{T}}(\lambda)=\frac{{\rm e}^{-|\lambda|^2+\frac{1}{4}}\sum_{k=0}^\infty\frac{\Gamma(\frac{\beta}{2}+k+1)}{\Gamma(k+1)^2}|\lambda|^{2k}D_{-\frac{\beta}{2}-k-1}(1)}{\pi \Gamma\bigl(\frac{\beta}{2}+1\bigr)D_{-\frac{\beta}{2}-1}(0)} \nonumber 
	\end{align}
for $\lambda\in \mathbb{C}$. $D_\nu(x)$ denotes the parabolic cylinder function.
\end{Proposition}
\begin{Remark}
We observe that spectral densities of the general ensemble $T$ and the symmetric one $S$ agree as $\beta/2+1\approx \beta/2+3/2$ in the low temperature limit $\beta \gg 1$, which reduces the hypergeometric function ${}_2F_2$ in equation~\eqref{eq-n=2 density T} to the same ${}_1F_1$ as in equation~\eqref{eq-n=2 density S}.
\end{Remark}
\begin{proof}
To perform the complex integral we choose the new coordinate $w=\lambda_2-\lambda$ and take polar coordinates for $\lambda$ and $w$, i.e., $\lambda=r_1{\rm e}^{{\rm i}\theta_1}$ and $w=r_2{\rm e}^{{\rm i}\theta_2}$ with $r_1,r_2\in [0,\infty)$ and $\theta_1,\theta_2 \in [0,2\pi)$. It follows that the angular-depending part of the integral is of the same form for all three ensembles $L=T,S,\widetilde{T}$, i.e., we have to perform the following integral:
	\begin{align*}
		\int_0^{2\pi}\mathrm{d}\Theta_2\, {\rm e}^{-r_1r_2\cos(\Theta_1-\Theta_2)}.
	\end{align*}
Therefore, we apply the addition theorem for the cosine and use \cite[equation~(3.338)\,(4)]{gradshteyn}. The result is given in terms of the modified Bessel function of the first kind $I_0(x)$.

\textit{Case $L=T$:} We are left with the following radial integral:
\begin{align*}
	\rho^T(\lambda)=2\pi C_\beta^T \int_0^\infty \mathrm{d}r_2 \, r_2^{\beta+1}{\rm e}^{-r_1^2-\frac{r_2^2}{4}}I_0(r_1r_2)K_0\left(\frac{r_2^2}{4}\right).
\end{align*}
	We use the summation formula \cite[equation~(8.445)]{gradshteyn}
	\begin{align}\label{eq-sum I_0}
		I_0(x)=\sum_{k=0}^\infty \frac{x^{2k}}{4^k\Gamma(k+1)^2}.
	\end{align}
	Next, we swap the summation and the integration. Therefore, we apply Tonelli's theorem (we use the counting measure over $\mathbb{N}$ to identify the summation as an integration), as the function is positive for all $r_2\in [0,\infty)$ and $k\in \mathbb{N}$,
	\begin{align*}
		\Rightarrow \rho^T(\lambda)&=2\pi C_\beta^T {\rm e}^{-r_1^2}\sum_{k=0}^\infty \frac{r_1^{2k}}{4^k\Gamma(k+1)^2}\int_0^\infty \mathrm{d}r_2\, r_2^{\beta+2k+1}{\rm e}^{-\frac{r_2^2}{4}}K_0\left(\frac{r_2^2}{4}\right).
	\end{align*}
	We change variables to $t=r_2^2/4$ and find
	\begin{align*}
		\Rightarrow \rho_\beta^T(\lambda)&=2\pi C_\beta^T {\rm e}^{-r_1^2}\sum_{k=0}^\infty \frac{r_1^{2k}2^{\beta+2k+1}}{4^k\Gamma(k+1)^2}\int_0^\infty \mathrm{d}t
\, t^{\frac{\beta}{2}+k}{\rm e}^{-t}K_0(t)\\
		&=2\pi C_\beta^T {\rm e}^{-r_1^2}\sum_{k=0}^\infty \frac{r_1^{2k}2^{\beta+2k-2k+1}}{\Gamma(k+1)^2}\sqrt{\pi}2^{-\frac{\beta}{2}-k-1} \frac{\Gamma\bigl(\frac{\beta}{2}+k+1\bigr)^2}{\Gamma\bigl(\frac{\beta}{2}+k+\frac{3}{2}\bigr)}\\
		&=2^{\frac{\beta}{2}+1}\pi^{\frac{3}{2}} \frac{\Gamma\bigl(\frac{\beta}{2}+1\bigr)^2}{\Gamma\bigl(\frac{\beta}{2}+\frac{3}{2}\bigr)} C_\beta^T {\rm e}^{-r_1^2}\,{}_2F_2\left(1+\frac{\beta}{2},1+\frac{\beta}{2}; 1,\frac{3}{2}+\frac{\beta}{2};\frac{r_1^2}{2}\right),
	\end{align*}
	where we used \cite[equation~(6.621)\,(3)]{gradshteyn} and ${}_2F_1(a,b;c;0)=1$ in the first step.	We are now left with the calculation of the normalisation such that
	\begin{align*}
		1=\int_\mathbb{C} \mathrm{d}\lambda\, \rho^T (\lambda).
	\end{align*}
	Using a change of variables $t=r_1^2$ and \cite[equation~(7.522)\,(5)]{gradshteyn}, we find
	\begin{align*}
		\int_0^\infty \mathrm{d}r_1\, r_1 {\rm e}^{-r_1^2}\,{}_2F_2\left(1+\frac{\beta}{2},1+\frac{\beta}{2}; 1,\frac{3}{2}+\frac{\beta}{2};\frac{r_1^2}{2}\right)
		=\frac{1}{2}\,{}_2F_1\left(1+\frac{\beta}{2},1+\frac{\beta}{2}; \frac{3}{2}+\frac{\beta}{2};\frac{1}{2}\right).
	\end{align*}
	This proves equation~\eqref{eq-n=2 density T}.

\textit{Case $L=S$:} These calculations are detailed in \cite{AMP}.

\textit{Case $L=\widetilde{T}$:} We start from the radial integration of the form
\begin{align}\label{eq-int form n=2 non sym}
	\rho_\beta^{\widetilde{T}}(\lambda)&=2\pi K_\beta^{\widetilde{T}} \int_0^\infty \mathrm{d}r_2\, r_2^{\beta+1}{\rm e}^{-r_1^2-\frac{r_2^2}{4}-\frac{r_2^4}{32}}I_0(r_1r_2).
\end{align}
We again apply the summation formula equation~\eqref{eq-sum I_0}, change to the variable $w=r_2^2$ and apply \cite[equation~(3.462)\,(1)]{gradshteyn}. We find
\begin{align*}
	\rho_\beta^{\widetilde{T}}(\lambda)	=\pi K_\beta^{\widetilde{T}} {\rm e}^{-r_1^2}\sum_{k=0}^\infty \frac{r_1^{2k}}{4^k\Gamma(k+1)^2}2^{\beta +2k+2}\Gamma\left(\frac{\beta}{2}+k+1\right){\rm e}^{\frac{1}{4}}D_{-\frac{\beta}{2}-k-1}(1).
\end{align*}
In order to calculate the normalisation, we start from the integral representation in \eqref{eq-int form n=2 non sym} and apply integration over $\lambda\in \mathbb{C}$. Using \cite[equation~6.614\,(3)]{gradshteyn} and \cite[equations~(9.220)\,(2) and (9.215)]{gradshteyn} to evaluate the resulting Whittaker function, we find
\begin{align*}
	\int_\mathbb{C}\mathrm{d} \lambda\, \rho^{\widetilde{T}}(\lambda) =\pi^2 K_\beta^{\widetilde{T}}\int_0^\infty \mathrm{d}r_2\, r_2^{\beta+1}{\rm e}^{-\frac{r_2^4}{32}}
	 =2^{2+2\beta}\pi^2K_\beta^{\widetilde{T}}\Gamma\left(1+\frac{\beta}{2}\right)D_{-(1+\frac{\beta}{2})}(0),
\end{align*}
where we applied \cite[equation~(3.462)\,(1)]{gradshteyn} and $D_\mu(x)$ denotes the parabolic cylinder function.\looseness=1
\end{proof}

\section{Spectral decomposition of symmetric tridiagonal matrices}\label{app-symm}
Let $\mathcal{S}(n)$ denote the set of matrices $S$, where
\begin{equation*}
 S= \begin{pmatrix}
a_n & \widetilde{c}_{n-1} & & & \\
\widetilde{c}_{n-1} & a_{n-1} & \widetilde{c}_{n-2} & & \\
 & \ddots & \ddots & \ddots&\\
 & & \widetilde{c}_2 & a_2& \widetilde{c}_1\\
 & & & \widetilde{c}_1 & a_1
\end{pmatrix},
\end{equation*}
with $a_j,\widetilde{c}_j\in\mathbb{C}$, such that
$\det(S)\neq 0$ and the spectrum of $S$ is non-degenerate.
The sets $D(n)$ and $\Lambda(n)$ are defined in the same way as in equations~\eqref{def:D(n)} and~\eqref{def:Lambda(n)}. We adapt the definition of $R(n)$, defined in equation~\eqref{eq:R(n)def}, to be
\begin{align*}
 R^{S}(n)=\bigl\{R\in \mathrm{GL}(n,\C) \mid S=R\Lambda R^{-1},\, S\in \Sw(n),\, \Lambda\in\Lambda(n)\bigr\}.
\end{align*}
When $S\in \Sw(n)$, the spectral decomposition of $S$ is given by
\begin{equation}\label{eq:specdecompsym}
 S=R\Lambda R^{-1},
\end{equation}
where $\Lambda\in\Lambda(n)$ and $R\in R^{S}(n)$. The spectral decomposition in equation~\eqref{eq:specdecompsym} is unique up to a permutation of the eigenvalues and the right multiplication $R\mapsto RD$, where $D\in D(n)$.

Consequently, the spectral decomposition in equation~\eqref{eq:specdecompsym} defines a bijection
\begin{equation}\label{sing_valued_map_sym}
 \mathcal{F}^{S}\colon \ \Sw(n)\rightarrow\mathcal{L}(n)\times \mathcal{R}^{S}(n),
\end{equation}
where $\mathcal{L}(n)$ is defined by equation~\eqref{def:curlyL(n)}, and
$\mathcal{R}^{S}(n)=R^{S}(n)/D(n)$.

\begin{Lemma}\label{parameterisation_sym}
 Let $\mathbf{r}=(r_1,\dots,r_n)$ and $\mathbf{v}=(v_1,\dots,v_n)$ be the first and last rows of $R\in R^{S}(n)$, respectively. We have $r_j\neq0$, $v_j\neq0$ for all $j\in\{1,2,\dots,n\}$. Moreover, the subset of $R^{S}(n)$ such that $ r_1+r_2+\dots+r_n=1$, spans a set of representatives in $\mathcal{R}^{S}(n)$.
\end{Lemma}
\begin{proof}
The proof follows analogously to the proof of Lemma~\ref{r_diff_from_zero}.
\end{proof}

Next, we aim to show that the bijection in equation~\eqref{sing_valued_map_sym} induces an injective map
\begin{equation*}
 \mathcal{G}^{S}\colon \ \Sw(n)\rightarrow\mathcal{L}(n)\times \mathcal{H}_n,
\end{equation*}
where $\mathcal{H}_n$ is defined by equation~\eqref{hyperplane}. It follows from equation~\eqref{sing_valued_map_sym} and Lemma~\ref{parameterisation_sym} that the map $\mathcal{G}^{S}$ is singular valued. Our next goal is to prove the injectivity of this map. To do so, we~let%
\begin{equation}
 \label{eq:newnotsym}
 R =\begin{pmatrix}\mathbf{r}_1^{\rm t}\\ \vdots \\ \mathbf{r}_n^{\rm t}\end{pmatrix}\qquad\text{and}\qquad
 R^{-1}= \bigl(\mathbf{l}_1,\dots,\mathbf{l}_n\bigr),
\end{equation}
where
$
\mathbf{r}_j^{\rm t} = \bigl( r_{j1},\dots, r_{jn}\bigr)$, $
 \mathbf{l}_j = (l_{1j},\dots, l_{nj})^{\rm t}$ for $j=1,\dots,n$.

\begin{Theorem}\label{thm-recursion matrix elements T cs}
Let $\Lambda\in\Lambda(n)$ and $\mathbf{r}=(r_1,\dots,r_n)\in\mathcal{H}_n$. Then, there exists a unique $R\in R^{S}(n)$ whose first row is $\mathbf{r}$, and a unique $S\in \Sw(n)$ such that $S=R\Lambda R^{-1}$.
\end{Theorem}
\begin{proof}
Take $R$ from the set of representatives of $\mathcal{R}^{S}(n)$ defined in Lemma~\ref{parameterisation_sym}. Let us write the matrix equations $SR=R\Lambda$ and $R^{-1}S=\Lambda R^{-1}$ as a set of vector equations
\begin{subequations}\label{eq-vector eqs T cs}
	\begin{align}
		&\mathbf{r}_j^{\rm t}\Lambda=\widetilde{c}_{n-j+1}\mathbf{r}_{j-1}^{\rm t}+a_{n-j+1}\mathbf{r}_j^{\rm t}+\widetilde{c}_{n-j}\mathbf{r}_{j+1}^{\rm t},\\
		&\Lambda\mathbf{l}_j=\widetilde{c}_{n-j+1}\mathbf{l}_{j-1}+a_{n-j+1}\mathbf{l}_j+\widetilde{c}_{n-j}\mathbf{l}_{j+1},
	\end{align}
\end{subequations}
for $j=1,\dots,n$, where we have adopted the notation in equation~\eqref{eq:newnotsym}. Upon letting $\rv_1^{\rm t}=\mathbf{r}$, we want to show that given $\Lambda$ and $\mathbf{r}_1^{\rm t}$ with $\mathbf{r}_1^{\rm t}\mathbf{l}_1=1$, where $\mathbf{l}_1=(1,\dots,1)^{\rm t}$,
we can reconstruct~$S$,~$R$ and $R^{-1}$ uniquely from the recurrence relations in equations~\eqref{eq-vector eqs T cs} with boundary conditions $\widetilde{c}_0=0$, $\mathbf{l}_0=\mathbf{0}$ and $\mathbf{r}_{0}^{\rm t}=\mathbf{0}$. The vectors $\mathbf{r}_{n+1}^{\rm t}$ and $\mathbf{l}_{n+1}$ are undetermined.

Write
\begin{subequations}\label{eq-first recurrence step T cs}
	\begin{align}
		&a_n=\mathbf{r}_1^{\rm t} \Lambda \mathbf{l}_1, \label{eq-recurrence a_n for T cs}\\
&\widetilde{c}_{n-1}= \sqrt{\bigl(\mathbf{r}_1^{\rm t} \Lambda-a_n\mathbf{r}_1^{\rm t}\bigr)\Lambda \mathbf{l}_1}, \label{eq-recurrence tilde c_n-1} \\
		&\mathbf{r}_2^{\rm t}=\frac{1}{\widetilde{c}_{n-1}}\bigl(\mathbf{r}_1^{\rm t} \Lambda-a_n\mathbf{r}_1^{\rm t}\bigr), \label{eq-recurrence v_2}\\
		&\mathbf{l}_2=\frac{1}{\widetilde{c}_{n-1}}(\Lambda \mathbf{l}_1-a_n\mathbf{l}_1).\label{eq-recurrence u_2}
	\end{align}
\end{subequations}
In the case where $\widetilde{c}_{n-1}=0$, the construction in equation~\eqref{eq-recurrence tilde c_n-1} fails. Therefore, we decompose~$S$ into two smaller-dimensional symmetric tridiagonal matrices. Specifically, $S=S_1\oplus U_{n-1}$, where~$S_1$ denotes the principal sub-matrix of $S\in\mathcal{S}(n)$ obtained by retaining the first row and column of $S$, and let $U_{n-1}$ represent the principal sub-matrix obtained by keeping the last $n-1$ rows and columns of $S$. This decomposition reduces the problem to two lower-dimensional cases, which can be handled in two separate proofs. In general, if we encounter a $j$ for which $\widetilde{c}_{n-j}=0$, we use a similar decomposition and carry out separate proofs for the lower-dimensional matrices that we form since they are symmetric tridiagonal matrices. Therefore, for the remainder of the proof, we assume that $\widetilde{c}_{n-j}\neq0$ for all $j$, thereby ensuring we avoid any singularities in our constructions.

The quantities $a_n$, $\widetilde{c}_{n-1}$, $\mathbf{r}_1^{\rm t}$ and $\mathbf{l}_1$ solve equations~\eqref{eq-vector eqs T cs} for $j=1$. Next, we shall check that they are consistent with the condition $RR^{-1}=I$. One finds
\begin{align*}
	&\mathbf{r}_2^{\rm t}\mathbf{l}_1=\frac{1}{\widetilde{c}_{n-1}}\bigl(\mathbf{r}_1^{\rm t}\Lambda \mathbf{l}_1-a_n\bigr)=0,\qquad
	\mathbf{r}_1^{\rm t}\mathbf{l}_2=\frac{1}{\widetilde{c}_{n-1}}\bigl(\mathbf{r}_1^{\rm t}\Lambda \mathbf{l}_1-a_n\bigr)=0,
	\\
&\mathbf{r}_2^{\rm t}\mathbf{l}_2=\frac{\mathbf{r}_2^{\rm t}}{\widetilde{c}_{n-1}}(\Lambda \mathbf{l}_1-a_n\mathbf{l}_1)=1.
\end{align*}
When $1<j<n$ take $\mathbf{r}_{j-1}^{\rm t},\ \mathbf{r}_{j}^{\rm t},\ \mathbf{l}_{j-1},\ \mathbf{l}_{j}$, subject to the conditions
	\begin{align*}
		&\mathbf{r}_j^{\rm t}\mathbf{l}_j=\mathbf{r}_{j-1}^{\rm t}\mathbf{l}_{j-1}=1,\qquad \mathbf{r}_{j}^{\rm t}\mathbf{l}_{j-1}=\mathbf{r}_{j-1}^{\rm t}\mathbf{l}_{j}=0,\qquad
		\widetilde{c}_{n-j+1}=\mathbf{r}_j^{\rm t}\Lambda \mathbf{l}_{j-1}\neq0.
	\end{align*}
Define
\begin{subequations}\label{eq-recurrence eqs T cs j to j+1}
\begin{align}
&a_{n-j+1}=\mathbf{r}_j^{\rm t}\Lambda \mathbf{l}_j,\label{eq-recurrence a_n-j+1 for T cs}\\
&\widetilde{c}_{n-j}=\sqrt{\bigl(\mathbf{r}_j^{\rm t}\Lambda-a_{n-j+1}\mathbf{r}_j^{\rm t}-\widetilde{c}_{n-j+1}\mathbf{r}_{j-1}^{\rm t}\bigr)\Lambda \mathbf{l}_j}, \label{eq-recurrence tilde c_n-j for T cs}\\
		&\mathbf{r}_{j+1}^{\rm t}=\frac{1}{\widetilde{c}_{n-j}}\bigl(\mathbf{r}_j^{\rm t}\Lambda-a_{n-j+1}\mathbf{r}_j^{\rm t}-\widetilde{c}_{n-j+1}\mathbf{r}_{j-1}^{\rm t}\bigr), \label{eq-recurrence v_j+1 for T cs}\\
&\mathbf{l}_{j+1} =\frac{1}{\widetilde{c}_{n-j}}(\Lambda\mathbf{l}_j-a_{n-j+1}\mathbf{l}_j-\widetilde{c}_{n-j+1}\mathbf{l}_{j-1}). \label{eq-recurrence u_j+1 for T cs}
\end{align}
\end{subequations}
We can demonstrate that equations~\eqref{eq-recurrence eqs T cs j to j+1} solve the recurrence relations in equations~\eqref{eq-vector eqs T cs}. Next, we need to establish consistency in the same way as in the step for $j=1$. To begin, we note that $\widetilde{c}_{n-j}=\mathbf{r}_{j+1}^{\rm t}\Lambda \mathbf{l}_j=\mathbf{r}_{j}^{\rm t}\Lambda \mathbf{l}_{j+1}$, since $S$ is a symmetric matrix. Thus,
\begin{alignat*}{3}
	&\mathbf{r}_{j+1}^{\rm t}\mathbf{l}_j=\frac{1}{\widetilde{c}_{n-j}}\bigl(\mathbf{r}_j^{\rm t}\Lambda \mathbf{l}_j -a_{n-j+1}\bigr)=0,\qquad&&
	\mathbf{r}_{j+1}^{\rm t}\mathbf{l}_{j-1}=\frac{1}{\widetilde{c}_{n-j}}\bigl(\mathbf{r}_j^{\rm t}\Lambda \mathbf{l}_{j-1}-\widetilde{c}_{n-j+1}\bigr)=0,& \\
	&\mathbf{r}_j^{\rm t}\mathbf{l}_{j+1}=\frac{1}{\widetilde{c}_{n-j}}\bigl(\mathbf{r}_j^{\rm t} \Lambda \mathbf{l}_j-a_{n-j+1}\bigr)=0,\qquad&&
	\mathbf{r}_{j-1}^{\rm t}\mathbf{l}_{j+1}=\frac{1}{\widetilde{c}_{n-j}}\bigl(\mathbf{r}_{j-1}^{\rm t} \Lambda \mathbf{l}_j-\widetilde{c}_{n-j+1}\bigr)=0,& \\
	&\mathbf{r}_{j+1}^{\rm t}\mathbf{l}_{j+1}=\frac{1}{\widetilde{c}_{n-j}}(\mathbf{r}_{j+1}\Lambda \mathbf{l}_j)=1.&&&
\end{alignat*}
The orthogonality relations $\mathbf{r}_{j+1}^{\rm t}\mathbf{l}_k=0$, $\mathbf{r}_k^{\rm t}\mathbf{l}_{j+1}=0$, for $k=1,\dots,j-2$, follow by induction and from the fact that by construction $\mathbf{r}_j^{\rm t}\Lambda \mathbf{l}_k=0$ if $|j-k|>1$. The boundary condition $\mathbf{r}_{n+1}^{\rm t}=\mathbf{0}$ terminates the recurrence relations.
\end{proof}

\subsection*{Acknowledgements}
We would like to thank Tamara Grava, Iv\'an Parra and Mariya Shcherbina for fruitful discussions and Leslie Molag, Francisco Jose Marcellan Espanol and Walter van Assche for useful correspondence.

G.~Akemann and P.~P\"a{\ss}ler were partially supported by the DFG through the grant CRC 1283 ``Taming
uncertainty and profiting from randomness and low regularity in analysis, stochastics and their applications''. G.~Akemann also acknowledges support by a Leverhulme Trust Visiting Professorship VP1-2023-007. He thanks the School of Mathematics at the University of Bristol for its kind hospitality during his sabbatical stay, where part of this work was conducted. H.~Taylor was supported by EPSRC grant no. EP/T517872/1.

We thank the anonymous referee for the suggestion to look into the NN statistics of the ensembles and pointing us to flaws in the proofs of Lemma~\ref{lemma-mean and variance kappa's} and Theorem~\ref{thm-Barbarino Noferini} in a first version of the paper.

\pdfbookmark[1]{References}{ref}
\LastPageEnding


\begin{thebibliography}{99}
\footnotesize\itemsep=0pt

\bibitem{Abramowitz-Stegun}
Abramowitz M., Stegun I.A., Handbook of mathematical functions with formulas,
 graphs, and mathematical tables, \textit{National Bureau of Standards Applied
 Mathematics Series}, Vol.~55, U.S.~Government Printing Office, Washington,
 DC, 1964.

\bibitem{AKMP}
Akemann G., Kieburg M., Mielke A., Prosen T., Universal signature from
 integrability to chaos in dissipative open quantum systems,
 \href{https://doi.org/10.1103/physrevlett.123.254101}{\textit{Phys. Rev.
 Lett.}} \textbf{123} (2019), 254101, 6~pages,
 \href{http://arxiv.org/abs/1910.03520}{arXiv:1910.03520}.

\bibitem{AMP}
Akemann G., Mielke A., P\"a{\ss}ler P., Spacing distribution in the
 two-dimensional {C}oulomb gas: {S}urmise and symmetry classes of
 non-{H}ermitian random matrices at noninteger~{$\beta$},
 \href{https://doi.org/10.1103/physreve.106.014146}{\textit{Phys. Rev.~E}}
 \textbf{106} (2022), 014146, 11~pages,
 \href{http://arxiv.org/abs/2112.12624}{arXiv:2112.12624}.

\bibitem{AB81}
Alastuey A., Jancovici B., On the classical two-dimensional one-component
 {C}oulomb plasma,
 \href{https://doi.org/10.1051/jphys:019810042010100}{\textit{J.~Physique}}
 \textbf{42} (1981), 1--12.

\bibitem{AKvI00}
Aptekarev A., Kaliaguine V., Van~Iseghem J., The genetic sums' representation
 for the moments of a system of {S}tieltjes functions and its application,
 \href{https://doi.org/10.1007/s003650010004}{\textit{Constr. Approx.}}
 \textbf{16} (2000), 487--524.

\bibitem{Barbarino-Noferini}
Barbarino G., Noferini V., The limit empirical spectral distribution of complex
 matrix polynomials,
 \href{https://doi.org/10.1142/S201032632250023X}{\textit{Random Matrices
 Theory Appl.}} \textbf{11} (2022), 2250023, 32~pages,
 \href{http://arxiv.org/abs/2102.02058}{arXiv:2102.02058}.

\bibitem{Bobkov-Chistyakov}
Bobkov S., Chistyakov G., Bounds for the maximum of the density of the sum of
 independent random variables,
 \href{https://doi.org/10.1007/s10958-014-1836-9}{\textit{J.~Math. Sci.}}
 \textbf{199} (2014), 100--106.

\bibitem{Can-et-al}
Can T., Forrester P.J., T\'ellez G., Wiegmann P., Singular behavior at the edge
 of {L}aughlin states,
 \href{https://doi.org/10.1103/PhysRevB.89.235137}{\textit{Phys. Rev.~B}}
 \textbf{89} (2014), 235137, 7~pages,
 \href{http://arxiv.org/abs/1307.3334}{arXiv:1307.3334}.

\bibitem{CSA}
Cardoso G., St\'ephan J.-M., Abanov A.G., The boundary density profile of a
 {C}oulomb droplet. {F}reezing at the edge,
 \href{https://doi.org/10.1088/1751-8121/abcab9}{\textit{J.~Phys.~A}}
 \textbf{54} (2021), 015002, 24~pages,
 \href{http://arxiv.org/abs/2009.02359}{arXiv:2009.02359}.

\bibitem{Chihara}
Chihara T.S., An introduction to orthogonal polynomials, \textit{Math. Appl.},
 Vol.~13, Gordon and Breach Science Publishers, New York, 1978.

\bibitem{CC83}
Choquard P., Clerouin J., Cooperative phenomena below melting of the
 one-component two-dimensional plasma,
 \href{https://doi.org/10.1103/PhysRevLett.50.2086}{\textit{Phys. Rev. Lett.}}
 \textbf{50} (1983), 2086--2089.

\bibitem{cullum-willoghby}
Cullum J.K., Willoughby R.A., A {$QL$} procedure for computing the eigenvalues
 of complex symmetric tridiagonal matrices,
 \href{https://doi.org/10.1137/S0895479894137639}{\textit{SIAM~J. Matrix Anal.
 Appl.}} \textbf{17} (1996), 83--109.

\bibitem{Cunden-et-al}
Cunden F.D., Mezzadri F., O'Connell N., Simm N., Moments of random matrices and
 hypergeometric orthogonal polynomials,
 \href{https://doi.org/10.1007/s00220-019-03323-9}{\textit{Comm. Math. Phys.}}
 \textbf{369} (2019), 1091--1145,
 \href{http://arxiv.org/abs/1805.08760}{arXiv:1805.08760}.

\bibitem{Demmel}
Demmel J.W., Applied numerical linear algebra,
 \href{https://doi.org/10.1137/1.9781611971446}{SIAM}, Philadelphia, PA, 1997.

\bibitem{DE02A}
Dumitriu I., Edelman A., Matrix models for beta ensembles,
 \href{https://doi.org/10.1063/1.1507823}{\textit{J.~Math. Phys.}} \textbf{43}
 (2002), 5830--5847, \href{http://arxiv.org/abs/math-ph/0206043}{arXiv:math-ph/0206043}.

\bibitem{DE02B}
Dumitriu I., Edelman A., Eigenvalues of {H}ermite and {L}aguerre ensembles:
 large beta asymptotics,
 \href{https://doi.org/10.1016/j.anihpb.2004.11.002}{\textit{Ann. Inst.
 H.~Poincar\'e Probab. Statist.}} \textbf{41} (2005), 1083--1099, \href{http://arxiv.org/abs/math-ph/0403029}{arXiv:math-ph/0403029}.

\bibitem{Dys62}
Dyson F.J., The threefold way. {A}lgebraic structure of symmetry groups and
 ensembles in quantum mechanics,
 \href{https://doi.org/10.1063/1.1703863}{\textit{J.~Math. Phys.}} \textbf{3}
 (1962), 1199--1215.

\bibitem{Edelman}
Edelman A., Eigenvalues and condition numbers of random matrices,
 \href{https://doi.org/10.1137/0609045}{\textit{SIAM~J. Matrix Anal. Appl.}}
 \textbf{9} (1988), 543--560.

\bibitem{PetersBook}
Forrester P.J., Log-gases and random matrices, \textit{London Math. Soc.
 Monogr. Ser.}, Vol.~34,
 \href{https://doi.org/10.1515/9781400835416}{Princeton University Press},
 Princeton, NJ, 2010.

\bibitem{Ginibre}
Ginibre J., Statistical ensembles of complex, quaternion, and real matrices,
 \href{https://doi.org/10.1063/1.1704292}{\textit{J.~Math. Phys.}} \textbf{6}
 (1965), 440--449.

\bibitem{GK05}
Goldsheid I.Ya., Khoruzhenko B.A., The {T}houless formula for random
 non-{H}ermitian {J}acobi matrices, 2005, 331--346,
 \href{http://arxiv.org/abs/math-ph/0312022}{arXiv:math-ph/0312022}.

\bibitem{gradshteyn}
Gradshteyn I.S., Ryzhik I.M., Table of integrals, series, and products,
 Academic Press, San Diego, 2000.

\bibitem{GHS}
Grobe R., Haake F., Sommers H.-J., Quantum distinction of regular and chaotic
 dissipative motion,
 \href{https://doi.org/10.1103/PhysRevLett.61.1899}{\textit{Phys. Rev. Lett.}}
 \textbf{61} (1988), 1899--1902.

\bibitem{Haake}
Haake F., Quantum signatures of chaos, \textit{Springer Ser. Synergetics},
 \href{https://doi.org/10.1007/978-3-642-05428-0}{Springer}, Berlin, 2010.

\bibitem{Hamazaki-et-al}
Hamazaki R., Kawabata K., Kura N., Ueda M., Universality classes of
 non-{H}ermitian random matrices,
 \href{https://doi.org/10.1103/PhysRevResearch.2.023286}{\textit{Phys. Rev.
 Res.}} \textbf{2} (2022), 023286, 18~pages,
 \href{http://arxiv.org/abs/1904.13082}{arXiv:1904.13082}.

\bibitem{Henrici}
Henrici P., Bounds for iterates, inverses, spectral variation and fields of
 values of non-normal matrices,
 \href{https://doi.org/10.1007/BF01386294}{\textit{Numer. Math.}} \textbf{4}
 (1962), 24--40.

\bibitem{Hypergeometric-Orthogonal-Polynomials}
Koekoek R., Lesky P.A., Swarttouw R.F., Hypergeometric orthogonal polynomials
 and their {$q$}-analogues, \textit{Springer Monogr. Math.},
 \href{https://doi.org/10.1007/978-3-642-05014-5}{Springer}, Berlin, 2010.

\bibitem{Lambert}
Lambert G., Poisson statistics for {G}ibbs measures at high temperature,
 \href{https://doi.org/10.1214/20-aihp1080}{\textit{Ann. Inst. Henri
 Poincar\'e Probab. Stat.}} \textbf{57} (2021), 326--350,
 \href{http://arxiv.org/abs/1912.10261}{arXiv:1912.10261}.

\bibitem{LRR91}
Lechtenfeld O., Ray R., Ray A., Phase diagram and orthogonal polynomials in
 multiple-well matrix models,
 \href{https://doi.org/10.1142/S0217751X91002148}{\textit{Internat.~J. Modern
 Phys.~A}} \textbf{6} (1991), 4491--4515.

\bibitem{Chalker-Mehlig}
Mehlig B., Chalker J.T., Statistical properties of eigenvectors in
 non-{H}ermitian {G}aussian random matrix ensembles,
 \href{https://doi.org/10.1063/1.533302}{\textit{J.~Math. Phys.}} \textbf{41}
 (2000), 3233--3256,
 \href{http://arxiv.org/abs/cond-mat/9906279}{arXiv:cond-mat/9906279}.

\bibitem{mezzadri-taylor}
Mezzadri F., Taylor H., A matrix model of a non-{H}ermitian {$\beta
 $}-ensemble, \href{https://doi.org/10.1142/S2010326324500278}{\textit{Random
 Matrices Theory Appl.}} \textbf{14} (2025), 2450027, 24~pages,
 \href{http://arxiv.org/abs/2305.13184}{arXiv:2305.13184}.

\bibitem{NIST}
Olver F.W.J., Lozier D.W., Boisvert R.F., Clark C.W. (Editors), N{IST} handbook
 of mathematical functions, Cambridge University Press, Cambridge, 2010.

\bibitem{PS06}
Parthasarathy P.R., Sudhesh R., A formula for the coefficients of orthogonal
 polynomials from the three-term recurrence relations,
 \href{https://doi.org/10.1016/j.aml.2005.10.023}{\textit{Appl. Math. Lett.}}
 \textbf{19} (2006), 1083--1089.

\bibitem{Trefethen-Embree}
Trefethen L.N., Embree M., Spectra and pseudospectra. {T}he behavior of
 nonnormal matrices and operators, Princeton University Press, Princeton, NJ,
 2005.

\end{thebibliography}
\end{document}